\input ./style/arxiv-general.cfg
\documentclass[aos,MSNbibl,seceqn,dvips]{arximspdf}
\makeatletter
   \@ifpackageloaded{graphicx}{}{\usepackage{graphicx}}
\makeatother
\usepackage{accents}
\usepackage{algorithmic}
\usepackage{algorithm}

\doi{10.1214/15-AOS1331}
\volume{43}
\issue{5}
\pubyear{2015}
\firstpage{2132}
\lastpage{2167}
\docsubty{FLA}

\makeatletter
\renewcommand{\mathring}[1]{\accentset{\circ}{#1}}
\def\cal{\mathcal}
\newtheorem{theorem}{Theorem}[section]
\newtheorem{lemma}[theorem]{Lemma}

\newtheorem{corollary}[theorem]{Corollary}
\newcommand{\eqref}[1]{(\ref{#1})}
\newproclaim{definition}[theorem]{Definition}
\newproclaim{assumption}{Assumption}

\newproclaim{example}[theorem]{Example}
\newproclaim{remark}[theorem]{Remark}
\newcommand{\ca}[1]{{\cal #1}}
\newcommand{\Leq}{ \leq}
\newcommand{\Q}{\mathbb{Q}}
\newcommand{\R}{\mathbb{R}}
\newcommand{\g}{\gamma}
\newcommand{\D}{\Delta}
\newcommand{\e}{\varepsilon}
\newcommand{\eps}{\epsilon}
\newcommand{\z}{\zeta}
\newcommand{\vt}{\vartheta}
\newcommand{\vr}{\varrho}
\renewcommand{\k}{\kappa}
\newcommand{\lb}{\lambda}
\newcommand{\s}{\sigma}
\newcommand{\vs}{\varsigma}
\renewcommand{\t}{\tau}
\newcommand{\diam}{\operatorname{diam}}
\newcommand{\supp}{\operatorname{supp}}
\newcommand{\symdif}{\vartriangle}
\newcommand{\eins}{\mathbf{1}}
\newcommand{\comparable}{\sqsubset}
\newcommand{\persist}{\sqsubseteq}
\newcommand{\hdd}{h_{D,\delta}}
\newcommand{\hpd}{h_{P,\delta}}
\newcommand{\dthick}{\delta_{\mathrm{thick}}}
\newcommand{\cthick}{c_{\mathrm{thick}}}
\newcommand{\csepl}{\underline c_{\mathrm{sep}}}
\newcommand{\csepu}{\overline c_{\mathrm{sep}}}
\newcommand{\cflat}{c_{\mathrm{flat}}}
\newcommand{\cbound}{c_{\mathrm{bound}}}
\newcommand{\cpart}{c_{\mathrm{part}}}
\newcommand{\dpart}{\mathrm{d}} 

\begin{document}
\begin{frontmatter}

\title{Fully adaptive density-based clustering}
\runtitle{Fully adaptive density-based clustering}

\begin{aug}
\author{\fnms{Ingo} \snm{Steinwart}\thanksref{T1}\ead[label=e1]{ingo.steinwart@mathematik.uni-stuttgart.de}}
\runauthor{I. Steinwart}

\affiliation{University of Stuttgart}

\address{Institute for Stochastics and Applications\\
Faculty 8: Mathematics and Physics\\
University of Stuttgart\\
D-70569 Stuttgart\\
Germany\\
\printead{e1}}
\thankstext{T1}{Supported by DFG Grant STE 1074/2-1.}
\end{aug}

%
\received{\smonth{3} \syear{2015}}
%
\revised{\smonth{3} \syear{2015}}

%
\begin{abstract}
The clusters of
a distribution are often defined by the
connected components of a density level set.
However, this definition
depends on the user-specified level.
We address this issue by
proposing a simple, generic algorithm, which uses
an almost arbitrary level set estimator to estimate
the smallest level at which there are more than one
connected components.
In the case where this algorithm is fed
with histogram-based level set estimates,
we provide a finite sample
analysis, which is then used to show that the
algorithm consistently estimates both the
smallest level and the corresponding connected
components. We further establish rates of
convergence for the two estimation problems, and
last but not least, we present a simple,
yet adaptive strategy
for determining the width-parameter of the
involved density estimator in a data-depending
way.
\end{abstract}

%

%
\end{frontmatter}

\section{Introduction}

One definition of density-based clusters,
which was first proposed by Hartigan \cite{Hartigan75},
assumes i.i.d. data $D=(x_1,\ldots,x_n)$ generated by
some unknown distribution $P$ 
that has a continuous density $h$. 
For a user-defined threshold $\rho\geq0$, the clusters of $P$ are then
defined to be the connected components of the level set
$\{h\geq\rho\}$.
%
This so-called single level approach
has been studied by several authors; see, for example,
\cite{Hartigan75,CuFr97a,Rigollet07a,MaHeLu09a,RiWa10a} and the
references therein.
Unfortunately, however, different values of $\rho$ may lead to different
(numbers of) clusters
(see, e.g., the illustrations in \cite{ChDa10a,RiSiNuWa12a}), and
there is no generally accepted rule for choosing $\rho$, either.
In addition,
using a couple of different candidate values creates the problem of
deciding which of the resulting clusterings is best.
For this reason, Rinaldo and Wasserman
\cite{RiWa10a} note that research on data-dependent, automatic
methods for
choosing $\rho$ (and the width parameter of the involved density estimator)
``would be very useful.''

A second, density-based definition for clustering, which is known as
the cluster tree approach,
avoids this issue
by considering all levels and the corresponding connected components
simultaneously.
Its focus thus lies on the identification of
the hierarchical tree structure
of the connected components for different levels; see, for
example, \cite{Hartigan75,Stuetzle03a,ChDa10a,StNu10a,KpLu11a} for details.
For example, Chaudhuri and Dasgupta~\cite{ChDa10a} show, under some
assumptions on $h$, that a modified single linkage
algorithm recovers this tree in the sense of \cite{Hartigan81a},
and
%
Kpotufe and von Luxburg \cite{KpLu11a} obtain\vadjust{\goodbreak} similar results for
an underlying $k$-NN density estimator. In addition, Kpotufe and von
Luxburg \cite{KpLu11a} propose
a simple pruning strategy
that removes connected components that artificially occur because of
finite sample variability.
However, the notion of recovery taken from \cite{Hartigan81a}
only focuses on the correct estimation of the cluster tree structure
and not on the
estimation of the clusters itself; cf. the discussion in \cite
{Steinwart11a}.

%

Defining clusters by the connected components of one or more level sets
clearly requires us to estimate
level sets in one form or the other. Level set estimation itself is a
classical nonparametric problem,
which has been considered by various authors; see, for
example, \cite
{DeWi80a,Hartigan87a,MuSa91a,Polonik95a,BDLi97a,Tsybakov97a,BaCACu01a,BaCuJu00a,ScHuSt05a,StHuSc05a,RiVe09a,SiScNo09a}.
In these articles, two different performance measures are considered
for assessing the
quality of a density level set estimate, namely the mass of the
symmetric difference between the estimate and the
true level set, and the Hausdorff distance between these two sets.
Estimators that are consistent with respect to the Hausdorff metric
clearly capture all
topological structures eventually, so that these estimators form an almost
canonical choice for density-based clustering with fixed level $\rho$.
In contrast, level set estimators that are only consistent with respect
to the first performance measure
are, in general, not suitable for the cluster problem,
since even sets that are equal up to measure zero may have
completely different topological properties.

Another, very recent density-based cluster definition (see \cite
{ChaconXXa}) uses Morse theory to define the clusters of $P$.
The idea of this approach is best illustrated by water flowing on a
terrain. 
Namely, for each mode $x_0$ of~$h$, the corresponding modal cluster is
the set of points from which water
flows, on the steepest descent path, to $x_0$ on the terrain described
by $-h$.
Under suitable smoothness assumptions on~$h$, it turns out that these
modal clusters
form a partition of the input space modulo a Lebesgue zero set.
Unlike in the single level approach, essentially all points of the
input domain are thus assigned to a cluster.
However, the required smoothness assumptions are somewhat strong, and
so far, a consistent estimator has only been found
for the one-dimensional case; see \cite{ChaconXXa}, Theorem~1.

In this work, we consider none of these approaches. Instead, we follow
the approach of
\cite{Steinwart11a}; that is, we
are interested in estimating (a) the infimum of all $\rho$ at which
the level set has more than one component and (b)
the corresponding components. In addition, the usual continuity
assumption on $h$ is avoided.
Let us therefore briefly describe the approach of \cite{Steinwart11a} here;
more details can be found in Section~\ref{Prelim}.

Its first step consists of defining level sets $M_\rho$ that are
\emph{independent of the actual choice of the density}; see \eqref{Def-Mr}.
Here we note that this independence is crucial for avoiding ambiguities
when dealing with discontinuous densities.
So far,
some approaches have been made
to address these difficulties. 
For example, Cuevas and Fraiman \cite{CuFr97a} introduced a
thickness assumption for sets $C$ that
rules out cases in which
neighborhoods of $x\in C$ have not sufficient mass.
This thickness assumption excludes some topological pathologies such as
topologically connecting
bridges of zero mass, while others, such as cuts of measure zero, are
not addressed.
These issues are avoided in
\cite{RiWa10a}\vadjust{\goodbreak} by considering
level sets of convolutions $k*P$ of the underlying
distribution $P$ with a continuous kernel $k$ on $\mathbb{R}^d$ having a
compact support.
Since such convolutions are always continuous, these authors
cannot only deal with discontinuous densities, but also with
distributions that do not have a Lebesgue density at all.
However, different kernels or kernel widths
may lead to different level sets, and consequently, their approach
introduces new parameters that are
hard to control by the user.
In this respect, recall that
for some other functionals of densities, Donoho \cite{Donoho88a}
could remove these ambiguities,
but so far it is unclear whether this is also possible for cluster analysis.


In a second step,
%
the infimum $\rho^*$ over all levels $\rho$ for which $M_\rho$
contains more than one connected component is considered.
To reliably estimate ${\rho^*}$, it is further assumed that there exists
some $\rho^{**}>\rho^*$ such that the
component structure of $M_\rho$ remains persistent for all $\rho\in
(\rho^*,\rho^{**}]$.
Note that such persistence
is
assumed either explicitly or implicitly in basically all density-based
clustering approaches
(see, e.g., \cite{ChDa10a,KpLu11a}), as it seems intuitively
necessary for dealing with
vertically uncertainty caused by finite sample effects.
Another assumption imposed on $P$,
namely that $M_\rho$ has exactly two components between ${\rho^*}$ and
${\rho^{**}}$,
seems
%
to be more restrictive at first glance.
However, the opposite is true: if, for example, $h:[0,1]\to(0,\infty)$
is a continuous density with exactly two distinct, strict local minima
at say $x_1$ and $x_2$,
then we only have more than two connected components in a small range
above ${\rho^*}$
if $h(x_1) = h(x_2)$. Compared to the case $h(x_1)\neq h(x_2)$,
the latter seems to be rather singular, in particular, if one considers
higher-dimensional analogs.
Finally note that we could look for further splits of components above
the level ${\rho^{**}}$ in
a similar fashion. This way we would recover the cluster tree approach,
and, at least for the one-dimensional case,
also the Morse approach by some trivial modifications already discussed
in \cite{ChaconXXa}.
However, such an iterative approach is clearly out of the scope of this paper.


The first main result of this paper is a generic algorithm, which is based
on an arbitrary level set estimator,
for estimating both
${\rho^*}$ and the corresponding clusters. In the case in which the
underlying level set estimator enjoys
guarantees on its vertical and
horizontal uncertainty, we further provide an error analysis for both
estimation problems in terms of these
guarantees.
A detailed statistical analysis is then conducted for histogram-based
level set estimators.
Here, our first result is a finite sample bound, which is then used to
derive (as in \cite{Steinwart11a})
consistency. 
We further provide rates of convergence for estimating ${\rho^*}$
under an
assumption on $P$ that
describes
how fast the connected components of $M_\rho$ move apart for increasing
$\rho\in({\rho^*},{\rho^{**}}]$.
The next main result establishes rates of convergence for estimating
the clusters.
Here we additionally need
the well-known flatness condition of Polonik (see \cite
{Polonik95a}) and an assumption that describes the mass of
$\delta$-tubes around the boundaries of the $M_\rho$'s.
Unlike previous articles, however,
we do not need to restrict our considerations to (essentially)
rectifiable boundaries.
All these rates can only be achieved if the histogram width is chosen
in a suitable,
distribution-dependent way, and therefore we finally propose a simple
data-driven
parameter selection strategy. Our last main result shows that
this strategy often achieves the above rates
without knowing characteristics of $P$.

Since this work strongly builds upon \cite{Steinwart11a,SrSt12a},
let us briefly describe
our main \emph{new} contributions. First, in \cite{Steinwart11a},
only the consistency of the histogram-based algorithm is established;
that is, no rate of convergence is presented. While in \cite
{SrSt12a}, such rates are established, the situation considered
in \cite{SrSt12a} is different. Indeed, in \cite{SrSt12a}, an
algorithm that uses
a Parzen window density estimator to estimate the level sets is
considered. However, this algorithm
requires the density to be $\alpha$-H\"older continuous for known $\alpha$.
Second, neither of the papers considers a data-dependent way of
choosing the width parameter of the involved
density estimator.
Besides these new contributions, this paper also adds a substantial
amount of extra information regarding the
imposed assumptions and, last but not least, polishes many of the
results from \cite{Steinwart11a}.

The rest of this paper is organized as follows. In Section~\ref{Prelim}
we recall the cluster definition from \cite{Steinwart11a} and
%
generalize the clustering algorithm from \cite{Steinwart11a}.
In Section~\ref{histo-sec} we provide a finite-sample analysis for the
case, in which the
generic algorithm is fed with plug-in estimates of a histogram.
%
In Section~\ref{sec:cons-rates} we then 
establish consistency
and the new learning rates.
Section~\ref{sec:par-select} contains the description and the analysis
of the new data-driven
width selection strategy.
%
Proofs of some of our results that are new, compared to those in
\cite{Steinwart11a,SrSt12a},
can be found in Section~\ref{sec-select-proofs}.
The remaining proofs, auxiliary results and
an example of a large class of distributions on $\R^2$ with continuous
densities
that satisfy all the assumptions made in this paper can be found in
\cite{SteinwartXXb}.

%
%

\section{Preliminaries: Level sets, clusters and a generic
algorithm}\label{Prelim}

In this section we recall and refine several notions related to the
definition of
clusters in \cite{Steinwart11a}. In addition, we present a generic
clustering algorithm,
which is based on the ideas developed in \cite{Steinwart11a}.

Let us begin by fixing some notation and assumptions used throughout
this paper:
$(X,d)$ is always a compact metric space, and
$\ca B(X)$ denotes its Borel $\s$-algebra.
Moreover, $\mu$ is a \emph{known} $\s$-finite
measure on $\ca B(X)$,
and $P$
is an \emph{unknown} $\mu$-absolutely
continuous distribution on $\ca B(X)$ from which the data
$D=(x_1,\ldots,x_n)\in X^n$
will be drawn in
an i.i.d. fashion.
In the following, we always assume that $\mu$ has full support, that
is, $\supp\mu= X$.
Of course,
the example we are most interested in is that of $X=[0,1]^d$ and $\mu$ being
the Lebesgue measure on $X$, but alternatives such as the
surface measure on a sphere are possible, too.

Given an $A\subset X$,
we write $\mathring A$ for its interior, $\overline{A}$ for its
closure and
$\partial A:= \overline{A} \setminus\mathring A$ for its boundary.
Finally, $\eins_A$ denotes the indicator function of $A$ and $A\symdif
B$, the symmetric difference of
two sets $A$ and $B$.

\subsection{Density-independent density level sets}\label{sec-regular}

Unlike most papers dealing with density-based clustering,
we will not assume that the data-generating distribution $P$ has a
\emph{continuous} density.
Unfortunately, this generality makes it more challenging to define
density-level-based clusters. Indeed, since the data is generated by $P$,
we actually need to define clusters for distributions and not for densities.
Consequently, a well-defined
density-based notion of clusters either needs to be independent of the
choice of the density,
or pick, for each $P$, a somewhat canonical density.
Now, if we assume that each considered $P$ has a continuous density~$h$, then these $h$'s
may serve as such canonical choices. In the absence of continuous
densities, however,
it is no longer clear how a ``canonical'' choice should look.
In addition, the level sets of two different densities of the same $P$
may have very distinct connected components (see, e.g., Figure~\ref
{fig:non_unique_h})
so that defining the clusters of $P$ by the connected components of $\{
h\geq\rho\}$
becomes inconsistent. In other words, neither of the two alternatives
above is readily
available for general $P$.

\begin{figure}

\includegraphics{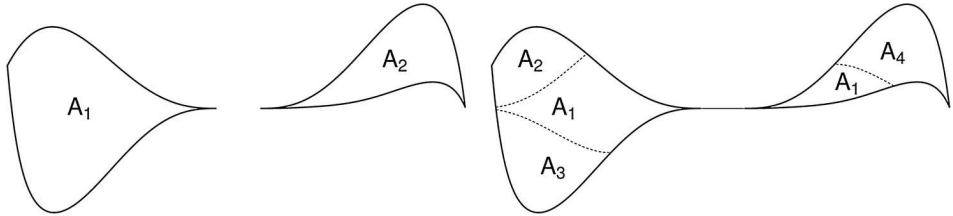}

\caption{topologically relevant changes on sets of measure zero.
Left: The thick solid lines indicate a set consisting of two
connected components $A_1$ and $A_2$.
If $h=c\eins_{A_1\cup A_2}$ is a density of $P$ for a suitable
constant $c$,
then $A_1$ and $A_2$ are the connected components of
$\{h \geq\rho\}$ for all $\rho\in[0,c]$.
Right: This is a similar situation, but with topologically relevant
changes on sets of measure zero.
The straight horizontal thin line indicates a line of measure zero
connecting the two components,
and the dashed lines indicate cuts of measure zero. Clearly, $h':= c
\eins_{A_1\cup A_2\cup A_3\cup A_4}$
is another density of~$P$, but the connected components of
$\{h' \geq\rho\}$ are the four sets $A_1,\ldots, A_4$ for all $\rho\in[0,c]$.
}
\label{fig:non_unique_h}
\end{figure}

This issue is addressed in \cite{Steinwart11a} by considering
``density level sets''
that are independent of the choice of the density.
To recall this idea from \cite{Steinwart11a}, we fix an \emph{arbitrary}
$\mu$-density $h$ of $P$. Then, for every $\rho\geq0$,
\[
\mu_\rho(A):= \mu\bigl(A\cap\{h\geq\rho\}\bigr),\qquad A\in\ca B(X)
\]
defines a $\s$-finite measure $\mu_\rho$ on $\ca B(X)$ that is actually
\emph{independent} of our choice of $h$.
As a consequence, the set
%
\begin{equation}
\label{Def-Mr} M_\rho := \supp\mu_\rho,
\end{equation}
which in \cite{Steinwart11a} is called the density level set of $P$
to the level $\rho$,
is independent of this choice, too. It is shown in \cite
{Steinwart11a} (see also
\cite{SteinwartXXb}, Lemma A.1.1) that
these sets are ordered in the usual way, that is,
$M_{\rho_2} \subset M_{\rho_1}$ whenever $\rho_1\leq\rho_2$.
Furthermore,
for any $\mu$-density $h$ of $P$, the definition immediately gives
%
\begin{equation}
\label{reg-half} \mu \bigl( \{h\geq\rho\} \setminus M_\rho \bigr) = \mu
\bigl( \{ h\geq\rho\} \cap(X\setminus M_\rho) \bigr) =
\mu_\rho(X\setminus M_\rho) = 0;
\end{equation}
that is, modulo $\mu$-zero sets, the level sets $\{h\geq\rho\}$ are not
larger than $M_\rho$.
In fact, $M_\rho$ turns out to be the smallest closed set satisfying
\eqref{reg-half},
and
it is shown in \cite{Steinwart11a} (see also
\cite{SteinwartXXb}, Lemma A.1.2) that we have both
%
\begin{equation}
\label{Mr-diff-closure} {\mathring{\{h\geq\rho\}}} \subset M_\rho\subset
\overline{\{h\geq \rho\}} \quad\mbox{and}\quad M_\rho\symdif\{h\geq\rho\}
\subset\partial\{h\geq\rho\} .
\end{equation}
For technical reasons we will not only
need \eqref{reg-half} but also the ``converse'' as well as a
modification of \eqref{reg-half}.
The exact requirements are introduced in
the following definition, which slightly deviates from \cite{Steinwart11a}.

\begin{definition}
We say that $P$ is normal at level $\rho\geq0$ if there exist two
$\mu$-densities $h_1$ and $h_2$ of $P$ such that
\[
\mu\bigl(M_\rho\setminus\{h_1\geq\rho\} \bigr) = \mu
\bigl(\{h_2>\rho\} \setminus \mathring M_\rho\bigr) = 0.
\]
Moreover, we say that $P$ is normal
if it is normal at every level.
\end{definition}

It is shown in
\cite{SteinwartXXb}, Lemma A.1.3,
that $P$ is normal if it has both an upper semi-continuous $\mu
$-density $h_1$ and a
lower semi-continuous $\mu$-density $h_2$.
Moreover, if
$P$ has a $\mu$-density $h$ such that $\mu(\partial\{h\geq\rho\})=0$,
then $P$ is
normal at level $\rho$ by \eqref{Mr-diff-closure}.
Finally, note that if the conditions of normality at level $\rho$ are
satisfied for some $\mu$-densities $h_1$ and $h_2$ of $P$, then
they are actually satisfied for all $\mu$-densities $h$ of $P$,
and we have $\mu(M_\rho\symdif\{h\geq\rho\})=0$.

The remarks made above show that most distributions one would
intuitively think of are normal.
The next lemma
demonstrates that there are also distributions
that are not normal at a continuous range of levels.
%
%

\begin{lemma}\label{not-reg-ex}
There exists a Lebesgue absolutely continuous
distribution $P$ on $[0,1]$ and a $c>0$ such that $P$ is not normal at
$\rho$ for all $\rho\in(0,c]$.
\end{lemma}

\subsection{Comparison of partitions and some notions of
connectivity}\label{subsec:cc}

Following~\cite{Steinwart11a} we will define clusters with the help
of connected components
over a range of level sets. To prepare this definition, we recall
some notions related to connectivity
in this subsection. Moreover,
we introduce a tool that makes it possible to
compare the connected components of two level sets.


To motivate the following definition, which generalizes the ideas from
\cite{Steinwart11a},
we note that the connected components of a set form a partition.

\begin{definition}\label{comp-partitions}
Let $A\subset B$ be nonempty sets and $\ca P(A)$ and $\ca P(B)$ be
partitions of $A$ and $B$, respectively.
Then $\ca P(A)$ is comparable to $\ca P(B)$, and we write $\ca P(A)
\comparable\ca P(B)$
if, for all $A'\in\ca P(A)$, there is a $B'\in\ca P(B)$
with $A'\subset B'$.
\end{definition}

Informally speaking, $\ca P(A)$ is comparable to $\ca P(B)$ if no cell
$A'\in\ca P(A)$ is broken into pieces in
$\ca P(B)$. In particular, if $\ca P_1$ and $\ca P_2$ are two
partitions of $A$, then
$\ca P_1 \comparable\ca P_2$ if and only if $\ca P_1$ is finer than
$\ca P_2$.

Let us now assume that we have two partitions $\ca P(A)$ and $\ca P(B)$
such that $\ca P(A) \comparable\ca P(B)$. Then it is easy to see
(cf. \cite{SteinwartXXb}, Lemma A.2.1)
that there exists a unique map $\z: \ca P(A) \to\ca P(B)$ such that,
for all $ A'\in\ca P(A)$, we have
\[
A' \subset\z\bigl(A'\bigr).
\]
Following \cite{Steinwart11a}, we call
$\z$ the cell relating map (CRM) between $A$ and $B$.
Moreover,
%
we write $\z_{A,B}:= \z$ when
we want to emphasize the involved pair $(A,B)$.
Note that $\z$
is injective, if and only if no two distinct cells of $\ca P(A)$
are contained in the same cell of $\ca P(B)$. Conversely, $\z$ is
surjective, if and only if every cell in
$\ca P(B)$ contains a cell of $\ca P(A)$.
Therefore, $\z$ is bijective, if and only if there is a structure
preserving a
one-to-one relation between the cells of the two partitions.
In this case, we say that $\ca P(A)$ is \emph{persistent} in $\ca P(B)$
and write $\ca P(A) \persist\ca P(B)$.

The next lemma establishes
a very useful composition formula for CRMs.
For a proof, which is again inspired by
\cite{Steinwart11a}, we refer to
\cite{SteinwartXXb}, Section A.2.

\begin{lemma}\label{crm-arithm}
Let $A\subset B\subset C$ be nonempty sets with partitions $\ca P(A)$,
$\ca P(B)$ and $\ca P(C)$ such that
$\ca P(A)\comparable\ca P(B)$ and $\ca P(B)\comparable\ca P(C)$.
Then we have $\ca P(A)\comparable\ca P(C)$, and the corresponding CRMs satisfy
\[
\z_{A,C} = \z_{B,C} \circ\z_{A,B}.
\]
\end{lemma}

The lemma above shows that the relations $\comparable$ and $\persist$
are transitive. 
Moreover, if $\ca P(A) \persist\ca P(C)$, then
$\z_{A,B}$ must be injective, and $\z_{B,C}$ must be surjective, and
we have 
$\ca P(A) \persist\ca P(B)$ if and only if $\ca P(B) \persist\ca P(C)$.

Now recall 
%
that an $A\subset X$ is (topologically) connected if, for every pair
$A',A''\subset A$
of relatively closed disjoint subsets of $A$ with $A'\cup A'' = A$,
we have $A'=\varnothing$ or $A''=\varnothing$.
The maximal connected subsets of $A$ are called the connected
components of $A$.
It is well known that these
components form a partition of $A$, which we denote by $\ca C(A)$.
Moreover, for closed $A\subset B$ with
$|\ca C(B)|<\infty$ we have $\ca C(A) \comparable\ca C(B)$;
see
\cite{Steinwart11a} or
\cite{SteinwartXXb}, Lemma A.2.3.


Following \cite{Steinwart11a}, we will also
consider a discrete version of path-connectivity.
To recall the latter, we fix a $\t>0$ and an $A\subset X$.
Then $x,x'\in A$ are $\t$-connected in $A$
if
there exist
$x_1,\ldots,x_n\in A$ such that $x_1=x$, $x_n=x'$ and $d(x_i,x_{i+1}) <
\t$ for all $i=1,\ldots,n-1$.
Clearly,
being $\t$-connected gives an equivalence relation on $A$.
We write $\ca C_\t(A)$ for the resulting partition and call its cells
the \emph{$\t$-connected components of $A$}.
It is shown in \cite{Steinwart11a} (see also \cite{SteinwartXXb}, Lemma
A.2.7)
that $\ca C_\t(A) \comparable\ca C_\t(B)$ for all $A\subset B$ and
$\t>0$.

For a closed $A$ and $\t>0$, we have $\ca C(A) \comparable\ca C_\t(A)$
with a surjective CRM $\z:\ca C(A)\to\ca C_\t(A)$; see \cite
{Steinwart11a} or \cite{SteinwartXXb}, Proposition A.2.10.
To characterize, when this CRM is even bijective, let us assume that
$1<|\ca C(A)|<\infty$. Then
%
\begin{equation}
\label{top-connect-new2-hxx} \t^*_A := \min \bigl\{d\bigl(A',A''
\bigr): A',A''\in\ca C(A) \mbox{ with }
A'\neq A'' \bigr\}
\end{equation}
denotes the minimal distance between mutually different components of
$\ca C(A)$.
Now it is shown in \cite{Steinwart11a} (or
\cite{SteinwartXXb}, Proposition A.2.10) that
\[
\ca C(A) = \ca C_\t(A) \quad\Longleftrightarrow\quad\t\in\bigl(0,
\t_A^*\bigr] ;
\]
see also Figure~\ref{fig:cc_vs_tcc} for an illustration. In other
words, $\t_A^*$ is the largest (horizontal) granularity $\t$
at which the connected components of $A$ are not glued together.
Finally, this threshold is ordered for closed $A\subset B$ in the sense that
$\t_A^* \geq\t_B^*$ whenever
$|\ca C(A)|<\infty$, $|\ca C(B)|<\infty$, and
the CRM $\z: \ca C(A) \to\ca C(B)$ is injective. We refer to
\cite{Steinwart11a} or
\cite{SteinwartXXb}, Lemma A.2.11.
\begin{figure}

\includegraphics{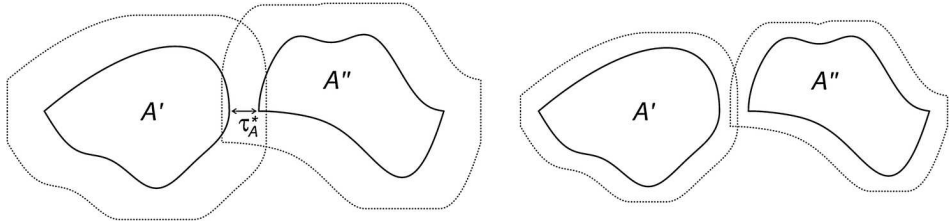}

\caption{The role of $\t_A^*$. Left: A set $A$ consisting of
two connected components $A'$ and $A''$ drawn in solid lines.
The dotted lines indicate the contours
of the set of all points that are within $\t$-distance of $A'$,
respectively $A''$,
for some fixed $\t>\t_A^*$ and the sup-norm. Since there are some
elements in $A''$ that are within $\t$-distance
of $A'$, there is only one $\t$-connected component, namely $A$. The
CRM $\z:\ca C(A)\to\ca C_\t(A)$ is thus surjective
but not injective.
Right: Here we have the same situation for some $\t< \t_A^*$. In
this case, $A'$ and $A''$ are also the $\t$-connected
components of $A$, and the CRM $\z:\ca C(A)\to\ca C_\t(A)$ is bijective.}
\label{fig:cc_vs_tcc}
\end{figure}



\subsection{Clusters}

Using the concepts developed in the previous subsections, we can now
recall the definition of clusters from \cite{Steinwart11a}.

\begin{definition}\label{top-clust-def1-neu}
The distribution
$P$ can be clustered between
$\rho^*\geq0$ and $\rho^{**}>\rho^*$ if $P$ is normal and
for all $\rho\in[0,\rho^{**}]$,
the following three
conditions are satisfied:
\begin{longlist}[(iii)]
\item[(i)] we have either $|\ca C(M_\rho)|= 1$ or $|\ca C(M_\rho)|= 2$;
\item[(ii)] if we have $|\ca C(M_\rho)|=1$, then $\rho\leq\rho^*$;
\item[(iii)] if we have $|\ca C(M_\rho)|=2$, then $\rho\geq\rho^*$
and $\ca
C (M_{\rho^{**}})\persist\ca C(M_\rho)$.
\end{longlist}
Using the CRMs $\z_\rho:\ca C (M_{\rho^{**}})\to\ca C(M_\rho)$, we
then define
the clusters of $P$ by
\[
A^{*}_i := \bigcup_{\rho\in(\rho^*,\rho^{**}]}
\z_\rho(A_i) ,\qquad i\in\{ 1,2\},
\]
where
$A_1$ and $A_2$ are the two topologically connected components of
$M_{\rho^{**}}$.
\end{definition}

By conditions (iii) and (ii), we find
$\rho<{\rho^*}\Rightarrow|\ca C(M_\rho)| = 1 \Rightarrow\rho\leq
{\rho^*}$
as well as
$\rho>{\rho^*}\Rightarrow|\ca C(M_\rho)| = 2 \Rightarrow\rho\geq
{\rho^*}$
for all $\rho\in[0,{\rho^{**}}]$; see also Figure~\ref{fig:cluster}.
At each level below ${\rho^*}$ there is thus only one component, while
there are
two components at all levels in between ${\rho^*}$ and ${\rho^{**}}$.
Moreover, in
both cases the corresponding partitions
are persistent.

%


%

Since all $\z_\rho$'s are bijective, we find $\z_\rho(A_1) \cap\z
_\rho(A_2)=
\varnothing$ for all
$\rho\in(\rho^*,\rho^{**}]$, and using $\z_\rho(A_1)\nearrow
A^*_i$ for $\rho
\searrow{\rho^*}$,
we conclude that
$A_1^*\cap A_2^* = \varnothing$. In general, the sets $A^*_i$ are
neither open nor closed, and
we may have $d(A_1^*,A_2^*)=0$; that is, the clusters may touch each
other; see again
Figure~\ref{fig:cluster}.

\begin{figure}

\includegraphics{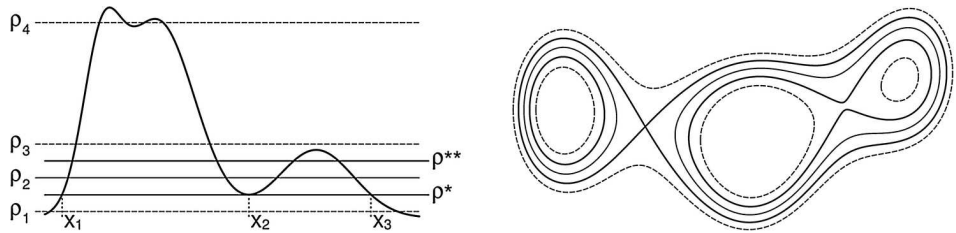}

\caption{Definition of clusters. Left: A 1-dimensional mixture
of three Gaussians
together with the level ${\rho^*}$ and a possible choice for ${\rho^{**}}$.
The
component structure at level $\rho_2\in({\rho^*}, {\rho^{**}})$ coincides
with that at level ${\rho^{**}}$, while for $\rho_1<{\rho^*}$, we
only have one
connected component. The levels $\rho_3, \rho_4> {\rho^{**}}$
are not considered by Definition \protect\ref{top-clust-def1-neu},
and thus the
component structure at these levels is arbitrary.
Finally, the clusters of the distribution are the open intervals
$(x_1, x_2)$ and $(x_2, x_3)$.
Right: Here we have a similar situation for a mixture of three
2-dimensional Gaussians drawn by contour lines.
The thick solid lines again indicate the levels ${\rho^*}$
and ${\rho^{**}}$, and the thin solid lines show a level $\rho\in
({\rho^*},
{\rho^{**}})$.
The dashed lines correspond to a level $\rho<{\rho^*}$ and a level
$\rho
>{\rho^{**}}$.
This time the clusters are the two connected components of the open set
that is surrounded by the outer thick solid line.}
\label{fig:cluster}
\end{figure}

\subsection{Cluster persistence under horizontal uncertainty}

In general, we can only expect nonparametric estimates
of $M_\rho$ that are both vertically and
horizontally uncertain. To some extent the vertical uncertainty, which is
caused by the estimation error,
has
already been addressed by the persistence assumed in our cluster definition.
In this subsection, we complement this by recalling
tools from \cite{Steinwart11a} for dealing with horizontal
uncertainty, which is usually caused by the approximation error.

To quantify horizontal uncertainty, we need for
$A\subset X$,
$\delta>0$, the sets
\begin{eqnarray*}
A^{+\delta}&:=& \bigl\{x\in X: d(x,A) \leq\delta\bigr\},
\\
A^{-\delta}&:=& X\setminus(X\setminus A)^{+\delta},
\end{eqnarray*}
where $d(x,A) := \inf_{x'\in A}d(x,x')$ denotes the distance between
$x$ and $A$.
Simply speaking, adding a $\delta$-tube to $A$ gives $A^{+\delta}$,
while removing
a $\delta$-tube gives $A^{-\delta}$. These operations, as well as
closely related
operations based on
the Minkowski addition and difference have already been used in the
literature on level set estimation; see, for example, \cite{Walther97a}.
Some simple properties of these operations can be found in
\cite{SteinwartXXb}, Lemma A.3.1.

Now let $L_\rho$ be an estimate of $M_\rho$ having vertical and
horizontal uncertainty in the sense of
\[
M_{\rho+\e}^{-\delta}\subset L_\rho\subset
M_{\rho-\e}^{+\delta},
\]
for some $\e,\delta>0$.
Ideally, we additionally have $\ca C(M_{\rho+\e}^{-\delta}) \persist
\ca
C(L_\rho
) \persist\ca C(M_{\rho-\e}^{+\delta})$.
To reliably use $\ca C(L_\rho)$ as an estimate of $\ca C(M_\rho)$, it then
suffices to know
$\ca C(M_{\rho+\e}^{-\delta}) \persist\ca C(M_\rho) \persist\ca
C(M_{\rho-\e}^{+\delta})$.
Unfortunately, however, the latter is typically not true. Indeed,
even in the absence of horizontal uncertainty, we do not have
$\ca C(M_{\rho+\e}) \persist\ca C(M_{\rho-\e})$ if $\rho+\e
>{\rho^*}
$ and $\rho-\e
<{\rho^*}$.
Moreover, in the absence of vertical uncertainty, we usually do not have
$\ca C(M_{\rho}^{-\delta}) \persist\ca C(M_\rho)\persist\ca
C(M_{\rho
}^{+\delta})$,
either, as
components of $\ca C(M_\rho)$ may be glued together in $\ca C(M_{\rho
}^{+\delta})$
or cut apart in $\ca C(M_{\rho}^{-\delta})$; see Figure~\ref{fig:pde_vs_mde}.
To repair such cuts, our algorithm will consider
$\t$-connected components instead of connected components.
%
In the rest of
this section we thus investigate under which
conditions we do have $\ca C_\t(M_{\rho+\e}^{-\delta}) \persist\ca
C(M_\rho)
\persist\ca C_\t(M_{\rho-\e}^{+\delta})$.
We begin with the following definition taken from \cite
{Steinwart11a} that excludes
bridges and cusps that are too thin.

\begin{definition}\label{def-reg}
We say that $P$
has thick level sets of order $\g\in(0,1]$ up to the level ${\rho^{**}}>0$,
if there
exist constants $\cthick\geq1$ and $\dthick\in(0,1]$ such that, for
all $\delta\in(0, \dthick]$ and $\rho\in[0,\rho^{**}]$, we have
%
\begin{equation}
\label{def-thick-h0} \sup_{x\in M_\rho} d\bigl(x,M_\rho^{-\delta}
\bigr) \leq\cthick\delta^\g.
\end{equation}
In this case, we call
$\psi(\delta) := 3\cthick\delta^\g$
the thickness function of $P$.
\end{definition}

Thickness assumptions have been widely used in the literature
on level set estimation (see, e.g., \cite{SiScNo09a}), where the case
$\g=1$ is considered.
To some extent, the latter is a natural choice, as is discussed in
detail in \cite{SteinwartXXb}, Section A.3.
In particular, for $d=1$ we always have $\g=1$, and for $d=2$ \cite{SteinwartXXb},
Example B.2.1,
provides a rich class of continuous densities with $\g=1$.
Figure~\ref{fig:diff_psi} illustrates how different
shapes of level sets lead to different $\g$'s.

\begin{figure}

\includegraphics{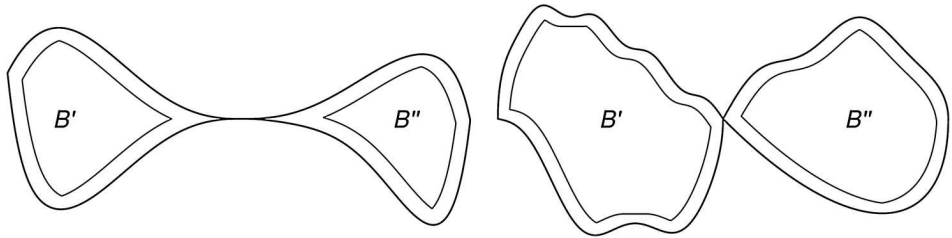}

\caption{Thick level sets. Left: The thick solid line
indicates a level set $M_\rho$ below or at the level ${\rho^*}$,
and the thin solid lines show the two components $B'$ and $B''$ of
$M^{-\delta}_\rho$. Because of the quadratic shape of $M_\rho$
around the thin
bridge, the set $M_\rho$
has thickness of order $\g=1/2$.
Right: Here we have the same situation for a distribution that has
thick level sets of order $\g=1$.
Note that the smaller $\g$ on the left leads to a significantly wider
separation
of $B'$ and $B''$ than on the right, which in turn requires larger $\t$
to glue the parts together.}
\label{fig:diff_psi}
\end{figure}

The following result, which
summarizes some findings from \cite{Steinwart11a} (see also \cite{SteinwartXXb},
Theorems A.4.2 and A.4.4),
provides an answer to our persistence question.

%

\begin{theorem}\label{reg-cluster-thm-combined}
Assume that $P$ can be clustered
between $\rho^*$ and $\rho^{**}$ and that it has thick level sets of order
$\g$ up to ${\rho^{**}}$.
Let $\psi$ be its thickness function.
Using~\eqref{top-connect-new2-hxx}, we define the function ${\t^*}
:(0,{\rho^{**}}
-{\rho^*}]\to(0,\infty)$ by
%
\begin{equation}
\label{reg-cluster-lem-def-tds-new} {\t^*}(\e) := \tfrac{1}3\t^*_{M_{\rho^*+\e}}.
\end{equation}
Then ${\t^*}$ is increasing, and
for all $\e^*\in(0,{\rho^{**}}-{\rho^*}]$, $\delta\in(0, \dthick
]$, $\t\in
(\psi(\delta),\break
{\t^*}(\e^*)]$ and
all $\rho\in[0, \rho^{**}]$,
the following statements hold:
\begin{longlist}[(ii)]
\item[(i)] we have $1\leq|\ca C_\t(M_\rho^{+\delta})|\leq2$ and
$1\leq
|\ca C_\t
(M_\rho^{-\delta})|\leq2$;
\item[(ii)] if $\rho< {\rho^*}$ or $\rho\geq{\rho^*}+ \e^*$,
then we have
\[
\ca C_\t\bigl(M_\rho^{-\delta}\bigr) \persist\ca
C(M_\rho) = \ca C_\t (M_\rho) \persist\ca
C_\t\bigl(M_\rho^{+\delta}\bigr).
\]
\end{longlist}
\end{theorem}

Theorem~\ref{reg-cluster-thm-combined} in particular shows that for
sufficiently small $\delta$ and $\t$,
the component structure of $M_\rho$ is not changed when
$\delta$-tubes are added or removed and $\t$-connected components are
considered instead.
Not surprisingly, however, the meaning of ``sufficiently small,''
which is expressed by the functions ${\t^*}$ and $\psi$, changes
when we approach the level ${\rho^*}$ from above.
Moreover, note that
even
for sufficiently small $\delta$ and $\t$,
Theorem~\ref{reg-cluster-thm-combined} does not specify
the structure of $\ca C_\t(M_\rho^{-\delta})$ and $\ca C_\t(M_\rho
^{+\delta})$
at the levels $\rho\in[{\rho^*}, {\rho^*}+\e^*)$. In fact, for
such $\rho$, the
components of $M_\rho$
may be accidentally glued together in $\ca C_\t(M_\rho^{+\delta})$;
see, for
example, Figure~\ref{fig:pde_vs_mde}.
This effect complicates our
analysis significantly.

\begin{figure}

\includegraphics{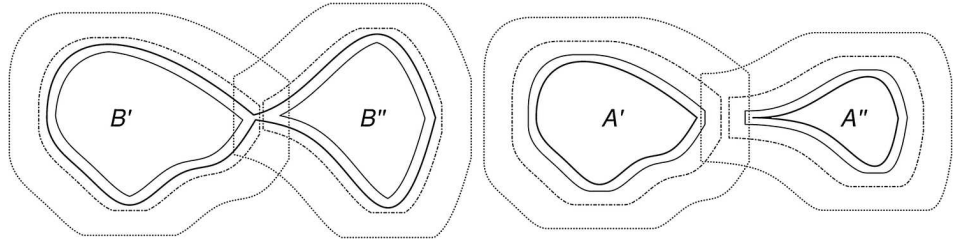}

\caption{Difficulties around ${\rho^*}$. Left: The thick solid line
indicates an $M_\rho$ for $\rho<{\rho^*}$,
and the thin solid lines show $M^{-\delta}_\rho$. While $M_\rho$ consists
of one
connected component,
$M^{-\delta}_\rho$ has two such components, $B'$ and $B''$, and hence
$\ca
C(M^{-\delta}_\rho)$ is not persistent in $\ca C(M_\rho)$.
The two types of dotted lines indicate
the set of all points that are within $\t$-distance of $B'$,
respectively $B''$ for two
values of $\t$. Only for the larger $\t$ we have
$\ca C_\t(M^{-\delta}_\rho)\persist\ca C(M_\rho)$;
%
that is, in this case $\t$-connectivity does
glue the separated regions together.
Right: The thick solid lines indicate an $M_\rho$ for some $\rho
\in({\rho^*}
, {\rho^{**}}]$
having two connected components, $A'$ and $A''$,
and thin solid lines show the two components of $M^{+\delta}_\rho$.
The two types of dotted lines indicate
the set of all points that are within $\t$-distance of $(A')^{+\delta}$,
respectively $(A'')^{+\delta}$ for the two
values of $\t$ used left. This time, we have
$\ca C(M_\rho) \persist\ca C_\t(M^{+\delta}_\rho)$ only for the smaller
value of
$\t$.
%
Together, these graphics thus illustrate that good values for $\delta$ and
$\t$ at one level may
be bad at a different level. However,
Theorem \protect\ref{reg-cluster-thm-combined} shows that this undesired
behavior can be excluded with the help of the functions $\t^*$ and
$\psi
$ for all levels $\rho\notin[{\rho^*},{\rho^*}+\e^*)$.}

\label{fig:pde_vs_mde}
\end{figure}

Let us now summarize the assumptions that will be used in the following.

{\renewcommand{\theassumption}{C}
\begin{assumption}\label{assc}
We have a
compact metric space $(X,d)$, a finite Borel measure $\mu$ on $X$ with
$\supp\mu=X$
and
a $\mu$-absolutely continuous distribution $P$
that can be clustered between $\rho^*$ and $\rho^{**}$.
In addition, $P$ has
thick level sets of order $\g\in(0,1]$
up to the level ${\rho^{**}}$. We denote
the corresponding thickness function by $\psi$ and write ${\t^*}$
for the function defined in \eqref{reg-cluster-lem-def-tds-new}.
\end{assumption}}

\subsection{A generic clustering algorithm and its analysis}\label
{subsec:gen-algo}

In this section, we present and analyze a generic version of the clustering
algorithm from \cite{Steinwart11a}. The main difference between our
algorithm and the algorithm of \cite{Steinwart11a}
is that
our generic algorithm can use any
level set estimator
that has control over both its vertical and horizontal uncertainty.

Our first result, which is a generic version of \cite{Steinwart11a}, Theorem~24,
relates the component structure of a family of level set estimates to
the component structure of certain sets $M_{\rho+\e}^{-\delta}$.
For a proof we refer to \cite{SteinwartXXb}, Section A.6.

\begin{theorem}\label{main-thick-hfr}
Let Assumption~\textup{\ref{assc}} be satisfied. Furthermore, let
$\e^*\in(0,{\rho^{**}}- {\rho^*}]$,
$\delta\in(0, \dthick]$,
$\t\in(\psi(\delta),{\t^*}(\e^*)]$ and $\e\in(0, \e^*]$.
In addition, let $(L_\rho)_{\rho\geq0}$ be a decreasing family of sets
$L_\rho\subset X$ such that
%
\begin{equation}
\label{generic-inclus} M_{\rho+\e}^{-\delta}\subset L_\rho\subset
M_{\rho-\e}^{+\delta}
\end{equation}
holds for all $\rho\geq0$. Then, for all
$\rho\in[0,{\rho^{**}}-3\e]$ and the corresponding CRMs $\z:\ca
C_\t
(M_{\rho+\e
}^{-\delta})\to\ca C_\t(L_\rho)$,
the following disjoint union holds:
%
\begin{equation}
\label{main-thick-hfr-result} \ca C_\t(L_\rho) = \z \bigl(\ca
C_\t\bigl(M_{\rho+\e}^{-\delta}\bigr) \bigr) \cup \bigl
\{B'\in\ca C_\t(L_\rho): B'\cap
L_{\rho+2\e} = \varnothing \bigr\}.
\end{equation}
\end{theorem}


Theorem~\ref{main-thick-hfr} shows that for suitable $\delta$, $\e$
and $\t$,
all $\t$-connected components $B'$ of
$L_\rho$
are \emph{either} contained in the image $\z(\ca C_\t(M_{\rho+\e
}^{-\delta}))$
\emph{or} vanish at level ${\rho+2\e}$,
that is, $B'\cap L_{\rho+2\e} = \varnothing$.
Now assume we can detect the latter components.
By Theorem~\ref{main-thick-hfr} we can then
identify the $\t$-connected components $B'$ that are contained in $\z
(\ca C_\t(M_{\rho+\e}^{-\delta}))$, and if,
in addition, $\z$ is injective, these identified components have the
same structure as $\ca C_\t(M_{\rho+\e}^{-\delta})$.
By Theorem~\ref{reg-cluster-thm-combined} we can further hope that
$\ca C_\t(M_{\rho+\e}^{-\delta}) \persist\ca C(M_{\rho+\e})$, so that
we can
relate the identified components
to those of $\ca C(M_{\rho+\e})$.
Assuming these steps can be carried out precisely, we obtain Algorithm~\ref{cluster-algo}; see also Figure~\ref{fig:cluster-prune},
which scans through the values of $\rho$ from small to large and stops as
soon as it identifies
either no component or at least two.

\begin{figure}

\includegraphics{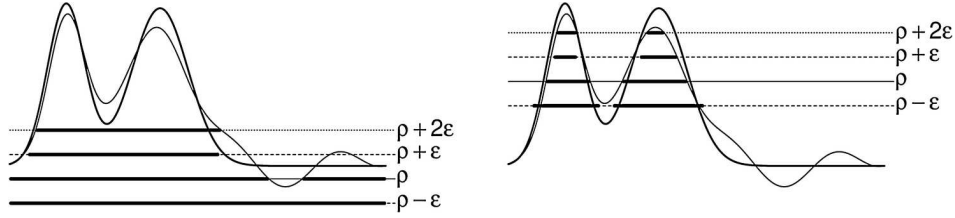}

\caption{Illustration of Algorithm \protect\ref{cluster-algo}
Left: A
density (thick solid line) having two modes on the left and a flat part
on the right.
A plug-in approach based on a density estimate (thin solid line with
three modes)
is used to provide the level set estimator $L_\rho$ (bold horizontal line
at level $\rho$), which satisfies
$M_{\rho+\e}^{-\delta}\subset L_\rho\subset M_{\rho-\e}^{+\delta}$.
Only the left component of
$L_\rho$ does not vanish at $\rho+2\e$, and thus Algorithm~\protect
\ref
{cluster-algo} identifies only one component at its line 3.
Right: Here we have the same situation at a higher level. This time
both components of $L_\rho$ do not vanish at $\rho+2\e$,
and hence Algorithm \protect\ref{cluster-algo} identifies two
components at its
line 3.
}
\label{fig:cluster-prune}
\end{figure}

The following theorem provides bounds for the level $\rho_D^*$ and the
components $B_i(D)$
returned by Algorithm \ref{cluster-algo}.
It extends the analysis from \cite{Steinwart11a}.

\begin{algorithm}[b]
\caption{Clustering with the help of a generic level set estimator}
\label{cluster-algo}
\begin{algorithmic}[1]
\REQUIRE{ Some $\t>0$ and $\e>0$.\\
A decreasing family
$(L_{D,\rho})_{\rho\geq0}$ of
subsets of $X$.}
\ENSURE An estimate of ${\rho^*}$ and the clusters $A^{*}_1$ and $A^{*}_2$.
\STATE$\rho\gets0$
\REPEAT
\STATE{Identify the $\t$-connected components $B_1',\ldots,B_M'$ of
$L_{D,\rho}$
satisfying
\[
B_i'\cap L_{D,\rho+2\e}\neq\varnothing.
\]
}
\STATE$\rho\gets\rho+\e$
\UNTIL{$M \neq1$}
\STATE$\rho\gets\rho+2\e$
\STATE{Identify the $\t$-connected components $B_1',\ldots,B_M'$ of
$L_{D,\rho}$
satisfying
\[
B_i'\cap L_{D,\rho+2\e} \neq\varnothing.
\]
}
\STATE{\textbf{return} ${\rho^*_D}:= \rho$ and the sets $B_i(D) :=
B'_i$ for
$i=1,\ldots,M$.}
\end{algorithmic}
\end{algorithm}

\begin{theorem}\label{analysis-main-combined-new}
Let Assumption~\textup{\ref{assc}} be satisfied. Furthermore, let
$\e^*\leq({\rho^{**}}- {\rho^*})/9$,
$\delta\in(0, \dthick]$,
$\t\in(\psi(\delta),{\t^*}(\e^*)]$ and $\e\in(0, \e^*]$.
In addition, let $D$ be a data set and $(L_{D,\rho})_{\rho\geq0}$ be a
decreasing family satisfying \eqref{generic-inclus}
for all $\rho\geq0$.
Then
the following statements are true for Algorithm~\ref{cluster-algo}:
\begin{longlist}[(ii)]
\item[(i)] the returned level ${\rho^*_D}$ satisfies both ${\rho
^*_D}\in[{\rho^*}
+2\e
, {\rho^*}+\e^*+5\e]$
and
%
\begin{equation}
\label{bound-rds-second} \t-\psi(\delta) < 3\t^* \bigl({\rho^*_D}-{\rho^*}+\e
\bigr) ;
\end{equation}
\item[(ii)] algorithm \ref{cluster-algo} returns two sets $B_1(D)$ and
$B_2(D)$, and these sets
can be ordered such that we have
%
\begin{eqnarray}\label{cluster-chunk-generic-bound}
\sum_{i=1}^2 \mu
\bigl(B_i(D) \symdif A_i^* \bigr) & \leq& 2\sum
_{i=1}^2 \mu \bigl(A_i^* \setminus
\bigl(A^i_{\rho^*_D+\e
}\bigr)^{-\delta} \bigr)
\nonumber
\\[-8pt]
\\[-8pt]
\nonumber
&& {}+ \mu \bigl( M_{{\rho^*_D}-\e}^{+\delta}
\setminus\bigl\{ h>{\rho^*}\bigr\} \bigr).
\end{eqnarray}
Here, $A^i_{\rho^*_D+\e}\in\ca C(M_{{\rho^*_D+\e}})$ are ordered
in the sense of
$A^i_{\rho^*_D+\e}\subset A_i^*$.
\end{longlist}
\end{theorem}

\section{Finite sample analysis of a histogram-based algorithm}\label
{histo-sec}

In this section, we 
consider the case where the level set estimates $L_{D,\rho}$ fed into
Algorithm \ref{cluster-algo}
are produced by a histogram.
The main result in this section shows
that the error estimates of Theorem~\ref{analysis-main-combined-new}
hold with high probability.

To ensure \eqref{generic-inclus}, we will use, as in \cite{Steinwart11a},
partitions that
are geometrically well behaved.
To this end, recall that
the diameter of an $A\subset X$ is
\[
\diam A:= \sup \bigl\{d\bigl(x,x'\bigr): x,x'\in A
\bigr\}.
\]
Now, the assumptions made on the used partitions are as follows:


{\renewcommand{\theassumption}{A}
\begin{assumption}\label{assa}
For each $\delta\in(0,1]$, $\ca A_\delta= (A_1,\ldots,A_{m_\delta
})$ is a
partition of $X$. Moreover,
there exist constants
$\dpart>0$ and
$\cpart\geq1$
such that, for all $\delta\in(0,1]$ and $i=1,\ldots,m_\delta$, we have
\[
\diam A_i \leq\delta,\qquad m_\delta\leq\cpart
\delta^{-\dpart} \quad\mbox{and}\quad \mu(A_i) \geq\cpart^{-1}
\delta^{\dpart}.
\]
\end{assumption}}

The most important examples of families of partitions 
satisfying Assumption~\ref{assa}
are hyper-cube partitions of $X\subset\mathbb{R}^d$
in combination with the Lebesgue measure; see \cite{SteinwartXXb}, Example
A.7.1, for details.
Other situations in which partitions satisfying Assumption~\ref{assa} can be
found include spheres
$X:=\mathbb S^{d}\subset\R^{d+1}$ together with their surface measures
and $\dpart= d-1$,
sufficiently compact metric groups in combination their Haar measure and
\emph{known}, sufficiently smooth
$\dpart$-dimensional sub-manifolds equipped their surface measure.
For details we refer to \cite{SteinwartXXb}, Lemma A7.2 and Corollary A.7.3.

Let us now assume that
Assumption~\ref{assa} is satisfied.
Moreover, for a data set $D= (x_1,\ldots,x_n)\in X^n$ we denote, in a
slight abuse of
notation,
the corresponding empirical measure by $D$, that is, $D:= \frac{1} n
\sum_{i=1}^n \delta_{x_i}$,
where $\delta_{x}$ is the Dirac measure at $x$.
Then the resulting histogram is
%
\begin{equation}
\label{histogram-def} h_{D, \delta}(x) = \sum_{j=1}^{m_\delta}
\frac{D(A_j)} {\mu(A_j)} \cdot\eins _{A_j}(x),\qquad x\in X.
\end{equation}
%
The following theorem provides a finite sample analysis for
using the plug-in estimates $L_{D,\rho} := \{h_{D, \delta} \geq\rho
\}$
in
Algorithm \ref{cluster-algo}.

\begin{theorem}\label{analysis-main2-new}
Let Assumptions \textup{\ref{assa}} and \textup{\ref{assc}} be satisfied.
For a fixed $\delta\in(0,\dthick]$, $\vs\geq1$, $n\geq1$ and
$\t> \psi(\delta)$, we fix an $\e>0$ satisfying the bound
%
\begin{equation}
\label{analysis-main2-h0} \e\geq\cpart
\sqrt{\frac{E_{\vs, \delta}}{2\delta^{2\dpart}n}},
\end{equation}
where $E_{\vs, \delta}:= \vs+\ln(2\cpart) -\dpart\ln\delta$,
or
if $P$ has a bounded $\mu$-density $h$, the bound
%
\begin{equation}
\label{analysis-main2-alter} \e\geq\sqrt{\frac{2\cpart(1 + \Vert h \Vert_\infty)E_{\vs,
\delta}}{\delta
^{\dpart}n}} + \frac{2\cpart E_{\vs, \delta}}{3\delta^{\dpart}n}.
\end{equation}
We further pick an $\e^*>0$ satisfying
%
\begin{equation}
\label{estar} \e^* \geq\e+ \inf \bigl\{\e'\in\bigl(0,{
\rho^{**}}-{\rho^*}\bigr] : \t ^*\bigl(\e'\bigr) \geq \t \bigr\}.
\end{equation}
For each data set $D\in X^n$, we now feed Algorithm \ref{cluster-algo}
with the parameters $\t$ and $\e$,
and with the family $(L_{D,\rho})_{\rho\geq0}$ given by
\[
L_{D,\rho} := \{h_{D, \delta} \geq\rho\},\qquad \rho\geq0.
\]
If $\e^* \leq({\rho^{**}}-{\rho^*})/9$,
then with probability $P^n$ not less than $1-e^{-\vs}$, we have a
$D\in X^n$ satisfying the assumptions and conclusions of Theorem~\ref
{analysis-main-combined-new}.
\end{theorem}

At this point we like to emphasize that a finite sample bound in the
form of Theorem~\ref{analysis-main2-new} can be derived from our analysis whenever
Algorithm \ref{cluster-algo} uses
a density level set estimator guaranteeing the inclusions
$M_{\rho+\e}^{-\delta}\subset L_{D,\rho} \subset M_{\rho-\e
}^{+\delta}$ with
high probability.
A possible example of such an alternative level set estimator is a
plug-in approach based
on a moving window density estimator, since for the latter it is
possible to establish a
uniform convergence result similar to
\cite{SteinwartXXb}, Theorem~A.8.1; see, for example, \cite
{SrSt12a,GiGu02a}.
Unfortunately, the resulting level sets become \emph{computationally
unfeasible} when used na\"ively, and hence
we have not included this approach here.
It is, however, an interesting open question, whether sets
$L_{D,\rho}$ that are constructed differently from the moving window estimator
can address this issue. So far, the only known result in this direction
\cite{SrSt12a} constructs such sets for $\alpha$-H\"older-continuous
densities $h$
with \emph{known} $\alpha$, but we conjecture that a similar construction may
be possible for general $h$, too.
In addition, strategies such as approximating the sets $L_{D,\rho}$ by
fine grids may be feasible, at least for
small dimensions, too.

\section{Consistency and rates}\label{sec:cons-rates}

The first goal of this section is to use
the finite sample bound of Theorem~\ref{analysis-main2-new}
to show that Algorithm \ref{cluster-algo} estimates both ${\rho^*}$
and the
clusters $A_i^*$ consistently.
We then introduce some assumptions on $P$ that lead to convergence
rates for both estimation problems.

The following consistency result is a modification of~\cite{Steinwart11a}, Theorem~26;
see also~\cite{SteinwartXXb}, Section A.9, for a corresponding
modification of its proof.

\begin{theorem}\label{main}
Let Assumptions \textup{\ref{assa}} and \textup{\ref{assc}} be satisfied, and let
$(\e_n)$, $(\delta_n)$ and $(\t_n)$ be strictly positive sequences
converging to zero
such that $\psi(\delta_n) < \t_n$ for all sufficiently large $n$, and
%
\begin{equation}
\label{cons-assump} \lim_{n\to\infty} \frac{ \ln\delta_n^{-1}} {n \delta_n^{2\dpart
}\e_n^2} = 0.
\end{equation}
For $n\geq1$, consider Algorithm \ref{cluster-algo}
with the input parameters $\e_n$, $\t_n$ and
the family $(L_{D,\rho})_{\rho\geq0}$ given by $ L_{D,\rho} := \{
h_{D, \delta
_n} \geq\rho\}$.
Then, for all $\eps>0$, we have
\[
\lim_{n\to\infty}P^n \bigl( \bigl\{ D\in
X^n: 0< {\rho^*_D}- {\rho ^*}\leq \eps \bigr\} \bigr) =
1,
\]
and if $\mu(\overline{A_i^* \cup A_2^*}\setminus(A_1^* \cup A_2^*)) =
0$, we also have
\[
\lim_{n\to\infty}P^n \bigl( \bigl\{ D\in
X^n: \mu\bigl(B_1(D) \symdif A_1^*\bigr) +
\mu\bigl(B_2(D) \symdif A_2^*\bigr) \leq\eps \bigr\}
\bigr) = 1,
\]
where, for $B_1(D)$ and $B_2(D)$,
we use the same numbering as in \eqref{cluster-chunk-generic-bound}.
\end{theorem}

Note that
the assumption
$\mu(\overline{A_i^* \cup A_2^*}\setminus(A_1^* \cup A_2^*)) = 0$
is satisfied if there exists a $\mu$-density $h$ of $P$
such that $ \mu(\partial\{h\leq{\rho^*}\})=0$; see
\cite{SteinwartXXb}, Section A.9. 

Theorem~\ref{main} shows that
for suitably chosen parameters and histogram-based level set estimates
Algorithm \ref{cluster-algo} asymptotically recovers both ${\rho^*}$ and
the clusters $A_1^*$ and $A_2^*$, if the distribution $P$ has level
sets that are thicker than a user-specified
order $\g$.
To illustrate this, suppose that we
choose $\delta_n\sim n^{-\alpha}$ and $\e_n\sim n^{-\beta}$ for some
$\alpha,\beta>0$. Then it is easy to check that \eqref{cons-assump} is
satisfied if and only if
\mbox{$2(\alpha\dpart+\beta)<1$}. For $\t_n \sim n^{-\alpha\g}\ln n$, we then have
$\psi(\delta_n) < \t_n$ for all sufficiently large $n$, and therefore,
Algorithm \ref{cluster-algo}
recovers the clusters for all distributions $P$ that have thick levels
of order $\g$. Similarly,
the choice $\t_n \sim(\ln n)^{-1}$ leads to consistency for all
distributions $P$ that have thick levels
of some order $\g>0$.
Finally note that \eqref{cons-assump} can be replaced by
\[
\frac{ \ln\delta_n^{-1}} {n \delta_n^{\dpart}\e_n^2} \to0
\]
if we restrict our consideration to distributions with bounded $\mu
$-densities. The proof of
this is a straightforward modification of the proof of Theorem~\ref{main}.

To give two examples, recall from the discussion in \cite{SteinwartXXb}, Section
A.5,
that for the one-dimensional case $X=[a,b]$, we always have $\g=1$. In
two dimensions this is, however, no longer true  as, for
example, Figure~\ref{fig:diff_psi}
illustrates.
Nonetheless, there do exist many examples of both discontinuous and
continuous densities
for which we have thickness $\g=1$; see \cite{SteinwartXXb}, Section B.2.
Finally note that the construction used there  can be easily
generalized to higher dimensions.

%

For our next goal, which is establishing rates
for both $\mu(B_i(D) \symdif A_i^*) \to0$ and ${\rho^*_D}\to{\rho^*}$,
we need, as usual,
some assumptions on $P$. Let us begin by introducing
an assumption that leads to rates for the estimation of ${\rho^*}$.

\begin{definition}
Let Assumption~\ref{assc} be satisfied.
Then the clusters of $P$ have separation exponent $\k\in(0,\infty]$
if there is a constant
$\csepl>0$ such that
\[
\t^*(\e) \geq\csepl\e^{1/\k}
\]
for all $\e\in(0,{\rho^{**}}-{\rho^*}]$.
Moreover, the separation exponent $\k$ is exact if
there exists another
constant $\csepu>0$ such that,
for all $\e\in(0,{\rho^{**}}-{\rho^*}]$, we have
\[
\t^*(\e) \leq\csepu\e^{1/\k}.
\]
\end{definition}

The separation exponent describes how fast the connected components
of the
$M_\rho$ approach each other for $\rho\searrow{\rho^*}$. Note that
the separation exponent is monotone, that is,  a distribution having
separation exponent $\k$ also has
separation exponent $\k'$ for all $\k'<\k$.
In particular, the ``best'' separation exponent is $\k=\infty$, and this
exponent describes distributions,
for which we have $d(A_1^*, A_2^*)\geq\csepl$;
that is, the clusters $A_1^*$ and $A_2^*$ do not touch each other.

To illustrate the separation exponent, let us consider
$X:= -[3,3]$ and, for $\theta, \beta\in(0,\infty]$ and ${\rho^*}\in[0,1/6)$,
the distribution $P_{\theta,\beta}$ that has the density
%
\begin{equation}\qquad
\label{one-dim-dens} h_{\theta,\beta}(x) := {\rho^*}+ c_{\theta,\beta} \bigl( \eins
_{[0,1]}\bigl(|x|\bigr)|x|^\theta + \eins_{[1,2]}\bigl(|x|\bigr) +
\eins_{[2,3]}\bigl(|x|\bigr) \bigl(3-|x|\bigr)^\beta \bigr),
\end{equation}
where $c_{\theta,\beta}$ is a constant ensuring
that $h_{\theta,\beta}$ is a probability density; 
see also Figure~\ref{fig:density_with_valley} for two examples.
Note that $P_{\theta,\beta}$ can be clustered between
${\rho^*}$ and ${\rho^{**}}:= {\rho^*}+ c_{\theta,\beta}$.
Moreover, $P_{\theta,\beta}$ always has exact separation exponent
$\theta$.

The polynomial behavior in the upper vicinity of ${\rho^*}$ of the
distributions \eqref{one-dim-dens}
is somewhat archetypal for smooth densities on $\R$.
%
For example, for $C^2$-densities $h$ whose first derivative $h'$ has
exactly one zero $x_0$ in the set $\{h={\rho^*}\}$
and whose second derivative satisfies $h''(x_0)>0$, one can easily show
with the help of Taylor's theorem
that their behavior in the upper vicinity of ${\rho^*}$ is asymptotically
identical to that of \eqref{one-dim-dens}
for $\k=\theta=2$ and $\beta=1$. Moreover, larger values for $\k
=\theta$
can be achieved
by assuming that higher derivatives of $h$ vanish at $x_0$.
Analogously, the class of continuous densities on $\R^2$
from \cite{SteinwartXXb}, Section B.2, have separation exponent $\k
=2$ (see
\cite{SteinwartXXb}, Example B.2.1), as these densities, similar to
Morse functions,
behave like $x_1^2-x_2^2$ in the vicinity of the
saddle point.
Again,
the
construction can be modified to achieve other exponents.

In the following we show how
the separation exponent influences the rate for estimating ${\rho^*}$.
We begin with a finite sample bound.

%

\begin{theorem}\label{rates-1}
Let Assumptions \textup{\ref{assa}} and \textup{\ref{assc}} be satisfied, and assume additionally that $P$
has a bounded
$\mu$-density $h$ and that its
clusters have separation exponent $\k\in(0,\infty]$.
For some fixed $\delta\in(0,\dthick]$, $\vs\geq1$, $n\geq1$ and
$\t\geq2 \psi(\delta)$, we pick an $\e>0$ satisfying \eqref
{analysis-main2-alter}, that is,
\[
\e\geq\sqrt{\frac{2\cpart(1 + \Vert h \Vert_\infty)(\vs+ \ln
(2\cpart) -
\dpart
\ln\delta)}{\delta^{\dpart}n}} + \frac{2\cpart(\vs+ \ln(2\cpart) - \dpart\ln\delta)}{3\delta
^{\dpart}n}.
\]
Let us assume that $\e^* := \e+ (\t/ \csepl)^\k$ satisfies $\e^*
\leq
({\rho^{**}}-{\rho^*})/9$.
Then if Algorithm \ref{cluster-algo} receives the input parameters $\e
$, $\t$ and
the family $(L_{D,\rho})_{\rho\geq0}$ given by $ L_{D,\rho} := \{
h_{D, \delta}
\geq\rho\}$,
the probability $P^n$ of a $D\in X^n$ that satisfies
%
\begin{eqnarray}
\label{rates-level-low} \e&< &{\rho^*_D}-{\rho^*},
\\
\label{rates-level-up} {\rho^*_D}-{\rho^*}&\leq& (\t/ \csepl)^\k+
6\e
\end{eqnarray}
is not less than $1-e^{-\vs}$. Moreover, if the separation exponent
$\k
$ is exact and $\k<\infty$, then
we can replace \eqref{rates-level-low} by
%
\begin{equation}
\label{rates-level-exact} \frac{1}4 \biggl(\frac{\t}{6\csepu} \biggr)^\k+
\e< {\rho^*_D}-{\rho^*}.
\end{equation}
\end{theorem}

The finite sample guarantees of Theorem~\ref{rates-1} can be easily
used to derive (exact) rates for ${\rho^*_D}\to{\rho^*}$.
The following corollary presents, modulo
(double) logarithmic factors, the best rates
we can derive by this approach.

\begin{corollary}\label{rates-cor1}
Let Assumptions \textup{\ref{assa}} and \textup{\ref{assc}} be satisfied, and assume that $P$ has bounded
$\mu$-density and that its
clusters have separation exponent $\k\in(0,\infty)$. Furthermore, let
$(\e_n)$, $(\delta_n)$ and $(\t_n)$ be sequences with
\[
\e_n \sim \biggl( \frac{\ln n\cdot\ln\ln n}n \biggr)^{{\g\k
}/{(2\g\k+ \dpart)}},\qquad
\delta_n \sim \biggl( \frac{\ln n}n \biggr)^{{1}/{(2\g\k+\dpart)}}\quad
\mbox{and}\quad \t_n \sim\e_n^{1/\k},
\]
and assume that, for $n\geq1$, Algorithm \ref{cluster-algo} receives
the input parameters $\e_n$, $\t_n$ and
the family $(L_{D,\rho})_{\rho\geq0}$ given by $ L_{D,\rho} := \{
h_{D, \delta
_n} \geq\rho\}$. Then
there exists a constant $\overline K\geq1$ such that for all
sufficiently large $n$, we have
%
\begin{equation}
\label{rates-cor1-hx} P^n \bigl( \bigl\{D\in X^n: {
\rho^*_D}-{\rho^*} \leq\overline K \e_n
 \bigr\} \bigr) \geq1 - \frac{1}n.
\end{equation}
Moreover, if the separation exponent $\k$ is exact, there exists
another constant $\underbar K\geq1$
such that for all sufficiently large $n$, we have
%
\begin{equation}
\label{rates-cor1-hxx} P^n \bigl(\bigl\{D\in X^n: \underbar K
\e_n 
\leq{\rho^*_D}-{\rho^*}\leq
\overline K \e_n 
 \bigr\}\bigr) \geq1 -
\frac{1}n.
\end{equation}
Finally, if $\k=\infty$, then \eqref{rates-cor1-hxx} holds
for all sufficiently large $n$ if
\[
\e_n \sim \biggl( \frac{\ln n \cdot\ln\ln n}n \biggr)^{{1}/ 2},\qquad
\delta_n \sim ( \ln\ln n )^{-{1}/{(2\dpart)}} \quad\mbox{and}\quad \t_n
\sim ( \ln\ln n )^{-{\g}/{(3\dpart)}}.
\]
%
\end{corollary}

%

Recall that for the one-dimensional distributions \eqref{one-dim-dens}
we have $\g=1$ and $\k= \theta$, so that the exponent in the rates
above becomes $\frac{\theta}{2\theta+1}$.
In particular, for the $C^2$-case discussed there, we have $\theta=2$,
and thus we
get a rate with exponent $2/5$, while for $\theta\to\infty$ the
exponent converges to
$1/2$. Similarly, for the typical, two-dimensional distributions
considered in \cite{SteinwartXXb}, Section B.2, we have
$\g=1$, $\k= 2$ and $\dpart=2$, and hence the exponent in the rate
is $1/3$.

Our next goal is to establish rates for $\mu(B_i(D)\symdif A_i^*)\to
0$. Since this is a modified level set estimation problem, let us
recall some
assumptions on $P$, which have been used in this context. 
The first assumption in this direction is a one-sided variant of a
well-known condition
introduced by Polonik
\cite{Polonik95a}.


\begin{definition}
Let $\mu$ be a finite measure on $X$ and $P$ be a distribution on $X$
that has a $\mu$-density $h$.
For a given level $\rho\geq0$, we say that
$P$ has flatness exponent $\vt\in(0,\infty]$ if there exists a
constant $\cflat>0$ such that
%
\begin{equation}
\label{polon} \mu \bigl(\{ 0< h-\rho<s \} \bigr) \leq(\cflat s)^\vt,\qquad
s>0.
\end{equation}
\end{definition}

Clearly, the larger the $\vt$, 
the more steeply $h$ must approach $\rho$ from above. In particular, for
$\vt=\infty$, the
density $h$ is allowed to take the value $\rho$ but is otherwise
bounded away from $\rho$.
For example,
the
densities in \eqref{one-dim-dens} have
a flatness exponent $\vt= \min\{1/\theta, 1/\beta\}$ if
$\theta<\infty$ and $\beta<\infty$ and a flatness exponent $\vt=
\infty$
if $\theta=\beta=\infty$.
Finally, for the two-dimensional distributions of \cite{SteinwartXXb}, Section
B.2, the flatness exponent
is not fully determined by their definition,
but some calculations
show that we have $\vt\in(0,1]$.

Next, we describe 
the roughness of the boundary of the clusters.

\begin{definition}
Let Assumption~\ref{assc} be satisfied. Given some $\alpha\in(0,1]$,
the clusters have
an $\alpha$-smooth boundary if
there exists a constant $\cbound>0$ such that, for all $\rho\in
({\rho^*}
,{\rho^{**}}
]$, $\delta\in(0,\dthick]$ and $i=1,2$, we have
%
\begin{equation}
\label{box-dim} \mu \bigl( \bigl(A^i_\rho
\bigr)^{+\delta}\setminus\bigl( A^i_\rho
\bigr)^{-\delta
} \bigr) \leq \cbound \delta^\alpha,
\end{equation}
where $A^1_\rho$ and $A^2_\rho$ denote the two connected components
of the
level set $M_\rho$.
\end{definition}

In $\R^d$, considering $\alpha>1$ does not make sense, and for an
$A\subset
\mathbb{R}^d$
with rectifiable boundary, we always have $\alpha=1$; see \cite{SteinwartXXb}, Lemma
A.10.4.
The $\alpha$-smoothness of the boundary thus enforces a uniform version of
this, which, however,
is not very restrictive; see,
for example, the densities of \eqref{one-dim-dens}, for which we have
$\alpha=1$ and $\cbound= 4$, and
\cite{SteinwartXXb}, Example B.2.2, for which we also have $\alpha=1$.

\begin{figure}

\includegraphics{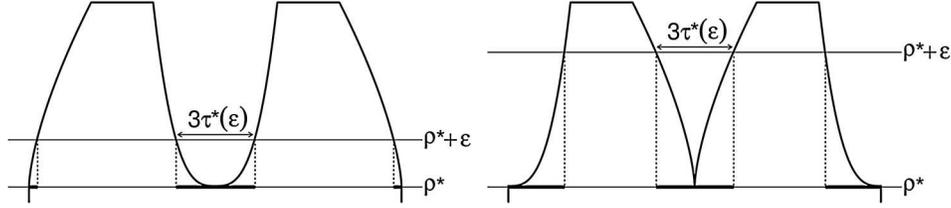}

\caption{Separation and flatness. Left: The density
$h_{\theta
,\beta}$ described in \protect\eqref{one-dim-dens} for $\theta= 3$ and
$\beta=2/3$.
The bold horizontal line indicates the set $\{{\rho^*}<h<{\rho^*}+\e
\}$, and
$3\t
^*(\e)$ describes the width of the valley at level
${\rho^*}+\e$.
Right: Here we have the same situation for $\theta= 2/3$ and $\beta
=3$. The value of $\e$ is chosen such that $3\t^*(\e)$
equals the value on the left. The smaller value of $\theta$ narrows the
valley, and
hence $\e$ needs to be chosen larger. As a result, it becomes more
difficult to estimate ${\rho^*}$ and the clusters.
Indeed, ignoring logarithmic
factors, Corollary \protect\ref{rates-cor1} gives a rate of
$n^{-3/7}$ on the
left and a rate of $n^{-2/7}$ on the right,
while Corollary \protect\ref{rates-cor2} gives a rate of $n^{-1/7}$
on the left
and a rate of $n^{-2/21}$ on the right.
Finally, in the most typical case $\theta= 2$ and $\beta=1$ not
illustrated here, we obtain the rates $n^{-1/3}$
and $n^{-1/5}$.}
\label{fig:density_with_valley}
\end{figure}

%

The following assumption collects all conditions we need to impose on
$P$ to get rates
for estimating the clusters.

\renewcommand{\theassumption}{R}
\begin{assumption}\label{assr}
Assumptions \textup{\ref{assa}} and \textup{\ref{assc}} are satisfied, and $P$ has\break a
bounded $\mu$-density $h$. Moreover,
$P$ has a flatness exponent $\vt\in(0,\infty]$ at level ${\rho^*}$,
its clusters have an $\alpha$-smooth boundary for some $\alpha\in(0,1]$
and its clusters have a separation exponent $\k\in(0,\infty]$.
\end{assumption}

Let us now investigate how well our algorithm estimates the clusters
$A_1^*$ and~$A_2^*$.
As usual, we begin with a finite-sample estimate.

\begin{theorem}\label{rates-2}
Let Assumption~\ref{assr} be satisfied, and assume that $\delta$, $\e$, $\t$,
$\e
^*$, $\vs$, $n$ and $(L_{D,\rho})_{\rho\geq0}$
are as in Theorem~\ref{rates-1}. Then
the probability $P^n$ of having a data set $D\in X^n$
satisfying \eqref{rates-level-low}, \eqref{rates-level-up} and
\[
\mu \bigl(B_1(D) \symdif A_1^* \bigr) + \mu
\bigl(B_2(D) \symdif A_2^* \bigr) \leq6\cbound
\delta^\alpha+ \bigl(\cflat( \t/ \csepl)^\k+ 7\cflat\e
\bigr)^\vt
\]
%
is not less than $1-e^{-\vs}$, where the sets $B_1(D)$ and $B_2(D)$
are ordered as in \eqref{cluster-chunk-generic-bound}.
Moreover,
if the separation exponent $\k$ is exact and satisfies $\k<\infty$,
then \eqref{rates-level-exact} also
holds for these data sets $D$.
\end{theorem}

Note that
for \emph{finite} values of $\vt$ and $\k$,
the bound in Theorem~\ref{rates-2} behaves like $\delta^\alpha+ \t^{\vt
\k} + \e
^\vt$,
and in this case it is thus easy to derive the best convergence rates
our analysis yields.
The following corollary
presents corresponding results and also provides rates for the cases
$\vt= \infty$ or $\k=\infty$.

\begin{corollary}\label{rates-cor2}
Assume that Assumption~\ref{assr} is satisfied, and write $\vr:= \min\{\alpha,
\vt
\g\k\}$.
Furthermore, let
$(\e_n)$, $(\delta_n)$ and $(\t_n)$
be sequences with
\begin{eqnarray*}
\e_n &\sim& \biggl( \frac{\ln n}{n} \biggr)^{{\vr}/{(2\vr+ \vt
\dpart)}} (\ln\ln
n)^{-{\vt\dpart}/{(8\vr+4\vt\dpart)}} , \\
\delta_n &\sim &\biggl( \frac{\ln n \cdot\ln\ln n}n
\biggr)^{
{\vt
}/{(2\vr+ \vt\dpart)}}\quad \mbox{and}
\\
\t_n &\sim& \biggl( \frac{\ln n \cdot(\ln\ln n)^2}n \biggr)^{
{\vt
\g}/{(2\vr+ \vt\dpart)}}.
\end{eqnarray*}
Assume that, for $n\geq1$, Algorithm \ref{cluster-algo} receives the
parameters $\e_n$, $\t_n$ and
the family $(L_{D,\rho})_{\rho\geq0}$ given by $ L_{D,\rho} := \{
h_{D, \delta
_n} \geq\rho\}$. Then
there is a constant $\overline K\geq1$ such that, for all $n\geq1$
and the ordering as in \eqref{cluster-chunk-generic-bound},
we have
\[
P^n \Biggl( D: \sum_{i=1}^2
\mu \bigl(B_i(D) \symdif A_i^* \bigr) 
\leq K \biggl( \frac{\ln n \cdot(\ln\ln n)^2}n \biggr)^{{\vt
\vr
}/{(2\vr+ \vt\dpart)}} \Biggr) \geq1-
\frac{1} n.
\]
%
\end{corollary}

Let us now
compare the established rates for estimating ${\rho^*}$ and the clusters
in the most important
case, that is, $\alpha=1$. If $\vt\g\k\leq1$, we obtain
$\vr= \vt\g\k$ in Corollary~\ref{rates-cor2}, and the exponent in
the asymptotic behavior of
the optimal $(\delta_n)$ becomes $\frac{1} {2\g\k+ \dpart}$. Since this
equals the exponent in Corollary~\ref{rates-cor1}, and, modulo the extra $\ln\ln n$ terms, we also have
the same behavior for
$(\e_n)$ and $(\t_n)$ in both corollaries, we conclude that we obtain
the rates in Corollaries \ref{rates-cor1}
and \ref{rates-cor2} with (essentially) the \emph{same} controlling sequences
$(\e_n)$, $(\delta_n)$ and $(\t_n)$ of Algorithm \ref
{cluster-algo}. If
$\vt\g\k\leq1$,
we can thus achieve the
best rates for estimating ${\rho^*}$ and the clusters \emph{simultaneously}.
Unfortunately, this changes if $\vt\g\k> 1$. Indeed, while the
exponent for $(\delta_n)$
in Corollary~\ref{rates-cor1} remains the same, it changes
from $\frac{1} {2\g\k+ \dpart}$
to $\frac{\vt}{2+\vt\dpart}$ in Corollary~\ref{rates-cor2}, and a
similar effect takes place for the sequences
$(\e_n)$ and $(\t_n)$.
The reason for this difference is that
in the case $\vt\g\k> 1$ the estimation of
${\rho^*}$ is easier
than the
estimation of the level set $M_{\rho^*}$, and since for estimating the
clusters we need to do both,
the level set estimation rate determines the rate for estimating the clusters.

To illustrate this difference between the estimation of
${\rho^*}$ and the clusters in more detail, let us consider the toy model
\eqref{one-dim-dens} in the case $\theta=\beta=\infty$, that is, $\k=
\infty$.
Then the clusters are stumps, and the sets $M_\rho$
do not change between ${\rho^*}$ and ${\rho^{**}}$.
Intuitively, the best choice for estimating ${\rho^*}$ are then sufficiently
small but fixed values for $\delta_n$ and $\t_n$,
so that $\e_n$ converges to $0$ as fast as possible.
In Corollary~\ref{rates-cor1} this is mimicked by choosing very slowly
decaying sequences
$(\delta_n)$ and $(\t_n)$.
On the other hand, to find $A_1^*$ and $A_2^*$
%
it suffices to identify one $\rho\in({\rho^*},{\rho^{**}}]$ and to estimate
the connected components of $M_\rho$. The best way to achieve this is to
use a sufficiently small but fixed
value for $\e_n$
and sequences $(\delta_n)$ and $(\t_n)$ that converge to zero as fast
as possible.
In Corollary~\ref{rates-cor2} this is mimicked by choosing a very
slowly decaying sequence $(\e_n)$
and quickly decaying sequences $(\delta_n)$ and $(\t_n)$.

As for estimating the critical level ${\rho^*}$, we do not know so
far, whether
our rates for estimating the clusters are minmax optimal, but our
conjecture is that
they are optimal modulo the
logarithmic terms. To motivate our conjecture, let us consider the case
$\alpha=\g=1$. Moreover, assume that two-sided versions of
\cite{SteinwartXXb}, (A.10.4) and (A.10.6),
hold for all $\rho\in({\rho^*}, {\rho^{**}}]$, respectively, $\rho
={\rho^*}$.
Then we have $\k= \theta$ and $\vt= 1/\theta$ by
\cite{SteinwartXXb}, Lemmas A.10.1 and A.10.5,
and thus we find $\vr= 1$. Consequently, the rates in Corollary~\ref{rates-cor2} have the exponent
$\frac{1}{2\theta+d}$. This is exactly the same exponent as the one
obtained in
\cite{SiScNo09a} for minmax optimal and adaptive Hausdorff
estimation of a fixed level set.
In addition, it seems that their lower bound, which is based on
\cite{Tsybakov97a},
is, modulo logarithmic factors, the same for assessing the estimator in
the way we have
done it in Corollary~\ref{rates-cor2}.
While this coincidence indicates that our rates may be (essentially)
optimal, it is,
of course, not a rigorous argument. A detailed analysis is, however,
out of the scope of
this paper. Another interesting question, which is also out of the
scope, is whether
the estimates $B_i(D)$ approximate the true clusters $A_i^*$ in the
Hausdorff metric, too, and
if so, whether we can achieve the rates reported in \cite{SiScNo09a}.

\section{Data-dependent parameter selection}\label{sec:par-select}

In the last section we derived rates of convergence for both the
estimation of ${\rho^*}$ and the
clusters.
In both cases, our best rates required sequences $(\e_n)$, $(\delta
_n)$ and
$(\t_n)$
that did depend on some properties of $P$,
namely $\alpha$, $\k$, $\vt$. 
Of course, these parameters are not available to us in practice, and
therefore the
obtained rates are of little practical value. The goal of this final
section is to address this
issue by proposing a simple data-dependent parameter selection strategy
that is able to recover the rates of Corollary~\ref{rates-cor1} without knowing anything about $P$. We further show that
this selection strategy recovers the rates of Corollary~\ref{rates-cor2}
in the case of $\vt\g\k\leq\alpha$.

We begin by presenting the parameter selection strategy. To this end,
%
let $\D\subset(0,1]$ be finite and $n\geq1$, $\vs\geq1$.
For $\delta\in\D$, we fix a $\t_{\delta,n}>0$ and define
%
\begin{eqnarray}\label{eps-adap}
\e_{\delta,n} &:=& C\sqrt{\frac{\cpart (\vs+ \ln(2\cpart|\D|) -\dpart\ln
\delta
 )\ln\ln n}{\delta^{\dpart}n}}
\nonumber
\\[-8pt]
\\[-8pt]
\nonumber
 &&{} + \frac{2\cpart (\vs+ \ln(2\cpart|\D|) -\dpart\ln\delta
)}{3\delta^{\dpart}n},
\end{eqnarray}
where $C\geq1$ is some user-specified constant.
Now assume that, for each $\delta\in\D$,
we run Algorithm \ref{cluster-algo} with the parameters $\e_{\delta
,n}$ and
$\t_{\delta,n}$,
and the family $(L_{D,\rho})_{\rho\geq0}$ given by $ L_{D,\rho} :=
\{h_{D, \delta
} \geq\rho\}$.
We
write $\rho^*_{D,\delta}$ for the
corresponding level returned by Algorithm \ref{cluster-algo}. Let us consider
a width $\delta_{D,\D}^*\in\D$ that achieves the smallest returned level,
that is,
%
\begin{equation}
\label{def-rdds} \delta_{D,\D}^* \in\arg\min_{\delta\in\D}
\rho^*_{D,\delta}.
\end{equation}
Note that in general, this width may not be uniquely determined, so
that in the following
we need to additionally assume that we have a well-defined choice, for
example, the smallest $\delta\in\D$
satisfying \eqref{def-rdds}.
Moreover, we write
%
\begin{equation}
\label{def-rsdd} \rho^*_{D,\D} := \rho^*_{D,\delta^*_{D,\D}} = \min
_{\delta\in\D
} \rho^*_{D,\delta}
\end{equation}
for the smallest returned level.
Note that unlike $\delta_{D,\D}^*$, the level $\rho^*_{D,\D}$ is
always unique.
Finally, we define $\e_{D,\D}:=\e_{\delta^*_{D,\D},n}$ and $\t
_{D,\D}:=\t
_{\delta^*_{D,\D},n}$.

Our first goal is to show that $\rho^*_{D,\D}$ achieves the rates of
Corollary~\ref{rates-cor1}
for suitably chosen
$\D$ and $\t_{\delta,n}$. We begin with a finite sample guarantee.

\begin{theorem}\label{adaptive-level}
Let Assumptions \textup{\ref{assa}} and \textup{\ref{assc}} be satisfied, and assume that $P$ has a
bounded $\mu$-density $h$, and that
the two clusters of $P$ have separation exponent $\k\in(0,\infty]$.
For a fixed finite $\D\subset(0,\dthick]$, and $n\geq1$, $\vs\geq1$
and $C\geq1$, we define
$\e_{\delta,n}$ by \eqref{eps-adap} and choose $\t_{\delta,n}$
such that
$\t_{\delta,n} \geq2 \psi(\delta)$ for all $\delta\in\D$.
Furthermore, assume that
$C^2\ln\ln n\geq2(1+\Vert h \Vert_\infty)$ and
$ \e^*_\delta:= \e_{\delta,n} + (\t_{\delta,n}/ \csepl)^\k\leq
({\rho^{**}}-{\rho^*})/9$
for all $\delta\in\D$. Then we have
\[
P^n \Bigl( \Bigl\{D\in X^n: \e_{D,\D} <
\rho^*_{D,\D}- {\rho ^*}\leq \min_{\delta
\in\D} \bigl((
\t_{\delta,n}/ \csepl)^\k+ 6\e_{\delta,n} \bigr) \Bigr\}
\Bigr) \geq1-e^{-\vs}.
\]
Moreover, if the separation exponent $\k$ is exact and $\k<\infty$,
then the assumptions above
actually guarantee
\[
P^n \Bigl(D: \min_{\delta\in\D} \bigl(c_1
\t_{\delta,n}^\k+ \e _{\delta,n} \bigr) <
\rho^*_{D,\D}- {\rho^*}\leq\min_{\delta\in\D}
\bigl(c_2\t _{\delta
,n}^\k+ 6\e_{\delta
,n}
\bigr) \Bigr) \geq1-e^{-\vs},
\]
where $c_1 := \frac{1}4( 6\csepu)^{-\k}$ and $c_2 := \csepl^{-\k}$,
and similarly
\[
P^n \bigl( \bigl\{D\in X^n: c_1
\t_{D,\D}^\k+ \e_{D,\D} < \rho ^*_{D,\D}- {
\rho^*}\leq c_2 \t_{D,\D}^\k+ 6\e_{D,\D}
\bigr\} \bigr) \geq1-e^{-\vs}.
\]
\end{theorem}

Theorem~\ref{adaptive-level} establishes the same finite sample
guarantees for
the estimator $\rho^*_{D,\D}$ as Theorem~\ref{rates-1} did for the
simpler estimator ${\rho^*_D}$.
Therefore, it is not surprising that for suitable choices of $\D$, the
rates of
Corollary~\ref{rates-cor1} can be recovered, too.
The next corollary shows that this can actually be achieved for
candidate sets $\D$ that are completely independent
of $P$.

\begin{corollary}\label{adaptive-level-cor}
Assume that Assumptions \textup{\ref{assa}} and \textup{\ref{assc}} are satisfied, that $P$ has a bounded
$\mu$-density $h$ and that
the two clusters of $P$ have separation exponent $\k\in(0,\infty]$.
For $n\geq16$, we consider the interval
\[
I_n:= \biggl[ \biggl( \frac{\ln n \cdot(\ln\ln n)^2}n \biggr)^{
{1}/{\dpart}},
\biggl(\frac{1} {\ln\ln n} \biggr)^{{1}/\dpart} \biggr]
\]
and fix some $n^{-1/\dpart}$-net $\D_n\subset I_n$ of $I_n$ with $|\D
_n|\leq n$.
Furthermore, for some fixed $C\geq1$ and $n\geq16$, we write $\t
_{\delta
,n}:= \delta^\g\ln\ln\ln n$ and define
$\e_{\delta,n}$ by \eqref{eps-adap}
for all $\delta\in\D_n$ and $\vs= \ln n$. Then there exists a constant
$\overline K$ such that, for all
sufficiently large $n$, we have
%
\begin{equation}\qquad
\label{adaptive-level-cor-hx} P^n \biggl(D: \e_{D,\D_n}<
\rho^*_{D,\D_n}-{\rho^*} \leq\overline K \biggl( \frac{\ln n \cdot(\ln\ln n)^2}n
\biggr)^{{\g\k}/{(2\g\k+ \dpart)}} \biggr) \geq1 - \frac{1}n.
\end{equation}
%
If, in addition, the separation exponent $\k$ is exact and $\k<\infty
$, then 
there is
another constant $\underline K$ such that for all sufficiently large
$n$, we have
\begin{eqnarray*}
&&P^n \biggl( D : \underline K \biggl( \frac{\ln n \cdot\ln\ln n}n
\biggr)^{{\g\k}/{(2\g\k+ \dpart)}} \leq\rho^*_{D,\D_n} -{\rho^*}
\\
&&\qquad\quad \leq\overline K
\biggl( \frac{\ln n \cdot(\ln\ln n)^2}n
 \biggr)^{{\g\k}/{(2\g\k+ \dpart)}} \biggr)\\
 &&\qquad \geq1 -
\frac{1} n.
\end{eqnarray*}
\end{corollary}

Finally, we show that our parameter selection strategy
partially recovers the rates
for estimating the clusters $A_i^*$
obtained in Corollary~\ref{rates-cor2}.

\begin{corollary}\label{rates-adaptive-cor2}
Assume that Assumption~\ref{assr} is satisfied with $\alpha\geq\vt\g\k$ and
exact separation exponent $\k$.
Then, for the procedure of Corollary~\ref{adaptive-level-cor}, there
is a $K\geq1$ such that
for $n\geq1$ and the ordering as in \eqref
{cluster-chunk-generic-bound}, we have
\[
P^n \Biggl( D: \sum_{i=1}^2
\mu\bigl(B_i(D) \symdif A_i^*\bigr) 
\leq K \biggl( \frac{\ln n \cdot(\ln\ln n)^2}n \biggr)^{{\vt
\g\k
}/{(2\g\k+ \vt\dpart)}} \Biggr) \geq1-
\frac{1}n.
\]
%
\end{corollary}

Unfortunately, the simple parameter selection strategy \eqref{def-rdds}
is not adaptive
in the case $\alpha< \vt\g\k$, that is, in the case in which the
estimation of ${\rho^*}$ is easier than
the estimation of the corresponding clusters. It is unclear to us
whether in this case a two-stage
procedure that first estimates ${\rho^*}$ by $\rho^*_{D,\D_n}$
as above, and then uses a different strategy to estimate
the connected components at the level $\rho^*_{D,\D_n}$ can be made adaptive.

\section{Selected proofs}\label{sec-select-proofs}

In this section we present some selected proofs. All remaining proofs
can be found in \cite{SteinwartXXb}.

\begin{pf*}{Proof of Lemma~\ref{not-reg-ex}}
Let $(x_n)$ be an enumeration of $\Q\cap[0,1]$ and $I_n :=
[x_n-2^{-n-2}, x_n+2^{-n-2}]\cap[0,1]$ for $n\geq1$.
For $x\in[0,1]$ and $I_0 := [0,1]$, we further define
\[
f(x) := \sup_{n\geq0} n\eins_{I_n}(x),
\]
that is,
$f(x)$ equals the largest integer $n\geq0$ (including infinity) such
that $x\in I_n$.
For $c>0$ specified below, we now define
\[
h(x) := %
\cases{\displaystyle 2c-\frac{c} {f(x)}, &\quad $\mbox{if } f(x)>0$,
\vspace*{2pt}
\cr
0, &\quad $\mbox{else.}$} %
\]
Then $h$ is measurable, nonnegative
and Lebes\-gue-integrable, and hence we can
choose $c$ such that $\int_0^1 h(x) \,dx =1$.
Then $h$ is a density of a Lebesgue-absolutely continuous distribution $P$.
Moreover, note that $h(x) \geq2c - c/n$ for all $x\in I_n$ and $n\geq1$.
For a fixed $\rho\in(0,2c)$ we now write $n_\rho:= c/ (2c-\rho)$.
Then we have $2c-c/n\geq\rho$ if and only if
$n \geq n_\rho$. Consequently, the set
\[
A_\rho:= \bigcup_{n\geq n_\rho}^\infty
\mathring I_n
\]
satisfies $A_\rho\subset\{h\geq\rho\}$. Moreover, since $A_\rho$
is open,
we find
$A_\rho\subset\mathring{\{h\geq\rho\}}$, and thus
\[
\overline{ A_\rho} \subset\overline{\mathring{\{h\geq\rho\} }}\subset
M_\rho
\]
by \cite{SteinwartXXb}, Lemma A.1.2. In addition, we have $\{x_n:
n\geq n_\rho\}\subset A_\rho$, and since the former set is dense in $[0,1]$,
we conclude that $M_\rho= [0,1]$. On the other hand, the Lebesgue
measure $\lb$ of $\{h\geq\rho\}$ can be estimated by
\[
\lb\bigl(\{h\geq\rho\}\bigr) \leq\lb\bigl(\{h> 0\}\bigr) =\lb \Biggl( \bigcup
_{n=1}^\infty I_n \Biggr) \leq\sum
_{n=1}^\infty\lb(I_n) \leq\sum
_{n=1}^\infty2^{-n-1} =
\frac{1} 2,
\]
and hence we conclude that $\lb(M_\rho\setminus\{h\geq\rho\})\geq1/2$.
In other words, $P$ is not normal at level $\rho$.
\end{pf*}

\begin{pf*}{Proof of Theorem~\ref{reg-cluster-thm-combined}}
The monotonicity of ${\t^*}$ is shown in \cite{SteinwartXXb}, Theorem~A.4.2,
and (i) follows from
parts (i) of \cite{SteinwartXXb}, Theorems~A.4.2 and A.4.4.

(ii) Let us first consider the case $\rho<{\rho^*}$. Since $P$
can be
clustered, we have $|\ca C(M_\rho)|=1$, and
\cite{SteinwartXXb}, Proposition A.2.10, gives both $\t^*_{M_\rho} =
\infty$ and $\ca C(M_\rho) = \ca C_\t(M_\rho)$.
By \cite{SteinwartXXb}, Lemma A.4.1, we further find $\ca C_\t(M_\rho)
\persist\ca C_\t(M_\rho^{+\delta})$.
Finally, part~(ii) of \cite{SteinwartXXb}, Lemma A.4.3, yields
$1\leq|\ca C_\t(M_\rho^{-\delta})|\leq|\ca C(M_\rho)| = 1$, and
hence its
part (iii)
gives the persistence $\ca C_\t(M_\rho^{-\delta}) \persist\ca
C(M_\rho)$.

In the case $\rho\geq{\rho^*}+\e^*$,
$\ca C_\t(M_\rho) \persist\ca C_\t(M_\rho^{+\delta})$ follows
from part~(ii) of
\cite{SteinwartXXb}, Theorem A.4.2, and the equality $\ca C(M_\rho) =
\ca C_\t(M_\rho)$ follows from \cite{SteinwartXXb}, Proposition~A.2.10, in combination with $\t\leq{\t^*}(\e^*) \leq
{\t^*}
(\rho
-{\rho^*})$.
By part~(ii) of \cite{SteinwartXXb}, Theorem A.4.4, we further
know $\ca C_\t(M_{\rho^{**}}^{-\delta})\persist\ca C_\t( M_{\rho
^{**}}^{+\delta})$.
Using $\rho\geq{\rho^*}+\e^*$ and part (iv) of \cite{SteinwartXXb}, Theorem
A.4.2,
we find $|\ca C_\t(M_\rho^{-\delta})|=2$, and hence part
(iii) of~\cite{SteinwartXXb}, Theorem A.4.4,
gives $\ca C_\t(M_\rho^{-\delta}) \persist\ca C(M_\rho)$.
\end{pf*}

\begin{pf*}{Proof of Theorem~\ref{analysis-main-combined-new}}
(i) The first bound on ${\rho^*_D}$
directly follows from part (i) 
of \cite{SteinwartXXb}, Theorem A.6.2.

To show \eqref{bound-rds-second}, we observe that
parts (iii) and (iv) of \cite{SteinwartXXb}, Theorem A.6.2,
imply $2 = |\ca C_\t(M_{\rho^*_D+\e}^{-\delta})| = |\ca C(M_{\rho
^*_D+\e})|$.
Since we further have ${\rho^*_D}+\e\leq{\rho^*}+\e^*+6\e\leq
{\rho^{**}}$ by the
first bound on ${\rho^*_D}$,
%
part (iii) of 
\cite{SteinwartXXb}, Lemma A.4.3,
thus shows
\[
d(B_1,B_2) \geq\t- 2\psi^*_{M_{\rho^*_D+\e}}(\delta) \geq
\t- 2\cthick \delta^\g> \t-\psi(\delta),
\]
where
$B_1$ and $B_2$ are the two connected components of $M_{\rho^*_D+\e}$.
On the other hand, the definition of $\t^*_{M_{\rho^*_D+\e}}$ in
\cite{SteinwartXXb}, Proposition A.2.10,
together with the definition of $\t^*$ in \eqref{reg-cluster-lem-def-tds-new}
gives
\[
3\t^*\bigl({\rho^*_D}-{\rho^*}+\e\bigr) = \t^*_{M_{\rho^*_D+\e}} =
d(B_1,B_2).
\]
Combining both we find \eqref{bound-rds-second}.

(ii) Part (iii) of \cite{SteinwartXXb}, Theorem A.6.2,
shows that Algorithm \ref{cluster-algo} returns two sets.
Our next goal is to find a suitable ordering of these sets. To this end,
we adopt the notation of \cite{SteinwartXXb}, Theorem A.6.2. Moreover,
we denote
the two topologically connected components of $M_{\rho^{**}}$ by $A_1$ and
$A_2$. We further write
\[
V^i_{{\rho^*_D}+\e} := \z_{{\rho^{**}}, {\rho^*_D}+\e} \bigl(\z
_{\rho^{**}} ^{-1}(A_i) \bigr),\qquad i=1,2,
\]
for the two $\t$-connected components of $M_{\rho^*_D+\e}^{-\delta
}$. Note that part
(iv) of \cite{SteinwartXXb}, Theorem A.6.2,
ensures that we can actually make this definition, and, in addition,
it shows $V^1_{{\rho^*_D}+\e} \neq V^2_{{\rho^*_D}+\e}$.
Moreover, by parts (ii) and (iii) of \cite{SteinwartXXb}, Theorem
A.6.2, we may assume that
the sets returned by Algorithm
\ref{cluster-algo} are ordered in the sense of $ B_i(D) = \z
(V^i_{{\rho^*_D}
+\e} )$, that is,
%
\begin{equation}
\label{ret-level-sets-order-new} B_i(D) = \z\circ\z_{{\rho^{**}}, {\rho^*_D}+\e} \bigl(
\z_{\rho
^{**}}^{-1}(A_i) \bigr),\qquad i=1,2.
\end{equation}

To simplify notation in the following calculations, we write $B_i :=
B_i(D)$ for $i\in\{1,2\}$ and $\rho:= {\rho^*_D}$.
Consequently, $A^1_{\rho+\e}$ and $A^2_{\rho+\e}$ are the two connected
components of $M_{\rho+\e}= M_{{\rho^*_D}+\e}$, which
by Definition~\ref{top-clust-def1-neu} can be ordered in the sense of
$A^i_{\rho+\e} \subset A^*_i$.
Moreover,
$V^1_{\rho+\e}$ and $V^2_{\rho+\e}$ become the two $\t$-connected
components of $M_{\rho+\e}^{-\delta}$.
For $i\in\{1,2\}$, we further write $W^i_{\rho+\e} := (A^i_{\rho+\e
})^{-\delta}$.
Our first goal is to show that
%
\begin{equation}
\label{approx-main-h5-neu} W^i_{\rho+\e}\subset V^i_{\rho+\e},\qquad
i\in\{1,2\}.
\end{equation}
To this end, we fix an $x\in W^1_{\rho+\e}$.
Since $W^1_{\rho+\e}\subset A^1_{\rho+\e}$ and $W^1_{\rho+\e
}\subset M_{\rho+\e
}^{-\delta}$,
where the latter follows from $(A^1_{\rho+\e})^{-\delta}\subset
M_{\rho
+\e}^{-\delta}$,
we then have
$x\in A^1_{\rho+\e}$ and $x\in V^1_{\rho+\e}\cup V^2_{\rho+\e}$.
Let us assume that $x\in V^2_{\rho+\e}$. Then we have $V^2_{\rho+\e
}\cap
A^1_{\rho+\e}\neq\varnothing$.
Now, the diagram of \cite{SteinwartXXb}, Theorem A.6.2,
shows that
$\z_{\rho+\e}: {\ca C_{\t}(M_{\rho+\e}^{-\delta})}\to{\ca
C(M_{\rho
+\e})}$
satisfies $\z_{\rho+\e_n} (V^2_{\rho+\e}) = A^2_{\rho+\e}$, and
hence we
have $V^2_{\rho+\e} \subset A^2_{\rho+\e}$.
Consequently, $V^2_{\rho+\e}\cap A^1_{\rho+\e}\neq\varnothing$ implies
$A^2_{\rho+\e}\cap A^1_{\rho+\e}\neq\varnothing$,
which is a contradiction. Therefore, we have $x\in V^1_{\rho+\e}$; that
is, we have shown \eqref{approx-main-h5-neu}
for $i=1$. The case $i=2$ can be shown analogously.

By \eqref{approx-main-h5-neu} we find $W_{\rho+\e}^i \subset V_{\rho
+\e
}^i\subset B_i$, and thus
$\mu(A_i^*\setminus B_i ) \leq\mu(A_i^* \setminus W^i_{\rho+\e})$
for $i=1,2$.
Conversely, using $\mu(B\setminus A) = \mu(B) - \mu(A\cap B)$ twice,
we obtain
\begin{eqnarray*}
\mu \bigl(B_1 \setminus\bigl(A_1^* \cup
A_2^*\bigr) \bigr) & =& \mu(B_1 ) - \mu
\bigl(B_1 \cap\bigl(A_1^* \cup A_2^*\bigr)
\bigr)
\\
&\geq &\mu(B_1 ) - \mu\bigl(B_1 \cap A_1^*
\bigr) - \mu\bigl(B_1 \cap A_2^*\bigr)
\\
& = &\mu\bigl(B_1 \setminus A_1^*\bigr) - \mu
\bigl(B_1 \cap A_2^*\bigr).
\end{eqnarray*}
Since $B_1\cap B_2 = \varnothing$ implies $B_1\cap A_2^*\subset
A_2^*\setminus B_2$, we thus find
\begin{eqnarray*}
\mu\bigl(B_1 \symdif A_1^*\bigr) & =& \mu
\bigl(B_1\setminus A_1^*\bigr) + \mu\bigl(A_1^*
\setminus B_1 \bigr)
\\
&\leq&\mu \bigl(B_1 \setminus\bigl(A_1^* \cup
A_2^*\bigr) \bigr) + \mu \bigl(A_2^*\setminus
B_2\bigr) + \mu\bigl(A_1^*\setminus B_1
\bigr)
\\
& \leq&\mu \bigl( B_1 \setminus\bigl\{ h>{\rho^*}\bigr\} \bigr) + \mu
\bigl(A_1^* \setminus W^1_{\rho+\e}\bigr) + \mu
\bigl(A_2^* \setminus W^2_{\rho+\e}\bigr),
\end{eqnarray*}
where in the last estimate we also used \cite{SteinwartXXb}, (A.1.3).
Repeating this estimate for $\mu(B_2 \symdif A_2^*)$
and using $B_1\cup B_2 \subset L_{D,\rho} \subset M_{\rho-\e
}^{+\delta}$
yields the assertion.
\end{pf*}

\begin{pf*}{Proof of Theorem~\ref{analysis-main2-new}}
Let us fix a $D\in X^n$ with $\Vert\hdd-\hpd\Vert_\infty< \e$.
By the first estimate of \cite{SteinwartXXb}, Theorem A.8.1, we see
that the probability $P^n$ of such a $D$
is not smaller than $1-e^{-\vs}$. In the case of a bounded density and
\eqref{analysis-main2-alter},
the same holds by the second estimate of
\cite{SteinwartXXb}, Theorem A.8.1, and
\begin{eqnarray*}
&& \sqrt{ \frac{6\cpart\Vert h \Vert_\infty\vs+\ln(2\cpart)
-\dpart\ln\delta
}{3\delta
^{\dpart}n} + \biggl(\frac{2\cpart\vs}{3\delta^{\dpart}n} \biggr)^2} +
\frac
{\cpart\vs}{3\delta^{\dpart}n}
\\
&&\qquad\leq \sqrt{ \frac{6\cpart\Vert h \Vert_\infty\vs+\ln(2\cpart)
-\dpart\ln\delta
}{3\delta
^{\dpart}n}} + \frac{2\cpart\vs}{3\delta^{\dpart}n}
\\
&&\qquad\leq \sqrt{\frac{2\cpart(1 + \Vert h \Vert_\infty)(\vs+ \ln(2\cpart)
- \dpart\ln
\delta
)}{\delta^{\dpart}n}} \\
&&\qquad\quad{}+ \frac{2\cpart(\vs+\ln(2\cpart) -\dpart
\ln\delta)}{3\delta
^{\dpart}n},
\end{eqnarray*}
where we use $\ln(2\cpart) \geq\dpart\ln\delta$. Now, \cite{SteinwartXXb}, Lemma
A.8.2,
shows \eqref{generic-inclus}
for all $\rho\geq0$.
Let us check that the remaining assumptions of Theorem~\ref
{analysis-main-combined-new} are also satisfied if $\e^* \leq({\rho^{**}}
-{\rho^*})/9$.
Clearly, we have $\delta\in(0, \dthick]$, $\e\in(0, \e^*]$ and
$\psi(\delta
)< \t$. To show
$\t\leq{\t^*}(\e^*)$ we write
\[
E:= \bigl\{\e'\in\bigl(0,{\rho^{**}}-{\rho^*}\bigr] : \t^*\bigl(
\e'\bigr) \geq\t\bigr\}.
\]
Since we assume $\e^*< \infty$, we obtain
$E\neq\varnothing$ by the definition of $\e^*$.
There thus exists an $\e'\in E$ with $\e' \leq\inf E + \e\leq\e^*$.
Using the monotonicity of ${\t^*}$ established in
\cite{SteinwartXXb}, Theorem A.4.2, we then conclude that $\t\leq{\t^*}
(\e')\leq{\t^*}(\e^*)$, and hence all
assumptions of Theorem~\ref{analysis-main-combined-new} are indeed satisfied.
\end{pf*}

\begin{pf*}{Proof of Theorem~\ref{rates-1}}
Let us begin by checking the conditions of Theorem~\ref{analysis-main2-new}.
Obviously, $\e$ is chosen this way, and
the definition of $\e^*$ together with the assumption
$\e^* \leq({\rho^{**}}-{\rho^*})/9$ yields
%
\begin{equation}
\label{rates-1-h1} (\t/ \csepl)^\k\leq\e^* < {\rho^{**}}-{
\rho^*}.
\end{equation}
By the assumed separation exponent $\k$, we thus find in the case $\k
<\infty$ that
\begin{eqnarray*}
\inf \bigl\{\tilde\e\in\bigl(0,{\rho^{**}}-{\rho^*}\bigr] : \t^*(\tilde\e ) \geq
\t \bigr\} & \leq&\inf \bigl\{\tilde\e\in\bigl(0,{\rho^{**}}-{\rho^*}\bigr]
 :\csepl
\tilde\e ^{1/\k
}\geq\t \bigr\}
\\
&=& (\t/ \csepl)^\k.
\end{eqnarray*}
Consequently, \eqref{estar} holds in the case $\k<\infty$.
Moreover, in the case $\k=\infty$, \eqref{rates-1-h1} together
with ${\rho^{**}}<\infty$ implies $\t\leq\csepl$.
In addition, the separation exponent $\k=\infty$ ensures
$\t^*(\tilde\e)\geq\csepl$ for all $\tilde\e>0$, and hence
we obtain
\[
\e+ \inf \bigl\{\tilde\e\in\bigl(0,{\rho^{**}}-{\rho^*}\bigr] : \t ^*(\tilde\e)
\geq \t \bigr\} = \e\leq\e^* ;
\]
that is, \eqref{estar} is also established in the case $\k= \infty$.
Now, applying Theorem~\ref{analysis-main2-new}, we see that
${\rho^*_D}\in[{\rho^*}+2\e, {\rho^*}+\e^*+5\e]$
with probability $P^n$ not less than $1-e^{-\vs}$;
that is,
\eqref{rates-level-low} is proved. In addition, the definition of $\e
^*$ yields
\[
{\rho^*_D}- {\rho^*}\leq\e^*+5\e\leq(\t/ \csepl)^\k+ 6
\e,
\]
and hence we obtain
\eqref{rates-level-up}.
Let us finally show \eqref{rates-level-exact}. To this end, we first
observe that
Theorem~\ref{analysis-main2-new}
ensures
\begin{eqnarray*}
\t/2\leq\t-\psi(\delta) < 3\t^* \bigl({\rho^*_D}- {\rho^*}+\e\bigr)
&\leq&3\csepu\bigl({\rho^*_D}- {\rho^*}+\e\bigr)^{1/\k}
\\
&<& 3\csepu2^{1/\k} \bigl({\rho^*_D}- {\rho^*}
\bigr)^{1/\k},
\end{eqnarray*}
where in the last step, we use the already established \eqref{rates-level-low}.
By some elementary transformations we conclude that
\[
\frac{1}2 \biggl(\frac{\t}{6\csepu} \biggr)^\k< {
\rho^*_D}-{\rho^*},
\]
and combining this with $2\e\leq{\rho^*_D}-{\rho^*}$, we obtain the
assertion.
\end{pf*}

\begin{pf*}{Proof of Corollary~\ref{rates-cor1}}
We first show \eqref{rates-cor1-hxx} for $\k<\infty$ and
sufficiently large $n$ with the help of Theorem~\ref{rates-1}.
To this end, we define $\e_n^* := \e_n + (\t_n/ \csepl)^\k$ for
$n\geq1$.
Since $(\e_n)$, $(\delta_n)$ and $(\t_n)$ converge to 0, we then
have $\delta
_n\in(0,\dthick]$ and
$\e^*_n\leq({\rho^{**}}-{\rho^*})/9$ for all sufficiently large $n$.
Furthermore, our definitions ensure $\t_n/\delta_n^\g\to\infty$, and
hence we have $\t_n \geq6\cthick\delta_n^\g= 2\psi(\delta_n)$
for all sufficiently large $n$, too.
Before we can apply Theorem~\ref{rates-1}, it thus remains to show~\eqref{analysis-main2-alter}
for sufficiently large $n$.
To this end, we observe that for $\vs_n := \ln n$ and
$\xi_n:=2\cpart(\vs_n+\ln(2\cpart) -\dpart\ln\delta_n)$, we have
\[
\e'_n:=\sqrt{\frac{(1+\Vert h \Vert_\infty)\xi_n}{\delta
_n^{\dpart}n}} +
\frac{\xi_n}{3\delta_n^{\dpart}n} \preceq \biggl( \frac{\ln n}n
\biggr)^{{\g\k}/{(2\g\k+ d)}}.
\]
Using $\e_n \cdot( \frac{\ln n}n )^{-{\g\k}/{(2\g\k+ d)}} \to
\infty$, we then see that
$\e_n\geq\e_n'$ for all sufficiently large $n$.
Now, applying Theorem~\ref{rates-1}, namely
\eqref{rates-level-up},
we obtain an $n_0\geq1$ and a constant $\overline K$
such that \eqref{rates-cor1-hx} holds for all $n\geq n_0$.
Moreover, if $\k$ is exact, \eqref{rates-level-exact} yields a constant
$\underline K$ such that \eqref{rates-cor1-hxx} holds for all $n\geq n_0$.

In the case $\k=\infty$, we first observe that $\e_n^* := \e_n +
(\t_n/
\csepl)^\k$
satisfies $\e_n^* = \e_n$ for all $n$ with $\t_n < \csepl$, that is,
for all sufficiently large $n$.
Moreover, we have $\t_n/\delta_n^\g\to\infty$, and, like the case
$\k
<\infty$, it thus suffices to
show \eqref{analysis-main2-alter} for sufficiently large $n$.
To this end, we observe that for $\vs_n := \ln n$ and $\e'_n$ as above,
we find that, for all sufficiently large $n$,
\[
\e'_n \leq c_2 \biggl( \frac{\ln n \cdot\sqrt{\ln\ln n}}n
\biggr)^{{1}/ 2} \leq\e_n,
\]
where $c_2$ is a suitable constant independent of $n$. Consequently,
\eqref{rates-level-low} and
\eqref{rates-level-up}
yield \eqref{rates-cor1-hxx} for all sufficiently large $n$.
\end{pf*}

\begin{lemma}\label{further-diff-estim-new}
Under the assumptions of Theorem~\ref{analysis-main-combined-new} we have
\begin{eqnarray*}
\sum_{i=1}^2 \mu \bigl(B_i(D)
\symdif A_i^* \bigr) 
& \leq&2 \sum
_{i=1}^2 \mu \bigl( A^i_{\rho^*_D+\e}
\setminus\bigl( A^i_{\rho^*_D+\e} \bigr)^{-\delta} \bigr)
\\
&&{} + \mu \bigl(M_{{\rho^*_D}-\e}^{+\delta}\setminus M_{{\rho^*_D}-\e
} \bigr)
+ \mu \bigl(\bigl\{{\rho^*}< h< {\rho^*_D+\e}\bigr\} \bigr).
\end{eqnarray*}
\end{lemma}

\begin{pf*}{Proof of Lemma~\ref{further-diff-estim-new}}
We
will use inequality \eqref{cluster-chunk-generic-bound}
established in Theorem~\ref{analysis-main-combined-new}. To this end,
we first observe that
\cite{SteinwartXXb}, (A.1.3),
implies
%
\[
\mu \bigl( M_{\rho-\e}^{+\delta}\setminus\bigl\{ h>{\rho^*}\bigr\}
\bigr) 
= \mu \biggl(M_{\rho-\e}^{+\delta}\Bigm\backslash\bigcup
_{\rho'>{\rho
^*}}M_{\rho
'} \biggr) \leq\mu
\bigl(M_{\rho-\e}^{+\delta}\setminus M_{\rho-\e} \bigr).
\]
%
To bound the remaining terms on the right-hand side of \eqref
{cluster-chunk-generic-bound}, we further
observe that
the disjoint relation $A\cap B^{+\delta}= (A\cap(B^{+\delta
}\setminus B)) \cup
(A\cap B)$ applied to $B:= X\setminus A^i_{\rho+\e}$ yields
\begin{eqnarray*}
\mu \bigl(A_i^* \setminus\bigl(A^i_{\rho+\e}
\bigr)^{-\delta} \bigr) & = &\mu \bigl(A_i^* \cap\bigl(X\setminus
A^i_{\rho+\e}\bigr)^{+\delta} \bigr)
\\
& =& \mu \bigl(A_i^* \cap\bigl(X\setminus A^i_{\rho+\e}
\bigr)^{+\delta}\cap A^i_{\rho+\e
} \bigr) + \mu
\bigl(A_i^* \setminus A^i_{\rho+\e}\bigr)
\\
& =& \mu \bigl(A^i_{\rho+\e} \setminus
\bigl(A^i_{\rho+\e}\bigr)^{-\delta
} \bigr) + \mu
\bigl(A_i^* \setminus A^i_{\rho+\e}\bigr).
\end{eqnarray*}
Moreover, $A^i_{\rho+\e} \subset A_i^*$, $A_1^*\cap A_2^*= \varnothing$
together with 
\cite{SteinwartXXb}, (A.1.2) and (A.1.3),
imply
\begin{eqnarray*}
\mu\bigl(A_1^* \setminus A^1_{\rho+\e}\bigr) + \mu
\bigl(A_2^* \setminus A^2_{\rho
+\e}\bigr) &=& \mu
\bigl(\bigl(A_1^*\cup A_2^*\bigr)\setminus
\bigl(A^1_{\rho+\e}\cup A^2_{\rho
+\e
}\bigr)
\bigr)
\\
&= &\mu \bigl(\bigl\{{\rho^*}< h< \rho+\e\bigr\} \bigr).
\end{eqnarray*}
Combining all estimates with \eqref{cluster-chunk-generic-bound}, we
obtain the assertion.
\end{pf*}

\begin{pf*}{Proof of Theorem~\ref{rates-2}}
Since Assumption~\ref{assr} includes the assumptions made in Theorem~\ref{rates-1},
we obtain \eqref{rates-level-low} and \eqref{rates-level-up}.
Furthermore, recall that the proofs of Theorems \ref{rates-1} and \ref
{analysis-main2-new}
show that the probability $P^n$ of having a dataset $D\in X^n$
satisfying the assumptions of Theorem~\ref{analysis-main-combined-new} is not less than $1-e^{-\vs}$. For
such $D$,
Lemma~\ref{further-diff-estim-new} is applicable, and hence we obtain
%
\begin{eqnarray*}
&&\mu\bigl(B_1(D) \symdif A_1^*\bigr) + \mu
\bigl(B_2(D) \symdif A_2^*\bigr)
\\
&&\qquad \leq \mu \bigl(M_{{\rho^*_D}-\e}^{+\delta}\setminus M_{{\rho
^*_D}-\e}
\bigr) + \mu \bigl(\bigl\{{\rho^*}< h< {\rho^*_D+\e}\bigr\} \bigr)
\\
&&\qquad\quad{} + 2\mu \bigl( A^1_{\rho^*_D+\e}\setminus\bigl(
A^1_{\rho^*_D+\e
}\bigr)^{-\delta} \bigr) +2\mu \bigl(
A^2_{\rho^*_D+\e}\setminus\bigl(A^2_{\rho^*_D+\e
}
\bigr)^{-\delta} \bigr)
\\
&&\qquad \leq \mu \bigl(M_{{\rho^*_D}-\e}^{+\delta}\setminus M_{{\rho
^*_D}-\e}
\bigr) + \mu \bigl(\bigl\{0 < h - {\rho^*}< {\rho^*_D}- {\rho^*}+\e
\bigr\} \bigr) + 4\cbound\delta^\alpha,
\end{eqnarray*}
where in the second estimate we use that the clusters have an $\alpha
$-smooth boundary by Assumption~\ref{assr}.
Moreover, the $\alpha$-smooth boundaries also yield
\begin{eqnarray*}
\mu \bigl(M_{{\rho^*_D}-\e}^{+\delta}\setminus M_{{\rho^*_D}-\e
} \bigr) &
\leq&\mu \bigl(\bigl( A^1_{{\rho^*_D}-\e}\bigr)^{+\delta}\setminus
M_{{\rho
^*_D}-\e} \bigr) + \mu \bigl(\bigl( A^2_{{\rho^*_D}-\e}
\bigr)^{+\delta}\setminus M_{{\rho^*_D}-\e
} \bigr)
\\
&\leq&\mu \bigl(\bigl( A^1_{{\rho^*_D}-\e}\bigr)^{+\delta}
\setminus A^1_{{\rho^*_D}-\e} \bigr) + \mu \bigl(\bigl(
A^2_{{\rho^*_D}-\e}\bigr)^{+\delta}\setminus A^2_{{\rho^*_D}-\e
}
\bigr)
\\
& \leq&2\cbound\delta^\alpha.
\end{eqnarray*}
Finally, by \eqref{rates-level-up} and the flatness exponent $\vt$
from Assumption~\ref{assr},
we find
\[
\mu \bigl(\bigl\{0 < h - {\rho^*}< {\rho^*_D}- {\rho^*}+\e\bigr\}
\bigr) \leq \bigl(\cflat\bigl({\rho^*_D}- {\rho^*}+\e\bigr)
\bigr)^\vt \leq \bigl((\t/ \csepl)^\k+ 7\e
\bigr)^\vt. 
\]
%
Combining these three estimates, we then obtain the assertion.
\end{pf*}

\begin{pf*}{Proof of Corollary~\ref{rates-cor2}}
To apply Theorem~\ref{rates-2}
we check that
$\e_n$, $\delta_n$ and $\t_n$ satisfy the assumptions of Theorem~\ref{rates-1}
for $\vs_n := \ln n$ and all sufficiently large $n$. To this end, we
observe that for $\vs_n := \ln n$
and $\xi_n := 2\cpart(\vs_n+\ln(2\cpart) -\dpart\ln\delta_n)$,
we have
\[
\e'_n:=\sqrt{\frac{ (1+\Vert h \Vert_\infty)\xi_n}{\delta
_n^{\dpart}n}} +
\frac{\xi_n}{3\delta_n^{\dpart}n} \preceq
\biggl( \frac{\ln n}n \biggr)^{{\vr}/{(2\vr+ \vt d)}} ( \ln
\ln n )^{-{\vt d}/{(4\vr+ 2\vt d)}}. 
\]
Using $\e_n \cdot( \frac{\ln n}n )^{-{\vr}/{(2\vr+ \vt d)}} (
\ln
\ln n )^{{\vt d}/{(4\vr+ 2\vt d)}} \to\infty$, we then see that
$\e_n\geq\e_n'$ for all sufficiently large $n$.
Moreover, the remaining conditions on $\e_n$, $\delta_n$ and $\t_n$ from
Theorem~\ref{rates-1}
are clearly satisfied for all
sufficiently large $n$, and hence we can apply Theorem~\ref{rates-2}
for such $n$.
This yields
\[
\mu \bigl(B_1(D) \symdif A_1^* \bigr) + \mu
\bigl(B_2(D) \symdif A_2^* \bigr) \leq6\cbound
\delta_n^\alpha+ \bigl(\cflat( \t_n/
\csepl)^\k+ 7\cflat\e_n \bigr)^\vt
\]
with probability $P^n$ not smaller than $1-1/n$ for all sufficiently
large $n$.
Some elementary calculations then show that there is a $K$ with
\begin{eqnarray*}
&&P^n \biggl( D : \mu \bigl(B_1(D) \symdif
A_1^* \bigr) + \mu \bigl(B_2(D) \symdif A_2^*
\bigr) \leq K \biggl( \frac{\ln n \cdot(\ln\ln n)^2}n \biggr)^{{\vt
\vr
}/{(2\vr+ \vt d)}} \biggr)\\
&&\qquad \geq1 -
\frac{1} n
\end{eqnarray*}
for
all sufficiently large $n$. Moreover, since we always have
\[
\mu\bigl(B_1(D) \symdif A_1^*\bigr) + \mu
\bigl(B_2(D) \symdif A_2^*\bigr)\leq2\mu (X)<\infty,
\]
it is an easy exercise to suitably increase $K$ such that the desired
inequality actually holds for all $n\geq1$.
\end{pf*}

\begin{pf*}{Proof of Theorem~\ref{adaptive-level}}
First observe that $C^2\ln(\ln n)\geq2(1+\Vert h \Vert_\infty)$
guarantees that
all $\e_{\delta,n}$ satisfy
\eqref{analysis-main2-alter} for $\vs':= \vs+\ln|\D|$. Consequently,
Theorem~\ref{rates-1}, namely \eqref{rates-level-low} and \eqref
{rates-level-up}, yields
\[
P^n \bigl( \bigl\{D\in X^n: \e_{\delta,n} < {
\rho^*_{D,\delta}}- {\rho^*}\leq\bigl(\t_{\delta,n}^\g/
\csepl\bigr)^\k+ 6\e_{\delta,n} \bigr\} \bigr)
\geq1-e^{-\vs- \ln|\D|}
\]
for all $\delta\in\D$. Applying the union bound, we thus find
\[
P^n \bigl(D\in X^n: \e_{\delta,n} < {
\rho^*_{D,\delta}}- {\rho ^*}\leq \bigl(\t_{\delta,n}^\g/
\csepl \bigr)^\k+ 6\e_{\delta,n} \mbox{ for all } \delta\in\D
\bigr) \geq1-e^{-\vs}.
\]
Let us now consider a $D\in X^n$ such that
$\e_{\delta,n} < {\rho^*_{D,\delta}}- {\rho^*}\leq(\t_{\delta
,n}^\g
/ \csepl)^\k+ 6\e_{\delta
,n}$ for all $\delta\in\D$.
Then the definitions of ${\rho^*_{D,\D}}$ and $\e_{D,\D}$ [see
\eqref
{def-rsdd}] imply
\[
{\rho^*_{D,\D}}- {\rho^*} = \min_{\delta\in\D} {
\rho^*_{D,\delta}}- {\rho^*} \in \Bigl(\min_{\delta\in\D}
\e_{\delta,n}, \min_{\delta\in\D
} \bigl(\bigl(
\t_{\delta
,n}^\g/ \csepl\bigr)^\k+ 6
\e_{\delta,n} \bigr) \Bigr]
\]
and $\e_{D,\D}=\e_{\delta^*_{D,\D},n} < \rho^*_{D,\delta^*_{D,\D
}} - {\rho^*}= {\rho^*_{D,\D}}
- {\rho^*}$;
that is, we have shown the first assertion. To show the remaining
assertions, we first observe that a literal
repetition of the argument above, in which we only replace the use of
\eqref{rates-level-low}
by that of \eqref{rates-level-exact}, yields
\[
P^n \bigl(D\in X^n: c_1
\t_{\delta,n}^\k+ \e_{\delta,n} < {\rho ^*_{D,\delta}}-
{\rho^*}\leq c_2 \t_{\delta,n}^\k+ 6
\e_{\delta,n} \mbox{ for all } \delta\in \D \bigr) \geq1-e^{-\vs}.
\]
Using \eqref{def-rsdd} we then immediately obtain the second assertion,
while considering $\delta=\delta_{D,\D}^*$
gives the third assertion.
%
\end{pf*}

\begin{pf*}{Proof of Corollary~\ref{adaptive-level-cor}}
Let us fix an $n\geq16$. For later use we note that this choice
implies $I_n\subset(0,1]$.
Our first goal is to show that we can apply Theorem~\ref
{adaptive-level} for sufficiently large $n$.
To this end, we first observe that $\max\D_n= (\ln\ln n)^{-\dpart}
\to
0$ for $n\to\infty$,
and hence we obtain $\D_n\subset(0,\dthick]$ for all sufficiently
large $n$.
Analogously, $\max\D_n\ln\ln\ln n\to0$ implies
$\max_{\delta\in\D_n}(\t_{\delta,n}/ \csepl)^\k\leq({\rho
^{**}}-{\rho^*})/18$
for all sufficiently large $n$,
and the definition of $\t_{\delta,n}$ ensures $\min_{\delta\in\D
_n}\t_{\delta,n}
\geq2\psi(\delta)$ for all sufficiently
large $n$. Let us now show that eventually we also have $\max_{\delta
\in\D
_n} \e_{\delta,n} \leq({\rho^{**}}-{\rho^*})/18$.
To this end, note that the derivative of $g_n:(0,\infty)\to\R$
defined by
\[
g_n(\delta) := \frac{ \ln(2\cpart|\D_n|n) -\dpart\ln\delta
}{\delta^{d}n}
\]
is given by
\[
g_n'(\delta) = - \frac{\dpart(1+ \ln(2\cpart|\D_n|n) -\dpart\ln
\delta)
}{\delta^{1+\dpart}n},
\]
and using $\cpart\geq1$, we thus find that $g_n$ is monotonically
decreasing on $(0,1]$ for all $n\geq1$.
In addition, using $|\D_n|\leq n$ we obtain
\begin{eqnarray*}
g_n(\min I_n) &=& g_n \biggl( \biggl(
\frac{\ln n \cdot(\ln\ln n)^2}n \biggr)^{
{1}/{d}} \biggr)
\\
&=& \frac{\ln(2\cpart|\D_n|n) +\ln n - \ln\ln n - 2\ln\ln\ln n}{\ln
n\cdot(\ln\ln n)^2}
\\
&\leq& \frac{4\ln n - \ln\ln n - 2\ln\ln\ln n}{\ln n\cdot(\ln\ln n)^2}
\\
& \leq& \frac{4} {(\ln\ln n)^2}
\end{eqnarray*}
for all $n\geq\max\{16, 2\cpart\}$, and hence $g_n(\min I_n)\ln\ln
n\to0$ for $n\to\infty$.
Since the definition of $\e_{\delta,n}$ gives
$\e_{\delta,n} = C \sqrt{\cpart g_n(\delta)\ln\ln n} + \frac{2} 3
\cpart g_n(\delta)$,
we can thus conclude that
\begin{eqnarray*}
\max_{\delta\in\D_n} \e_{\delta,n} &\leq& \max_{\delta\in\D_n}
C \sqrt{\cpart g_n(\delta)\ln\ln n} + \max_{\delta\in\D_n}
\cpart g_n(\delta)
\\
&\leq& C \sqrt{\cpart g_n(\min I_n)\ln\ln n} + \cpart
g_n(\min I_n) \to0 
\end{eqnarray*}
for $n\to\infty$. This ensures the desired $\max_{\delta\in\D_n}
\e_{\delta
,n} \leq({\rho^{**}}-{\rho^*})/18$ for all
sufficiently large $n$. Combining this with our previous estimate, we find
\[
\max_{\delta\in\D_n} \bigl( (\t_{\delta,n}/ \csepl)^\k+
\e _{\delta,n} \bigr) \leq \bigl({\rho^{**}}-{\rho^*}\bigr)/9
\]
for all sufficiently large $n$, and
thus we can apply Theorem~\ref{adaptive-level} for such $n$.

Before we proceed, let us now fix an $n\geq16$ and assume that without
loss of generality that $\D_n$ is of
the form $\D= \{\delta_1,\ldots,\delta_m\}$ with $\delta
_{i-1}<\delta_{i}$ for all
$i=2,\ldots,m$.
We write $\delta_0 := \min I_n$ and $\delta_{m+1}:= \max I_n$.
Our intermediate goal is to show that
%
\begin{equation}
\label{concbas:appro-approx-fkt-h1} \delta_{i} - \delta_{i-1}
\Leq2n^{-1/\dpart},\qquad i=1,\ldots,m+1.
\end{equation}
To this end, we fix an $i\in\{1,\ldots,m\}$ and write $\bar\delta:=
(\delta
_{i}+\delta_{i-1})/2\in I_n$.
Since $\D_n$ is an $n^{-1/\dpart}$-net of $I_n$, we then have $\delta
_{i} -
\bar\delta\leq n^{-1/\dpart}$
or $\bar\delta- \delta_{i-1} \leq n^{-1/\dpart}$, and from
both,
(\ref{concbas:appro-approx-fkt-h1}) follows. Moreover, to show \eqref
{concbas:appro-approx-fkt-h1}
in the case $i=m+1$, we first observe that there exists an $\delta
_i\in\D
_n$ with $\delta_i-\delta_m\leq n^{-1/\dpart}$
since $\D_n$ is an $n^{-1/\dpart}$-net of $I_n$. Using our ordering of
$\D_n$, we can assume without loss of
generality that $i=m$, which immediately implies \eqref
{concbas:appro-approx-fkt-h1}.

We now prove the first assertion in the case
$\k<\infty$. To this end, we write
\[
\delta^*_n:= \biggl( \frac{\ln n \cdot\ln\ln n}n \biggr)^{
{1}/{(2\g\k+ \dpart)}},
\]
where we note that for sufficiently large $n$ we have $\delta^*_n \in I_n$.
In the following
we thus restrict our considerations to such $n$.
Then there exists an index $i\in\{1,\ldots,m+1\}$ such that $\delta
_{i-1}\leq\delta^*_n\leq\delta_i$,
and by (\ref{concbas:appro-approx-fkt-h1}) we conclude that
$\delta^*_n \leq\delta_i\leq\delta^*_n+2n^{-1/\dpart}$.
Clearly, this yields
%
\begin{eqnarray}\label{par-sel-cor-h1}
\min_{\delta\in\D_n} \bigl(c_2
\t_{\delta,n}^\k+ 6\e_{\delta
,n} \bigr) &=& \min
_{\delta\in\D_n} \bigl(c_2\delta^{\g\k}(\ln\ln\ln
n)^\k+ 6\e_{\delta,n} \bigr)
\nonumber\\
&\leq& c_2\delta_i^{\g\k}(\ln\ln\ln
n)^\k+ 6\e_{\delta_i,n}
\nonumber
\\[-8pt]
\\[-8pt]
\nonumber
&\leq& c_2\bigl(\delta^*_n+2n^{-1/\dpart}
\bigr)^{\g\k}(\ln\ln\ln n)^\k+ 6\e _{\delta
_i,n}
\\
 & \leq&6c_2 \biggl( \frac{\ln n \cdot(\ln\ln n)^2}n
\biggr)^{
{1}/{(2\g\k+ \dpart)}} + 6\e_{\delta_i,n}\nonumber
\end{eqnarray}
for all sufficiently large $n$, where $c_2 := \csepl^{-\k}$ is the
constant from Theorem~\ref{adaptive-level}.
Moreover, using $|\D_n|\leq n$ and the monotonicity of $g_n$, we
further obtain
%
%
\begin{eqnarray}\label{par-sel-cor-h2}
g_n(\delta_i) &\leq &g_n\bigl(
\delta_n^*\bigr) = \frac{ \ln(2\cpart|\D_n|n) -\dpart\ln\delta_n^* }{(\delta
_n^*)^{\dpart}n} \leq\frac{ \ln(2\cpart) + 2\ln n -\dpart\ln\delta_n^* }
{(\delta
_n^*)^{\dpart}n}\nonumber
\\
&\leq&\frac{ 4\ln n }{(\delta_n^*)^{\dpart}n}
\\
& \leq&\frac{4} {(\ln\ln n)^{{\dpart}/{(2\g\k+ \dpart)}} } \cdot
 \biggl( \frac{\ln n }n
\biggr)^{{2\g\k}/{(2\g\k+ \dpart)}}\nonumber
\end{eqnarray}
for all sufficiently large $n$.
By the relation between $\e_{\delta,n}$ and $g_n(\delta)$, we then find
\[
\e_{\delta_i,n} \leq2C\sqrt{\cpart} \biggl( \frac{\ln n \cdot\ln\ln n}n
\biggr)^{{\g\k}/{(2\g\k+ \dpart)}}
+3\cpart \biggl( \frac{\ln n }n \biggr)^{{2\g\k}/{(2\g\k+\dpart)}},
\]
and combining this estimate with \eqref{par-sel-cor-h1} and Theorem~\ref
{adaptive-level},
we obtain the first assertion in the case $\k<\infty$.

Let us now consider the case $\k=\infty$.
To this end, we fix an $n$
such that
\[
\delta_n^* := \biggl(\frac{1} {\ln\ln n} \biggr)^{{1}/\dpart}
\]
satisfies
$(\delta_n^* +2n^{-1/\dpart})^\g\ln\ln\ln n < \csepl$, and thus
\[
\bigl(\bigl(\delta_n^*+2n^{-1/\dpart}\bigr)^\g\ln\ln
\ln n/ \csepl\bigr)^\k= 0.
\]
Since $\delta_n^*\in I_n$,
there also exists an index $i\in\{1,\ldots,m+1\}$ such that $\delta
_{i-1}\leq\delta^*\leq\delta_i$, and by (\ref{concbas:appro-approx-fkt-h1})
we again conclude
$\delta^* \leq\delta_i\leq\delta^*+2n^{-1/\dpart}$. Clearly, the
latter implies
\begin{eqnarray*}
\min_{\delta\in\D_n} \bigl((\t_{\delta,n}/ \csepl)^\k+
6\e _{\delta,n} \bigr) &\leq&\bigl(\delta_i^\g\ln\ln
\ln n/ \csepl\bigr)^\k+ 6\e_{\delta_i,n}
\\
&\leq& \bigl( \bigl(\delta_n^*+2n^{-1/\dpart}\bigr)^\g
\ln\ln\ln n/ \csepl \bigr)^\k+ 6\e_{\delta_i,n}
\\
&=& 6\e_{\delta_i,n}
\end{eqnarray*}
by our assumptions on $n$.
Analogously to \eqref{par-sel-cor-h2} we further find,
for
sufficiently large $n$, that
\[
g_n(\delta_i) \leq g_n\bigl(
\delta_n^*\bigr) \leq\frac{ 3\ln n - \dpart\ln\delta_n^*}{(\delta_n^*)^{\dpart}n} = \frac{3\ln n + \ln\ln\ln n}{n (\ln\ln n)^{-1}} \leq4
\frac{\ln n \cdot\ln\ln n}{n},
\]
and by the relation between $\e_{\delta,n}$ and $g(\delta)$, we then
find the
assertion with the help of
Theorem~\ref{adaptive-level}.

Let us finally prove the second assertion. To this end we first recall
that we have already seen that for
sufficiently large $n$, we can apply Theorem~\ref{adaptive-level}. Thus
it suffices to find a lower bound for
the right-hand side of
%
\begin{equation}
\label{adaptive-level-cor-ph1} \min_{\delta\in\D_n} \bigl(c_1
\t_{\delta,n}^\k+ \e_{\delta
,n} \bigr) \geq \min\{1,
c_1 \} \cdot\min_{\delta\in\D_n} \bigl(
\t_{\delta
,n}^\k+ \e_{\delta
,n} \bigr) ,
\end{equation}
where $c_1$ is the constant appearing in Theorem~\ref{adaptive-level}.
Now, for $n\geq16$, we have $I_n\subset(0,1]$, and thus we find
$\delta\in(0,1]$ for all $\delta\in\D_n$. For sufficiently large
$n$ this yields
\begin{eqnarray*}
&&\min_{\delta\in\D_n} \bigl( \t_{\delta,n}^\k+
\e_{\delta,n} \bigr)
\\
&&\qquad = \min_{\delta\in\D_n} \biggl( \delta^{\g\k}(\ln\ln\ln
n)^\k+ C \sqrt {\cpart g_n(\delta)\ln\ln n} +
\frac{2} 3 \cpart g_n(\delta) \biggr)
\\
&&\qquad \geq \min_{\delta\in\D_n} \bigl( \delta^{\g\k} + C \sqrt {
\cpart g_n(\delta )\ln\ln n} \bigr)
\\
&& \qquad \geq\min_{\delta\in\D_n} \biggl( \delta^{\g\k} + C \sqrt {
\frac{\cpart
\ln n \cdot\ln\ln n }{\delta^dn}} \biggr)
\\
&&\qquad \geq \min_{\delta\in(0,1]} \biggl( \delta^{\g\k} + C \sqrt {
\frac{\cpart
\ln n\cdot\ln\ln n }{\delta^dn}} \biggr).
\end{eqnarray*}
An elementary application of calculus then yields the assertion.
\end{pf*}

\begin{pf*}{Proof of Corollary~\ref{rates-adaptive-cor2}}
As in the proof of Corollary~\ref{rates-cor2} it suffices to show the
assertion for sufficiently large $n$.
Now,
we have seen in the proof of Corollary~\ref{adaptive-level-cor} that
for sufficiently large $n$, Inequality
\eqref{adaptive-level-cor-hx} follows from the fact that the procedure
satisfies the
assumptions of Theorem~\ref{adaptive-level} for such $n$ and $\vs:=
\ln n$.
Consequently, for sufficiently large $n$, the probability $P^n$ of
having a data set $D\in X^n$ satisfying
both \eqref{adaptive-level-cor-hx} and the third inequality of Theorem~\ref{adaptive-level} is not less than $1-1/n$.
Let us fix such a $D$. Then we have
%
\begin{equation}
\label{par-sel-cor2-h1} c_1 \t_{D,\D}^{\k} +
\e_{D,\D} \leq\rho^*_{D,\D_n}-{\rho^*}\leq \overline K \biggl(
\frac{\ln n \cdot(\ln\ln n)^2}n \biggr)^{
{\g\k}/{(2\g\k+ \dpart)}}.
\end{equation}
Moreover, an elementary estimate yields
\[
c_1 \t_{D,\D}^{\k} + \e_{D,\D} \geq\min
\bigl\{1/7, c_1 \csepl^\k\bigr\} \cdot \bigl((
\t_{D,\D}/ \csepl)^\k+ 7 \e_{D,\D} \bigr),
\]
and setting $c:= \min\{1/7, c_1 \csepl^\k\}$, we hence obtain
%
\begin{equation}
\label{par-sel-cor2-h2} ( \t_{D,\D}/ \csepl)^\k+ 7 \e_{D,\D}
\leq c^{-1} \overline K \biggl( \frac{\ln n \cdot(\ln\ln n)^2}n
\biggr)^{{\g\k}/{(2\g\k+\dpart)}}.
\end{equation}
In addition, for sufficiently large $n$, inequality \eqref
{par-sel-cor2-h1} implies
%
\begin{equation}
\label{par-sel-cor2-h3} \delta_{D,\D}^* \leq\t_{D,\D}^{1/\g}
\leq (4\overline K)^{{1} /{\g\k}} ( 6 \csepu )^{{1}/ \g} \biggl(
\frac{\ln n \cdot(\ln\ln n)^2}n \biggr)^{{1}/{(2\g\k+
\dpart)}}.
\end{equation}
Now we have already seen in the proofs of Theorem~\ref{adaptive-level}
and Corollary~\ref{adaptive-level-cor}
that for sufficiently large $n$, the assumptions on $\delta$, $\e
_{\delta,n}$,
$\e_{\delta,n}^*$, $\t_n$, $\vs_n:= \ln n$ and $n$ of Theorem~\ref
{rates-1} are satisfied for all $\delta\in\D_n$
simultaneously. We can thus combine
\eqref{par-sel-cor2-h2} and \eqref{par-sel-cor2-h3} with
Theorem~\ref{rates-2} to obtain the assertion.
\end{pf*}


\begin{supplement}[id=suppA]
\stitle{Supplement to ``Fully adaptive density-based clustering''\\}
\slink[doi]{10.1214/15-AOS1331SUPP} 
\sdatatype{.pdf}
\sfilename{aos1331\_supp.pdf}
\sdescription{We provide two appendices A and B.
In Appendix A,
several auxiliary results, which are partially taken from \cite{Steinwart11a},
are presented, and
the assumptions made in the paper are discussed in more detail.
In Appendix B, we present a couple of two-dimensional examples that show
that the assumptions imposed in the paper are not only met by many
discontinuous densities, but also
by many continuous densities.}
\end{supplement}

%

%
%

\begin{thebibliography}{30}

\bibitem{BaCACu01a}
%
\begin{barticle}[mr]
\bauthor{\bsnm{Ba{\'{\i}}llo},~\bfnm{Amparo}\binits{A.}},
\bauthor{\bsnm{Cuesta-Albertos},~\bfnm{Juan~A.}\binits{J.~A.}} \AND
\bauthor{\bsnm{Cuevas},~\bfnm{Antonio}\binits{A.}}
(\byear{2001}).
\btitle{Convergence rates in nonparametric estimation of level sets}.
\bjournal{Statist. Probab. Lett.}
\bvolume{53}
\bpages{27--35}.
\bid{doi={10.1016/S0167-7152(01)00006-2}, issn={0167-7152}, mr={1843338}}
\end{barticle}
%

\bptok{imsref}%
\endbibitem

\bibitem{BaCuJu00a}
%
\begin{barticle}[mr]
\bauthor{\bsnm{Ba{\'{\i}}llo},~\bfnm{Amparo}\binits{A.}},
\bauthor{\bsnm{Cuevas},~\bfnm{Antonio}\binits{A.}} \AND
\bauthor{\bsnm{Justel},~\bfnm{Ana}\binits{A.}}
(\byear{2000}).
\btitle{Set estimation and nonparametric detection}.
\bjournal{Canad. J. Statist.}
\bvolume{28}
\bpages{765--782}.
\bid{doi={10.2307/3315915}, issn={0319-5724}, mr={1821433}}
\end{barticle}
%

\bptok{imsref}%
\endbibitem

\bibitem{BDLi97a}
%
\begin{barticle}[mr]
\bauthor{\bsnm{Ben-David},~\bfnm{Shai}\binits{S.}} \AND
\bauthor{\bsnm{Lindenbaum},~\bfnm{Michael}\binits{M.}}
(\byear{1997}).
\btitle{Learning distributions by their density levels: A paradigm for
learning without a teacher}.
\bjournal{J. Comput. System Sci.}
\bvolume{55}
\bpages{171--182}.
\bid{doi={10.1006/jcss.1997.1507}, issn={0022-0000}, mr={1473058}}
\end{barticle}
%

\bptok{imsref}%
\endbibitem

\bibitem{ChaconXXa}
%
\begin{bmisc}[auto:parserefs-M02]
\bauthor{\bsnm{Cha\'on},~\bfnm{J.~C.}\binits{J.~C.}}
(\byear{2014}).
\bhowpublished{A population background for nonparametric density-based
clustering.
Technical report.
Available at \arxivurl{arXiv:1408.1381}}.
\end{bmisc}
%

\bptok{imsref}%
\endbibitem

\bibitem{ChDa10a}
%
\begin{bincollection}[auto:parserefs-M02]
\bauthor{\bsnm{Chaudhuri},~\bfnm{K.}\binits{K.}} \AND
\bauthor{\bsnm{Dasgupta},~\bfnm{S.}\binits{S.}}
(\byear{2010}).
\btitle{Rates of convergence for the cluster tree}.
In \bbooktitle{Advances in Neural Information Processing Systems}
\bvolume{23}
(\beditor{\bfnm{J.}\binits{J.}~\bsnm{Lafferty}},
\beditor{\bfnm{C.~K.~I.}\binits{C.~K.~I.}~\bsnm{Williams}},
\beditor{\bfnm{J.}\binits{J.}~\bsnm{Shawe-Taylor}},
\beditor{\bfnm{R.~S.}\binits{R.~S.}~\bsnm{Zemel}} \AND
\beditor{\bfnm{A.}\binits{A.}~\bsnm{Culotta}}, eds.)
\bpages{343--351}.
\end{bincollection}
%

\bptok{imsref}%
\endbibitem

\bibitem{CuFr97a}
%
\begin{barticle}[mr]
\bauthor{\bsnm{Cuevas},~\bfnm{Antonio}\binits{A.}} \AND
\bauthor{\bsnm{Fraiman},~\bfnm{Ricardo}\binits{R.}}
(\byear{1997}).
\btitle{A plug-in approach to support estimation}.
\bjournal{Ann. Statist.}
\bvolume{25}
\bpages{2300--2312}.
\bid{doi={10.1214/aos/1030741073}, issn={0090-5364}, mr={1604449}}
\end{barticle}
%

\bptok{imsref}%
\endbibitem

\bibitem{DeWi80a}
%
\begin{barticle}[mr]
\bauthor{\bsnm{Devroye},~\bfnm{Luc}\binits{L.}} \AND
\bauthor{\bsnm{Wise},~\bfnm{Gary~L.}\binits{G.~L.}}
(\byear{1980}).
\btitle{Detection of abnormal behavior via nonparametric estimation of
the support}.
\bjournal{SIAM J. Appl. Math.}
\bvolume{38}
\bpages{480--488}.
\bid{doi={10.1137/0138038}, issn={0036-1399}, mr={0579432}}
\end{barticle}
%

\bptok{imsref}%
\endbibitem

\bibitem{Donoho88a}
%
\begin{barticle}[mr]
\bauthor{\bsnm{Donoho},~\bfnm{David~L.}\binits{D.~L.}}
(\byear{1988}).
\btitle{One-sided inference about functionals of a density}.
\bjournal{Ann. Statist.}
\bvolume{16}
\bpages{1390--1420}.
\bid{doi={10.1214/aos/1176351045}, issn={0090-5364}, mr={0964930}}
\end{barticle}
%

\bptok{imsref}%
\endbibitem

\bibitem{GiGu02a}
%
\begin{barticle}[mr]
\bauthor{\bsnm{Gin{\'e}},~\bfnm{Evarist}\binits{E.}} \AND
\bauthor{\bsnm{Guillou},~\bfnm{Armelle}\binits{A.}}
(\byear{2002}).
\btitle{Rates of strong uniform consistency for multivariate kernel
density estimators}.
\bjournal{Ann. Inst. Henri Poincar\'e Probab. Stat.}
\bvolume{38}
\bpages{907--921}.
\bid{doi={10.1016/S0246-0203(02)01128-7}, issn={0246-0203}, mr={1955344}}
\end{barticle}
%

\bptok{imsref}%
\endbibitem

\bibitem{Hartigan75}
%
\begin{bbook}[mr]
\bauthor{\bsnm{Hartigan},~\bfnm{John~A.}\binits{J.~A.}}
(\byear{1975}).
\btitle{Clustering Algorithms}.
\bpublisher{Wiley},
\blocation{New York}.
\bid{mr={0405726}}
\end{bbook}
%

\bptok{imsref}%
\endbibitem

\bibitem{Hartigan81a}
%
\begin{barticle}[mr]
\bauthor{\bsnm{Hartigan},~\bfnm{J.~A.}\binits{J.~A.}}
(\byear{1981}).
\btitle{Consistency of single linkage for high-density clusters}.
\bjournal{J. Amer. Statist. Assoc.}
\bvolume{76}
\bpages{388--394}.
\bid{issn={0162-1459}, mr={0624340}}
\end{barticle}
%

\bptok{imsref}%
\endbibitem

\bibitem{Hartigan87a}
%
\begin{barticle}[mr]
\bauthor{\bsnm{Hartigan},~\bfnm{J.~A.}\binits{J.~A.}}
(\byear{1987}).
\btitle{Estimation of a convex density contour in two dimensions}.
\bjournal{J. Amer. Statist. Assoc.}
\bvolume{82}
\bpages{267--270}.
\bid{issn={0162-1459}, mr={0883354}}
\end{barticle}
%

\bptok{imsref}%
\endbibitem

\bibitem{KpLu11a}
%
\begin{binproceedings}[auto:parserefs-M02]
\bauthor{\bsnm{Kpotufe},~\bfnm{S.}\binits{S.}} \AND
\bauthor{\bparticle{von} \bsnm{Luxburg},~\bfnm{U.}\binits{U.}}
(\byear{2011}).
\btitle{Pruning nearest neighbor cluster trees}.
In \bbooktitle{Proceedings of the 28th International Conference on
Machine Learning}
(\beditor{\bfnm{L.}\binits{L.}~\bsnm{Getoor}} \AND
\beditor{\bfnm{T.}\binits{T.}~\bsnm{Scheffer}}, eds.)
\bpages{225--232}.
\bpublisher{ACM},
\blocation{New York}.
\end{binproceedings}
%

\bptok{imsref}%
\endbibitem

\bibitem{MaHeLu09a}
%
\begin{barticle}[mr]
\bauthor{\bsnm{Maier},~\bfnm{Markus}\binits{M.}},
\bauthor{\bsnm{Hein},~\bfnm{Matthias}\binits{M.}} \AND
\bauthor{\bparticle{von} \bsnm{Luxburg},~\bfnm{Ulrike}\binits{U.}}
(\byear{2009}).
\btitle{Optimal construction of {$k$}-nearest-neighbor graphs for
identifying noisy clusters}.
\bjournal{Theoret. Comput. Sci.}
\bvolume{410}
\bpages{1749--1764}.
\bid{doi={10.1016/j.tcs.2009.01.009}, issn={0304-3975}, mr={2514706}}
\end{barticle}
%

\bptok{imsref}%
\endbibitem

\bibitem{MuSa91a}
%
\begin{barticle}[mr]
\bauthor{\bsnm{M{\"u}ller},~\bfnm{D.~W.}\binits{D.~W.}} \AND
\bauthor{\bsnm{Sawitzki},~\bfnm{G.}\binits{G.}}
(\byear{1991}).
\btitle{Excess mass estimates and tests for multimodality}.
\bjournal{J. Amer. Statist. Assoc.}
\bvolume{86}
\bpages{738--746}.
\bid{issn={0162-1459}, mr={1147099}}
\end{barticle}
%

\bptok{imsref}%
\endbibitem

\bibitem{Polonik95a}
%
\begin{barticle}[mr]
\bauthor{\bsnm{Polonik},~\bfnm{Wolfgang}\binits{W.}}
(\byear{1995}).
\btitle{Measuring mass concentrations and estimating density contour
clusters---An excess mass approach}.
\bjournal{Ann. Statist.}
\bvolume{23}
\bpages{855--881}.
\bid{doi={10.1214/aos/1176324626}, issn={0090-5364}, mr={1345204}}
\end{barticle}
%

\bptok{imsref}%
\endbibitem

\bibitem{Rigollet07a}
%
\begin{barticle}[mr]
\bauthor{\bsnm{Rigollet},~\bfnm{Philippe}\binits{P.}}
(\byear{2007}).
\btitle{Generalized error bounds in semi-supervised classification
under the cluster assumption}.
\bjournal{J. Mach. Learn. Res.}
\bvolume{8}
\bpages{1369--1392}.
\bid{issn={1532-4435}, mr={2332435}}
\end{barticle}
%

\bptok{imsref}%
\endbibitem

\bibitem{RiVe09a}
%
\begin{barticle}[mr]
\bauthor{\bsnm{Rigollet},~\bfnm{Philippe}\binits{P.}} \AND
\bauthor{\bsnm{Vert},~\bfnm{R{\'e}gis}\binits{R.}}
(\byear{2009}).
\btitle{Optimal rates for plug-in estimators of density level sets}.
\bjournal{Bernoulli}
\bvolume{15}
\bpages{1154--1178}.
\bid{doi={10.3150/09-BEJ184}, issn={1350-7265}, mr={2597587}}
\end{barticle}
%

\bptok{imsref}%
\endbibitem

\bibitem{RiSiNuWa12a}
%
\begin{barticle}[mr]
\bauthor{\bsnm{Rinaldo},~\bfnm{Alessandro}\binits{A.}},
\bauthor{\bsnm{Singh},~\bfnm{Aarti}\binits{A.}},
\bauthor{\bsnm{Nugent},~\bfnm{Rebecca}\binits{R.}} \AND
\bauthor{\bsnm{Wasserman},~\bfnm{Larry}\binits{L.}}
(\byear{2012}).
\btitle{Stability of density-based clustering}.
\bjournal{J. Mach. Learn. Res.}
\bvolume{13}
\bpages{905--948}.
\bid{issn={1532-4435}, mr={2930628}}
\end{barticle}
%

\bptok{imsref}%
\endbibitem

\bibitem{RiWa10a}
%
\begin{barticle}[mr]
\bauthor{\bsnm{Rinaldo},~\bfnm{Alessandro}\binits{A.}} \AND
\bauthor{\bsnm{Wasserman},~\bfnm{Larry}\binits{L.}}
(\byear{2010}).
\btitle{Generalized density clustering}.
\bjournal{Ann. Statist.}
\bvolume{38}
\bpages{2678--2722}.
\bid{doi={10.1214/10-AOS797}, issn={0090-5364}, mr={2722453}}
\end{barticle}
%

\bptok{imsref}%
\endbibitem

\bibitem{ScHuSt05a}
%
\begin{barticle}[mr]
\bauthor{\bsnm{Scovel},~\bfnm{Clint}\binits{C.}},
\bauthor{\bsnm{Hush},~\bfnm{Don}\binits{D.}} \AND
\bauthor{\bsnm{Steinwart},~\bfnm{Ingo}\binits{I.}}
(\byear{2005}).
\btitle{Learning rates for density level detection}.
\bjournal{Anal. Appl. (Singap.)}
\bvolume{3}
\bpages{357--371}.
\bid{doi={10.1142/S0219530505000625}, issn={0219-5305}, mr={2181253}}
\bptnote{check pages}%
\end{barticle}
%

\bptok{imsref}%
\endbibitem

\bibitem{SiScNo09a}
%
\begin{barticle}[mr]
\bauthor{\bsnm{Singh},~\bfnm{Aarti}\binits{A.}},
\bauthor{\bsnm{Scott},~\bfnm{Clayton}\binits{C.}} \AND
\bauthor{\bsnm{Nowak},~\bfnm{Robert}\binits{R.}}
(\byear{2009}).
\btitle{Adaptive {H}ausdorff estimation of density level sets}.
\bjournal{Ann. Statist.}
\bvolume{37}
\bpages{2760--2782}.
\bid{doi={10.1214/08-AOS661}, issn={0090-5364}, mr={2541446}}
\end{barticle}
%

\bptok{imsref}%
\endbibitem

\bibitem{SrSt12a}
%
\begin{binproceedings}[auto:parserefs-M02]
\bauthor{\bsnm{Sriperumbudur},~\bfnm{B.~K.}\binits{B.~K.}} \AND
\bauthor{\bsnm{Steinwart},~\bfnm{I.}\binits{I.}}
(\byear{2012}).
\btitle{Consistency and rates for clustering with {DBSCAN}}.
In \bbooktitle{Proceedings of the 15th International Conference on
Artificial Intelligence and Statistics 2012}
(\beditor{\bfnm{N.}\binits{N.}~\bsnm{Lawrence}} \AND
\beditor{\bfnm{M.}\binits{M.}~\bsnm{Girolami}}, eds.).
\bseries{JMLR Workshop and Conference Proceedings}
\bvolume{22}
\bpages{1090--1098}.
\end{binproceedings}
%

\bptok{imsref}%
\endbibitem

\bibitem{Steinwart11a}
%
\begin{binproceedings}[auto:parserefs-M02]
\bauthor{\bsnm{Steinwart},~\bfnm{I.}\binits{I.}}
(\byear{2011}).
\btitle{Adaptive density level set clustering}.
In \bbooktitle{Proceedings of the 24th Conference on Learning Theory 2011}
(\beditor{\bfnm{S.}\binits{S.}~\bsnm{Kakade}} \AND
\beditor{\bfnm{U.}\binits{U.}~\bparticle{von}~\bsnm{Luxburg}}, eds.).
\bseries{JMLR Workshop and Conference Proceedings}
\bvolume{19}
\bpages{703--738}.
\bpublisher{JMRL}.
\end{binproceedings}
%

\bptok{imsref}%
\endbibitem

\bibitem{SteinwartXXb}
%
\begin{bmisc}[author]
\bauthor{\bsnm{Steinwart},~\bfnm{I.}\binits{I.}}
(\byear{2015}).
\bhowpublished{Supplement to ``Fully adaptive density-based clustering.''
DOI:\doiurl{10.1214/15-AOS1331SUPP}}.
\bptok{imsref}%
\end{bmisc}
%
\endbibitem
%
%
%

\bibitem{StHuSc05a}
%
\begin{barticle}[mr]
\bauthor{\bsnm{Steinwart},~\bfnm{Ingo}\binits{I.}},
\bauthor{\bsnm{Hush},~\bfnm{Don}\binits{D.}} \AND
\bauthor{\bsnm{Scovel},~\bfnm{Clint}\binits{C.}}
(\byear{2005}).
\btitle{A classification framework for anomaly detection}.
\bjournal{J. Mach. Learn. Res.}
\bvolume{6}
\bpages{211--232}.
\bid{issn={1532-4435}, mr={2249820}}
\end{barticle}
%

\bptok{imsref}%
\endbibitem

\bibitem{Stuetzle03a}
%
\begin{barticle}[mr]
\bauthor{\bsnm{Stuetzle},~\bfnm{Werner}\binits{W.}}
(\byear{2003}).
\btitle{Estimating the cluster type of a density by analyzing the
minimal spanning tree of a sample}.
\bjournal{J. Classification}
\bvolume{20}
\bpages{25--47}.
\bid{doi={10.1007/s00357-003-0004-6}, issn={0176-4268}, mr={1983120}}
\end{barticle}
%

\bptok{imsref}%
\endbibitem

\bibitem{StNu10a}
%
\begin{barticle}[mr]
\bauthor{\bsnm{Stuetzle},~\bfnm{Werner}\binits{W.}} \AND
\bauthor{\bsnm{Nugent},~\bfnm{Rebecca}\binits{R.}}
(\byear{2010}).
\btitle{A generalized single linkage method for estimating the cluster
tree of a density}.
\bjournal{J. Comput. Graph. Statist.}
\bvolume{19}
\bpages{397--418}.
\bid{doi={10.1198/jcgs.2009.07049}, issn={1061-8600}, mr={2675094}}
\end{barticle}
%

\bptok{imsref}%
\endbibitem

\bibitem{Tsybakov97a}
%
\begin{barticle}[mr]
\bauthor{\bsnm{Tsybakov},~\bfnm{A.~B.}\binits{A.~B.}}
(\byear{1997}).
\btitle{On nonparametric estimation of density level sets}.
\bjournal{Ann. Statist.}
\bvolume{25}
\bpages{948--969}.
\bid{doi={10.1214/aos/1069362732}, issn={0090-5364}, mr={1447735}}
\end{barticle}
%

\bptok{imsref}%
\endbibitem

\bibitem{Walther97a}
%
\begin{barticle}[mr]
\bauthor{\bsnm{Walther},~\bfnm{Guenther}\binits{G.}}
(\byear{1997}).
\btitle{Granulometric smoothing}.
\bjournal{Ann. Statist.}
\bvolume{25}
\bpages{2273--2299}.
\bid{doi={10.1214/aos/1030741072}, issn={0090-5364}, mr={1604445}}
\end{barticle}
%

\bptok{imsref}%
\endbibitem
\end{thebibliography}
%
%

%



\printaddresses
\end{document}


\begin{center}
  {\large\textbf{SUPPLEMENT TO ``FULLY ADAPTIVE DENSITY-BASED CLUSTERING''}}
  \vspace*{2ex}

  \textsc{By Ingo Steinwart}
  \vspace*{1.5ex}

  \textit{University of Stuttgart}
  \vspace*{1.5ex}

  {\small
  \begin{minipage}[t]{0.75\textwidth}
  In this  supplement
several auxiliary results, which are partially taken from \citeappendix{Steinwart11a},
are presented and
  the assumptions made in the paper are discussed in more detail. This material is contained in the sections
  \ref{sec-app:level-sets} to \ref{sec:proof-rates}.
   In addition, we present a couple of two-dimensional examples that show
that the assumptions imposed in the paper are not only met by many discontinuous densities, but also
by many continuous densities. This material is contained in the sections \ref{sec:single-twod} and \ref{suppB:cons-dens}.
\end{minipage}}
\end{center}

\ifthenelse{\equal{\arxvers}{y}}
{\noindent\textbf{Appendix A: Remaining Proofs and Additional Material.} In this  appendix  
several auxiliary results, which are partially taken from \citeappendix{Steinwart11a},
are presented and
}

\setcounter{section}{0}
\renewcommand{\thesection}{A.\arabic{section}}




\vspace*{3ex}
\noindent
\textbf{Appendix A. Remaining Proofs and Additional Material.}
In this appendix, the auxiliary results from \citeappendix{Steinwart11a} are presented and the assumptions
are discussed in more detail than it was possible in the main paper.
\section{Material Related to  Level Sets}\label{sec-app:level-sets}

In this section we present some additional results
from \citeappendix{Steinwart11a}
related to the definition of $M_\r$.

To begin with, we
note that
using the definition of the support of a measure it becomes obvious that $M_\r$ can be expressed by
\begin{equation}\label{Mr-eq}
 M_\r = \bigl\{x\in X: \mu_\r(U)>0 \mbox{ for all open neighborhoods $U$ of } x\bigr\}\, .
\end{equation}
Furthermore, if $\supp \mu = X$,  we actually have
$M_\r = X$ for all $\r\leq 0$,  but typically we are, of course,  interested in the case $\r>0$, only.
The next lemma shows that the sets $M_\r$ are ordered in the usual way.

\begin{lemma}\label{include-level}
 Let $(X,d)$ be a complete separable metric  space, $\mu$ be a $\s$-finite  measure on $X$,  and $P$
be a $\mu$-absolutely continuous distribution on $X$. Then, for all $\r_1\leq \r_2$, we have
\begin{align*}
 M_{\r_2} & \subset M_{\r_1} \, .
\end{align*}
\end{lemma}

\begin{proof}[Proof of Lemma \ref{include-level}]
 We fix an $x\in M_{\r_2}$
and an open set $U\subset X$ with $x\in U$. Moreover, we fix a $\mu$-density $h$ of $P$.
 Then we obtain
\begin{displaymath}
\mu_{\r_1}(U)
=  \mu\bigl(\{h\geq \r_1\} \cap U  \bigr)
\geq
 \mu\bigl(\{h\geq \r_2\} \cap U      \bigr)
=
\mu_{\r_2}(U)
>0 \, ,
\end{displaymath}
and hence  we obtain $x\in M_{\r_1}$ by (\ref{Mr-eq}).
\end{proof}

The following lemma describes the relationship between   $M_\r$ and $\{h\geq \r\}$.

\begin{lemma}\label{Mr-include-both-new}
  Let $(X,d)$ be a complete separable metric  space, $\mu$ be a $\s$-finite
 measure on $X$ with $\supp \mu=X$,
and $P$ be a $\mu$-absolutely continuous distribution on $X$.
Then, for all $\mu$-densities $h$ of $P$ and all $\r\in \R$, we have
\begin{displaymath}
  {\mathring{\{h\geq \r\}}} \subset M_\r  \subset \overline{\{h\geq \r\}}\, .
\end{displaymath}
If $h$ is continuous, we even have $\{h \!>\!\r\} \subset M_\r  \subset \{h\!\geq\! \r\}$
and $\partial M_\r \subset \{h\!=\!\r\}$.
\end{lemma}

\begin{proof}[Proof of Lemma \ref{Mr-include-both-new}]
By definition, $M_\r$ is the smallest closed set $A$ satisfying $\mu( \{h\geq \r\}\setminus A)=0$.
Moreover, we obviously have
\begin{displaymath}
   \mu\bigl( \{h\geq \r\} \setminus \overline{ \{h\geq \r\}} \bigr) = 0\, ,
\end{displaymath}
and hence we obtain $M_\r \subset \overline{\{h\geq \r\}}$.
 To show the other inclusion, we fix an $x\in \mathring{\{h\geq \r\}}$ and an open set $U\subset X$ with
$x\in U$. Then $\mathring{\{h\geq \r\}} \cap U$ is open and non-empty, and hence $\supp \mu = X$ yields
\begin{displaymath}
\mu_\r(U) = \mu\bigl(\{h\geq \r\} \cap U \bigr) \geq \mu\bigl(\mathring{\{h\geq \r\}} \cap U \bigr) > 0\, .
\end{displaymath}
By (\ref{Mr-eq}) we conclude that $x\in M_\r$, that is, we have shown
$\mathring{\{h\geq \r\}} \subset M_\r$.

Now assume that $h$ is continuous.
Clearly, we have $\{h> \r\} \subset \{h\geq \r\}$ and since $\{h>\r\}$ is open, we
conclude that $\{h> \r\} \subset \mathring{\{h\geq \r\}} \subset M_\r$ by the previously shown inclusion.
Moreover, since $\{h\geq \r\}$ is closed, we find  $M_\r  \subset \overline{\{h\geq \r\}} = \{h\geq \r\}$.
Recalling that $M_\r$ is closed by definition, we further find
$\partial M_\r \subset M_\r \subset \{h\geq \r\}$, and thus it remains to show
$\partial M_\r  \subset \{h\leq \r\}$. Let us assume the converse, i.e.,
that there exists an $x\in \partial M_\r$
such that
$h(x) > \r$. By the continuity we then find an open neighborhood $U$ of $x$ with $U\subset \{h>\r\}$.
Since $x\in \partial M_\r$, we further find an $y \in U\setminus M_\r$, while our construction
together with the previously shown $\{h> \r\} \subset   M_\r$
yields the contradicting statement $U\setminus M_\r\subset \{h>\r\}\setminus M_\r = \emptyset$.
\end{proof}

The next lemma provides some simple sufficient conditions for normality.

\begin{lemma}\label{regular-new}
  Let $(X,d)$ be a complete separable metric  space, $\mu$ be a  $\s$-finite measure on $X$ with $\supp \mu =X$, and $P$
be a $\mu$-absolutely continuous distribution on $X$.
Then the following statements hold:
\begin{enumerate}
 \item If $P$ has an upper semi-continuous $\mu$-density, then it is upper normal at every level.
 \item If $P$ has a lower semi-continuous $\mu$-density, then it is lower normal at every level.
 \item If, for some $\r\geq 0$,  $P$ has a $\mu$-density $h$ such that $\mu(\partial\{h\geq \r\})=0$, then $P$ is
  normal at level $\r$.
\end{enumerate}
\end{lemma}

\begin{proof}[Proof of Lemma \ref{regular-new}]
 \ada i Let us fix  an upper semi-continuous $\mu$-density $h$ of $P$.
Then $\{h\geq \r\}$ is closed, and hence Lemma \ref{Mr-include-both-new}
shows $M_\r \subset \overline{\{h\geq \r\}} = \{h\geq \r\}$. Thus,
$P$ is upper normal at level $\r$.

\ada {ii}   Let $h$ be a lower semi-continuous $\mu$-density of $P$. By Lemma \ref{Mr-include-both-new}
we then know $\{ h>\r \} = \mathring {\{h>\r\}}\subset {\mathring{\{h\geq \r\}}} \subset \mathring M_\r$.
This yields the assertion.

\ada {iii} The upper  normality follows from (2.3).
To see that $P$ is lower normal, we use the inclusion
$\{h>\r\}  \setminus \mathring M_\r \subset \overline{\{h\geq \r\}}\setminus \mathring{\{h\geq \r\}} = \partial\{h\geq \r\}$
which follows from Lemma \ref{Mr-include-both-new}.
\end{proof}

Let us now assume that $P$ is upper normal at some level $\r$. By (2.2) we then immediately see that
\begin{equation}\label{reg1}
 \mu(M_\r \symdif \{h\geq \r\}) = 0
\end{equation}
for {\em all\/}  $\mu$-densities  $h$ of $P$.
In other words,
 up to $\mu$-zero measures, $M_\r$ equals the $\r$-level set of {\em all\/}  $\mu$-densities $h$ of
$P$.
Moreover, if for some $\rs>0$ and $\rss>\rs$, the distribution $P$ is upper normal at every level $\r\in (\rs,\rss]$, then
using the monotonicity of the sets $M_\r$ and $\{h\geq \r\}$ in $\r$
as well as $(\cup_{i\in I} A_i)\symdif (\cup_{i\in I} B_i)\subset \cup_{i\in I}(A_i\symdif B_i)$, we   find
\begin{equation}\label{reg3}
 \mu\biggl(\!\{h\! >\! \rs\}\!\symdif\!\bigcup_{\r>\rs}M_\r  \biggr)
\leq
\mu\biggl(\bigcup_{n\in \N}\bigl(\{h\! \geq\! \rs\!+\!1/n\}\!\symdif\! M_{\rs\!+\!1/n} \bigr) \!\biggr)
=0
\end{equation}
for all $\mu$-densities  $h$ of $P$, and if $P$ has a continuous density $h$, we even have
$\bigcup_{\r>\rs}M_\r = \{h > \rs\}$ by an easy consequence of Lemma \ref{Mr-include-both-new}.
Similarly, if $P$ is lower normal at every level $\r\in (\rs,\rss]$,
we   find
\begin{equation}\label{reg4}
 \mu\biggl(\!\{h \!> \!\rs\}\setminus \bigcup_{\r>\rs}\mathring M_\r  \!\biggr)
\leq
\mu\biggl(\bigcup_{n\in \N}\bigl(\{h\! > \!\rs\!+\!1/n\}\setminus \mathring M_{\rs\!+\!1/n} \bigr) \!\biggr)
= 0\,,
\end{equation}
and if in addition, \eqref{reg3} holds, we obtain $\mu(\bigcup_{\r>\rs}M_\r  \symdif  \bigcup_{\r>\rs}\mathring M_\r )=0$.




\section{Proofs and Material on  Connected Components}\label{proof-cc-stuff}




This section contains the proofs related to   Subsection 2.2.
In addition, we recall several additional results on connected components from \citeappendix{Steinwart11a}.

\begin{lemma}\label{def-crm}
Let $A\subset B$ be two non-empty sets with partitions $\ca P(A)$ and $\ca P(B)$, respectively.
Then the following statements are equivalent:
\begin{enumerate}
   \item $\ca P(A)$ is comparable to $\ca P(B)$.
   \item There exists a  $\z: \ca P(A) \to \ca P(B)$ such that, for all $ A'\in \ca P(A)$, we have
\begin{equation}\label{crm-prop}
 A' \subset \z(A')\, .
\end{equation}
\end{enumerate}
Moreover, if one these statements   are true, the map $\z$ is uniquely determined by \eqref{crm-prop}.
We call
 $\z$ the cell relating map (CRM) between $A$ and $B$.
\end{lemma}

\begin{proof}[Proof of Lemma \ref{def-crm}]
    \atob  {ii} i Trivial.

   \atob   i {ii} For $A'\in \ca P(A)$ we find a $B'\in \ca P(B)$ such that $A'\subset B'$. Defining
    $\z(A') := B'$ then gives the desired Property \eqref{crm-prop}.

   Finally, assume that {\em ii)} is true but $\z$ is not unique.
        Then there exist $A'\in \ca P(A)$ and $B',B''\in \ca P(B)$ with $B'\neq B''$ and both
        $A'\subset B'$ and $A'\subset B''$. Since $A'\neq \emptyset$, this yields $B'\cap B''\neq\emptyset$,
        which in turn implies $B'=B''$ as $\ca P(B)$ is a partition, i.e.~we have found a contradiction.
\end{proof}

\begin{proof}[Proof of Lemma 2.4]
Clearly, $\z:= \z_{B,C} \circ \z_{A,B}$ maps from  $\ca P(A)$ to $\ca P(C)$. Moreover,
 for $A'\in \ca P(A)$ we have $A'\subset \z_{A,B}(A')$ and for $B':= \z_{A,B}(A')\in \ca P(B)$
we have $B'\subset \z_{B,C}(B')$. Combining these inclusions
we find
\begin{displaymath}
   A'\subset \z_{A,B}(A')\subset \z_{B,C}(\z_{A,B}(A')) = \z_{B,C} \circ \z_{A,B}(A') = \z(A')
\end{displaymath}
for all $A'\in \ca P(A)$.
Consequently, $\ca P(A)$ is comparable to $\ca P(C)$
and
by Lemma \ref{def-crm} we  see that
$\z$ is the
CRM $\z_{A,C}$, that is
$ \z_{A,C} =\z= \z_{B,C} \circ \z_{A,B}$.
\end{proof}

\begin{lemma}\label{connect-comp}
Let $(X,d)$ be a  metric space, $A\subset X$ be a non-empty subset and $\t>0$.
Then every $\t$-connected component of $A$ is $\t$-connected.
\end{lemma}

\begin{proof}[Proof of Lemma \ref{connect-comp}]
Let $A'$ be a $\t$-connected component of $A$ and $x,x'\in A'$.
Then $x$ and $x'$ are $\t$-connected in $A$, and hence there
 exist
$x_1,\dots,x_n\in A$ such that  $x_1=x$, $x_n=x'$ and $d(x_i,x_{i+1}) < \t$ for all $i=1,\dots,n-1$.
Now, $d(x_1,x_2)<\t$ shows that $x_1$ and $x_2$ are $\t$-connected in $A$, and hence they belong to the same
$\t$-connected component, i.e.~we have found $x_2\in A'$. Iterating this argument, we find $x_i\in A'$ for all
$i=1,\dots,n$. Consequently, $x$ and $x'$ are not only $\t$-connected in $A$, but also $\t$-connected in $A'$.
This shows that $A'$ is $\t$-connected.
\end{proof}

\begin{lemma}\label{cc-comp}
   Let $(X,d)$ be a metric space and $A\subset B$ be two closed non-empty subsets of $X$ with $|\ca C(B)|<\infty$.
   Then $\ca C(A)$ is comparable to $\ca C(B)$.
\end{lemma}

\begin{proof}[Proof of Lemma \ref{cc-comp}]
   Let us fix an $A'\in \ca C(A)$. Since $A\subset B$
   and $|\ca C(B)|<\infty$ there then exist
        an $m\geq 1$ and
 mutually distinct $B_1,\dots, B_m\in \ca C(B)$
   with $A'\subset B_1\cup \dots \cup B_m$ and $A'\cap B_i \neq \emptyset$ for all $i=1,\dots, m$.
   Since $A$ and $B$ are closed,   $A'$  and the
    sets  $A'\cap B_i$ are also closed. Consequently, the
    sets $A'\cap B_i$ are also closed in $A'$ with respect to the relative topology of $A'$.
   Let us now assume that $m>1$. Then $A'\cap B_1$ and $(A'\cap B_2) \cup \dots \cup (A'\cap B_m)$
   are two disjoint relatively closed non-empty subsets of $A'$ whose union equals $A'$. Consequently
   $A'$ is not connected, which contradicts $A'\in \ca C(A)$. In other words, we have $m=1$, that is,
   $\ca C(A)$ is comparable to $\ca C(B)$.
\end{proof}

\begin{lemma}\label{finite-comp}
 Let $(X,d)$ be a  metric space, $A\subset X$ be   non-empty   and $\t>0$.
 Then we have $d(A',A'') \geq \t$ for all $A',A''\in \ca C_\t(A)$ with $A'\neq A''$.
 Moreover, if $A$ is closed, all $A'\in \ca C_\t(A)$ are closed, and if $X$ is compact
 we have $|\ca C_\t(A)|<\infty$.
\end{lemma}

\begin{proof}[Proof of Lemma \ref{finite-comp}]
Let $A'\neq A''$ be two $\t$-connected components of $A$. Then we have $d(x',x'') \geq  \t$ for all $x'\in A'$ and $x''\in A''$,
since otherwise $x'$ and $x''$ would be $\t$-connected in $A$. Thus, we have $d(A',A'')\geq \t$, and from the
latter and the compactness of $X$, we  conclude that
$|\ca C_\t(A)|<\infty$.
Finally, let $(x_i)\subset A'$ be a sequence in some component $A'\in \ca C_\t(A)$
such that $x_i\to x$ for some $x\in X$.
Since $A$ is closed, we have $x\in A$, and hence
$x\in A''$ for some $A''\in \ca C_\t(A)$.
By construction we find $d(A',A'')=0$, and hence we obtain $A'=A''$
by the assertion that has been shown first.
\end{proof}

\begin{lemma}\label{connect-char}
Let $(X,d)$ be a metric space, $A\subset X$ be a non-empty subset and $\t>0$.
Then the following statements are equivalent:
\begin{enumerate}
 \item $A$ is $\t$-connected.
 \item For all non-empty subsets $A^+$ and $A^-$ of $A$ with $A^+\cup A^-= A$ and $A^+\cap A^-= \emptyset$ we have $d(A^+, A^-) <  \t$.
\end{enumerate}
\end{lemma}

\begin{proof}[Proof of Lemma \ref{connect-char}]
 \atob i {ii} We fix two  subsets $A^+$ and $A^-$ of $A$ with $A^+\cup A^-= A$ and $A^+\cap A^-= \emptyset$.
Let us further fix two points $x^+\in A^+$ and $x^-\in A^-$.
Since $A$ is $\t$-connected,   there then  exist
$x_1,\dots,x_n\in A$ such that  $x_1=x^-$, $x_n=x^+$ and $d(x_i,x_{i+1}) < \t$ for all $i=1,\dots,n-1$.
Then, $x^+\in A^+$ and $x^-\in A^-$ imply the existence of an $i\in \{1,\dots,n-1\}$ with
$x_i\in A^-$ and $x_{i+1}\in A^+$. This yields
 $d(A^+, A^-) \leq  d(x_i,x_{i+1}) < \t$.

\atob  {ii} i
Assume that $A$ is not $\t$-connected, that is $|\ca C_\t(A)| > 1$. 
We pick an $A^+\in \ca C_\t(A)$ and write $A^- := A\setminus A^+$. Since $|\ca C_\t(A)| > 1$,
both sets are non-empty, and our construction  ensures that they are also disjoint and satisfy $A^+\cup A^-= A$.
Moreover, for every $A'\in \ca C_\t(A)$ with $A'\neq A^+$ we know $d(A^+,A')\geq \t$
by  Lemma \ref{finite-comp} and since $A^-$ is the union of such $A'$, we conclude $d(A^+, A^-) \geq   \t$.
\end{proof}

\begin{corollary}\label{connect-char-cor}
Let $(X,d)$ be a metric space, $A\subset B\subset X$ be  non-empty subsets and $\t>0$.
If $A$ is $\t$-connected, then there exists exactly one $\t$-connected component $B'$ of $B$ with
$A\cap B' \neq \emptyset$. Moreover, $B'$ is the only $\t$-connected component $B''$ of $B$ that satisfies
$A\subset B''$.
\end{corollary}

\begin{proof}[Proof of Corollary \ref{connect-char-cor}]
The second assertion is a direct consequence of the first, and hence it suffice to show the first assertion.
Let us assume the first is not true. Since $A\subset B$
 there then
exist  $B', B''\in \ca C_\t(B)$ with $B'\neq B''$, $A\cap B'\neq \emptyset$, and $A\cap B''\neq \emptyset$.
We write $A^-:= A \cap B'$ and $A^+ := A \cap (B\setminus B')$. Since $B''\subset B\setminus B'$, we obtain $A^+\neq \emptyset$,
and therefore, Lemma \ref{connect-char} shows $d(A^-,A^+) < \t$. Consequently, there exist $x^-\in A^-$ and $x^+\in A^+$ with
$d(x^+,x^-)< \t$. Now we obviously have $x^-\in B'$, and by construction, we also find a $B'''\in \ca C_\t(B)$
with $x^+\in B'''$. Our previous inequality then yields $d(B',B''')< \t$, while Lemma \ref{finite-comp}
shows $d(B', B''')\geq  \t$, that is, we have found a contradiction.
\end{proof}

\begin{lemma}\label{tcc-comp}
 Let $(X,d)$ be a   metric space, $A\subset B$ be two non-empty subsets of $X$ and $\t>0$.
 Then $\ca C_\t(A)$ is comparable to $\ca C_\t(B)$.
\end{lemma}

\begin{proof}[Proof of Lemma \ref{tcc-comp}]
 For $A'\hspace*{-0.4pt}\in\hspace*{-0.4pt} \ca C_\t(A)$, Corollary \ref{connect-char-cor} shows that there is exactly $B'\in \ca C_\t(B)$
with $A'\subset B'$. Thus,  $\ca C_\t(A)$ is comparable to $\ca C_\t(B)$.
\end{proof}

\begin{lemma}\label{comp-char}
Let $(X,d)$ be a  metric space, $A\subset X$ be a non-empty subset and $\t>0$.
Then, for a partition $A_1,\dots,A_m$ of $A$, the following statements are equivalent:
\begin{enumerate}
 \item $\ca C_\t(A) = \{A_1,\dots,A_m\}$.
  \item  $A_i$ is $\t$-connected for all $i=1,\dots,m$,   and $d(A_i,A_j)\geq \t$
  for all $i\neq j$.
\end{enumerate}
\end{lemma}

\begin{proof}[Proof of Lemma \ref{comp-char}]
\atob i {ii} Follows from Lemma \ref{finite-comp}.

\atob  {ii} i Let us fix an $A'\in \ca C_\t(A)$ and an $A_i$ with  $A_i\cap A'\neq \emptyset$.
        Since $A_i$ is $\t$-connected and $A'\in \ca C_\t(A)$, Corollary \ref{connect-char-cor} applied
        to the sets $A_i\subset A\subset X$   yields $A_i\subset A'$.
Moreover, $A_1,\dots,A_m$ is a partition of $A$, and thus we conclude
 that
  \begin{displaymath}
   A' = \bigcup_{i\in I} A_i \, ,
  \end{displaymath}
 where $I:= \{i: A_i \cap A'\neq \emptyset\}$. Now let us assume that $|I|\geq 2$. We fix an $i_0\in I$
 and write $A^+:= A_{i_0}$ and $A^-:= \bigcup_{i\in I\setminus\{i_0\}} A_i$. Since $|I|\geq 2$, we obtain $A^-\neq \emptyset$,
 and  Lemma \ref{connect-char} thus shows $d(A^+,A^-) < \t$. On the other hand, our assumption ensures $d(A^+,A^-)\geq \t$,
  and hence  $|I|\geq 2$ cannot be true. Consequently, there exists a unique index $i$ with $A'=A_i$.
\end{proof}

\begin{lemma}\label{top-connect-new1}
  Let $(X,d)$ be a compact metric space and $A\subset X$ be a non-empty closed subset.
Then the following statements are equivalent:
\begin{enumerate}
 \item $A$ is connected.
 \item  $A$ is $\t$-connected for all $\t>0$.
\end{enumerate}
\end{lemma}

\begin{proof}[Proof of Lemma \ref{top-connect-new1}]
 \atob i {ii} Assume that $A$ is not $\t$-connected for some $\t>0$.
 Then, by Lemma \ref{finite-comp}, there are finitely many $\t$-connected components $A_1,\dots,A_m$ of $A$ with $m>1$.
 We write $A':= A_1$ and $A'' := A_2\cup\dots\cup A_m$. Then $A'$ and $A''$ are non-empty, disjoint and $A'\cup A'' = A$
 by construction. Moreover, Lemma \ref{finite-comp} shows that $A'$ and $A''$ are closed since $A$ is closed, and hence
 $A$ cannot be connected.

 \atob  {ii} i Let us assume that $A$ is not connected. Then there exist two non-empty
closed disjoint subsets of $A$ with $A'\cup A'' = A$. Since $X$ is compact, $A'$ and $A''$ are  also
compact,  and hence $A'\cap A'' = \emptyset$ implies $\t:=d(A',A'')>0$.
  Lemma \ref{connect-char} then shows that $A$ is not $\t$-connected.
\end{proof}

The next proposition investigates the relation between $\ca C_\t(A)$ and $\ca C(A)$.

\begin{proposition}\label{top-connect-new2-prop}
Let $(X,d)$ be a compact metric space and $A\subset X$ be a non-empty closed subset.
Then the following statements hold:
\begin{enumerate}
 \item For all $\t>0$, $\ca C(A)$ is comparable to $\ca C_\t(A)$ and the  CRM
  $\z:\ca C(A)\to \ca C_\t(A)$ is surjective.
 \item If $|\ca C(A)|<\infty$, we have
    \begin{displaymath}
   \t^*_A :=  \min\bigl\{d(A',A''): A',A''\in \ca C(A) \mbox{ with } A'\neq A''\bigr\}> 0 \, ,
  \end{displaymath}
  where $\min \emptyset := \infty$.
 Moreover, for all $\t\in (0,\t_A^*]\cap (0,\infty)$, we have
 $\ca C(A)= \ca C_\t(A)$ and, for
    such $\t$, the  CRM
    $\z:\ca C(A)\to \ca C_\t(A)$ is bijective.
 Finally, if $\t_A^*<\infty$, that is, $|\ca C(A)|>1$, we have
   \begin{displaymath}
      \t^*_A = \max\{ \t>0: \ca C(A)= \ca C_\t(A)\}\, .
     \end{displaymath}
\end{enumerate}
\end{proposition}

Note that, in general, a closed subset of $A$  may have infinitely many topologically connected
components as, e.g., the Cantor set shows. In this case, the second
assertion of the lemma above is, in general, no longer true.

\begin{proof}[Proof of Proposition \ref{top-connect-new2-prop}]
\ada i
Let $A'\in \ca C(A)$  and $\t>0$. Since $A$ is closed, so is $A'$, and hence $A'$ is $\t$-connected by Lemma
\ref{top-connect-new1}. Consequently,
Corollary \ref{connect-char-cor} shows that there exists an  $A''\in \ca C_\t(A)$ with $A'\subset A''$,
i.e.~$\ca C(A)$ is comparable to $\ca C_\t(A)$.
Now we fix an $A''\in \ca C_\t(A)$. Then there exists an $x\in A''$, and to this $x$, there exists an $A'\in \ca C(A)$
with $x\in A'$. This yields $A'\cap A''\neq \emptyset$, and since $A'$ is $\t$-connected by
Lemma
\ref{top-connect-new1},
Corollary \ref{connect-char-cor} shows $A'\subset A''$, i.e.~we obtain $\z(A')= A''$.

\ada {ii}
Let $A_1,\dots,A_m$ be the topologically connected components of $A$.
Then the components are closed, and
since $A$ is a closed and thus compact subset of $X$, the components are compact, too.
This shows $d(A_i,A_j) > 0$ for all $i\neq j$, and
consequently we obtain $\t_A^*>0$.
Let us fix a $\t\in (0,\t_A^*]\cap (0,\infty)$.
Then,
Lemma
\ref{top-connect-new1} shows that
each $A_i$ is $\t$-connected, and therefore
Lemma \ref{comp-char} together with  $d(A_i,A_j)\geq \t_A^*\geq  \t$ for all $i\neq j$ yields
$\ca C_\t(A) = \{A_1,\dots,A_m\}$.
Consequently,  we have proved $\ca C(A) = \ca C_\t (A)$.
The bijectivity of $\z$ now follows from its surjectivity.
For the proof of the last equation, we define $\t^*:= \sup\{ \t>0: \ca C(A)= \ca C_\t(A)\}$.
Then we have already seen that $\t^*_A\leq  \t^*$.
Now suppose that $\t^*_A<  \t^*$. Then there exists a $\t\in (\t^*_A,  \t^*)$ with
$\ca C(A) = \ca C_\t(A)$. On the one hand,
we then find $d(A_i,A_j)\geq \t$ for all $i\neq j$ by Lemma \ref{finite-comp}, while on the other hand
$\t>\t_A^*$ shows that there exist $i_0\neq j_0$
with $d(A_{i_0},A_{j_0})<\t$. In other words, the assumption $\t^*_A<  \t^*$ leads to
a contradiction, and hence we have $\t^*_A=  \t^*$.
\end{proof}

The last lemma in this subsection shows the monotonicity of $\t_A^*$.

\begin{lemma}\label{tau-star-comp}
  Let $(X,d)$ be a compact metric space
and  $A\subset B$ be two non-empty closed subsets of $X$ with $|\ca C(A)|<\infty$ and $|\ca C(B)|<\infty$.
If the CRM $\z: \ca C(A) \to \ca C(B)$ is injective, then we have $\t_A^* \geq \t_B^*$.
\end{lemma}

\begin{proof}[Proof of Lemma \ref{tau-star-comp}]
 Let us fix some $A',A''\in \ca C(A)$ with $A'\neq A''$. Since $\z$ is injective, we then obtain $\z(A')\neq \z(A'')$.
 Combining this with $A'\subset \z(A')$ and $A''\subset \z(A'')$, we
 find
  \begin{displaymath}
   d(A',A'') \geq d(\z(A'),\z(A'')) \geq \t_B^*\, ,
  \end{displaymath}
 where the last inequality follows from
        the definition of $\t_B^*$.
        Taking the infimum over all $A'$ and $A''$ with $A'\neq A''$
yields the assertion.
\end{proof}




\section{Additional Material Related to Tubes around Sets}\label{sec:proof-tubes}




This section contains additional material on the operations $A\pde$ and $A\mde$.

Let us begin by noting that
in the literature there is another, closely related concept for adding and cutting off $\d$-tubes,
which is based on
the Minkowski addition. Namely, in generic metric spaces $(X,d)$, we can define
\begin{align*}
 A\opde &:= \{x\in X: \exists y\in A \mbox{ with } d(x,y)\leq \d\}\\
 A\omde &:= \{x\in X: B(x,\d)\subset A\}
\end{align*}
for $A\subset X$ and $\d>0$, where $B(x,\d):=\{y\in X: d(x,y)\leq \d\}$ denotes the closed ball with radius
$\d$ and center $x$.
Some simple considerations then
show $A\omdee \subset A\mde\subset A\omde$ and
$A\opde \subset A\pde \subset A\opdee$ for all $\eps,\d>0$, that
is, the  operations of both concepts almost coincide. In addition, it is straightforward to check that
$A\omde = X\setminus (X\setminus A)\opde$.

Usually, the operations
$\oplus\d$ and $\ominus \d$
are considered
for the special case $X:= \Rd$ equipped with the Euclidean norm. In this case, we immediately obtain
the more common expressions
\begin{align*}
 A\opde &=  \{x +y : x\in A \mbox{ and } y\in \d B_{\ell_2^d}\} \\
 A\omde &= \{x\in \Rd: x+\d B_{\ell_2^d}\subset A\}\, ,
\end{align*}
where $B_{\ell_2^d}$ denotes the closed unit Euclidean ball at the origin.
Note that the latter formulas remain true for sufficiently small $\d>0$,
if we consider the ``relative case'' $X\subset \Rd$ and subsets $A\subset X$ satisfying $d(A, \Rd\setminus X)\in (0,\infty)$.

In general, it is  cumbersome to determine the exact forms of $A\pde$ and $A\mde$, respectively
$A\opde$ and $A\omde$ for a given $A$. For a particular class of sets $A\subset \R^2$,
Example \ref{Aomde} illustrates this by providing both $A\opde$ and $A\omde$.

The next lemma establishes some basic properties of the introduced operations.

\begin{lemma}\label{Td-lemma-new}
 Let $(X,d)$ be a   metric space and $A,B\subset X$ be two subsets. Then the
following statements hold:
\begin{enumerate}
 \item If $A$ is compact, then $ A\pde= A\opde$.
 \item We have $d(A,B)\leq d(A\pde, B\pde) + 2\d$.
 \item We have
\begin{equation}\label{inf-tube}
 \bigcap_{\d>0} A\pde = \overline{A}\, .
\end{equation}
\item We have $(A\cup B)\pde = A\pde \cup B\pde$ and $(A\cap B)\pde \subset  A\pde \cap B\pde$.
\item We have $A\mde \cup B\mde \subset (A\cup B)\mde$ and, if $d(A,B) >\d$, we actually have $A\mde \cup B\mde = (A\cup B)\mde$.
\item For $A_1,A_2\subset X$ with $A_1\cap A_2 = \emptyset$ and $B_i\subset A_i$ with $d(B_1,B_2)> \d$, we have
\begin{displaymath}
  ( A_1\mde \setminus B_1\mde) \cup  ( A_2\mde \setminus B_2\mde) \subset (A_1\cup A_2)\mde \setminus (B_1\cup B_2)\mde\, ,
\end{displaymath}
  and equality holds, if $d(A_1,A_2) > \d$.
\item For all $\d>0$ and $\eps>0$, we have $A\subset (A\pdee)\mde$ and $(A\mdee)\pde\subset A$.
\item For all $\d>0$ and $\eps>0$, we have $(\partial A)\pde \subset A\pdee \setminus A\mdee$.
\end{enumerate}
\end{lemma}

\begin{proof}[Proof of Lemma \ref{Td-lemma-new}]
 \ada i Clearly, it suffices to prove $A\pde\subset A\opde$. To prove this inclusion,
 we fix an  $x\in A\pde$. Then there exists a sequence $(x_n)\subset A$ with $d(x,x_n) \leq \d+1/n$ for all $n\geq 1$.
Since $A$ is compact, we may assume without loss of generality that $(x_n)$ converges to some $x'\in A$.
Now we easily obtain the assertion from $d(x,x')\leq d(x,x_n) + d(x_n,x')$.

 \ada {ii} Let us fix an $x\in A\pde$ and an $y\in B\pde$. Then there exist two sequences $(x_n)\subset A$ and
$(y_n)\subset B$ such that $d(x,x_n) \leq \d+1/n$ and $d(y,y_n) \leq \d+1/n$ for all $n\geq 1$.
For $n\geq 1$, this construction now yields
\begin{displaymath}
 d(A,B) \leq d(x_n,y_n) \leq d(x_n,x) + d(x ,y ) + d(y, y_n) \leq d(x ,y ) + 2\d + 2/n\,  ,
\end{displaymath}
and by first letting $n\to \infty$ and then taking the infimum over all $x\in A\pde$ and   $y\in B\pde$, we obtain the assertion.

\ada {iii} To show the inclusion $\supset$, we fix an $x\in \overline A$. Then there exists a sequence $(x_n)\subset A$
with $x_n\to x$ for $n\to \infty$. For $\d>0$ there then exists an $n_\d$ such that $d(x,x_n) \leq \d$ for all $n\geq n_\d$.
This shows $d(x,A)\leq \d$, i.e.~$x\in A\pde$. To show the converse inclusion $\subset$, we fix an $x\in X$ that satisfies $x\in A^{+1/n}$
for all $n\geq 1$. Then there exists a sequence $(x_n)\subset A$ with $d(x,x_n) \leq 1/n$, and hence we find $x_n\to x$
for $n\to \infty$. This shows $x\in \overline A$.

\ada {iv} If $x\in (A\cup B)\pde$, there exists a sequence $(x_n)\subset A\cup B$ with $d(x,x_n)\leq \d+1/n$. Without loss of
generality we may assume that $(x_n)\subset A$, which immediately yields $x\in A\pde$. The converse inclusion
$A\pde \cup B\pde \subset (A\cup B)\pde $ and the inclusion $(A\cap B)\pde \subset  A\pde \cap B\pde$
are trivial.

\ada v The first inclusion follows from part \emph{iv)} and simple set algebra, namely
\begin{align*}
 A\mde \cup B\mde = X\setminus \bigl((X\setminus A)\pde \cap (X\setminus B)\pde\bigr)
 &\subset X\setminus \bigl((X\setminus A)\cap (X\setminus B)\bigr)\pde\\
 & = X \setminus \bigl(X\setminus (A\cup B)  \bigr)\pde \\
 & = (A\cup B)\mde\, .
\end{align*}
To show the converse inclusion, we fix an $x\in (A\cup B)\mde$. Since $(A\cup B)\mde \subset A\cup B$, we may
assume without loss of generality that $x\in A$. It then remains to show that $x\in A\mde$, that is $d(x,X\setminus A)>\d$.
Obviously, $A\cap B=\emptyset$, which follows from $d(A,B)>\d$, implies
\begin{displaymath}
 X\setminus A
= ((X\setminus A) \cap (X\setminus B)) \cup ((X\setminus A) \cap B) = (X\setminus (A\cup B)) \cup B\, ,
\end{displaymath}
and hence
we obtain
 $d(x,X\setminus A) = d(x, X\setminus (A\cup B)) \wedge d(x,B) > \d\wedge \d = \d$
where we used both $x\in (A\cup B)\mde$ and $d(A,B)>\d$.

\ada {vi} Using the formula $(A_1 \cup A_2) \setminus (B_1\cup B_2) = (A_1 \setminus B_1) \cup (A_2\setminus B_2)$,
which easily follows from $A_i \setminus  B_j = A_i$ for $i\neq j$, we obtain
\begin{align*}
    ( A_1\mde \setminus B_1\mde) \cup  ( A_2\mde \setminus B_2\mde)
&= (A_1\mde \cup A_2\mde) \setminus (B_1\mde \cup B_2\mde)\\
 &   \subset (A_1\cup A_2)\mde \setminus (B_1\cup B_2)\mde\, ,
\end{align*}
 where in the last step we used \emph{v)}.  The second assertion also follows from  \emph{v)}.

 \ada {vii} Obviously, $A\subset (A\pdee)\mde$ is equivalent to $(X\setminus A\pdee)\pde\subset X\setminus A$.
 To prove the latter, we fix an $x\in (X\setminus A\pdee)\pde$. Then there exists a sequence $(x_n)\subset X\setminus A\pdee$
 with $d(x,x_n) \leq \d+1/n$ for all $n\geq 1$. Moreover, $(x_n)\subset X\setminus A\pdee$ implies
 $d(x_n,x') > \d+\eps$ for all $n\geq 1$ and $x'\in A$. Now assume that we had $x\in A$. For  an index $n$ with
 $1/n\leq \eps$, we would then obtain $\d+\eps < d(x_n,x)\leq \d+\eps$, and hence $x\in A$ cannot be true.

 To show the second inclusion we fix an $x\in (A\mdee)\pde$. Then there exists a sequence $(x_n)\subset A\mdee$ such that
 $d(x,x_n)\leq \d+1/n$ for all $n\geq 1$. This time, $x_n\in  A\mdee$ implies $x_n\not \in (X\setminus A)\pdee$, that is
 $d(x_n,x') > \d + \e$ for all $n\geq 1$ and $x'\in X\setminus A$.
 Choosing an $n$ with $1/n\leq \eps$, we then find $x\in A$.

\ada {viii} We fix an $x\in (\partial A)\opde$.
By definition,
 there  then exists an $x'\in \partial A$
with $d(x,x') \leq \d$. Moreover, by the definition of the boundary,
there exists an $x''\in A$ with $d(x',x'')\leq \eps$, and hence we find
$d(x,x'')\leq \d+\eps$, i.e.~$x\in A\pdee$. Since $\partial A = \partial (X\setminus A)$, the same argument
yields $x\in (X\setminus A)\pdee$, i.e.~$x\not \in A\mdee$.
Thus, we have shown $(\partial A)\opde \subset A\pdee\setminus A\mdee$.
Using $(\partial A)\pde \subset (\partial A)\opdee$ and a simple change of variables then yields the
assertion.
\end{proof}




\section{Additional Material Related to Persistence}\label{sec:proof-pers-new}




In this section we recall and prove  two results of \citeappendix{Steinwart11a} that
extend Theorem 2.7.

We begin with the
 following lemma, which shows  that $\ca C_\t(A)$ is persistent in $\ca C_\t(A\pde)$,
if $\t>0$ and $\d>0$ are sufficiently small.

\begin{lemma}\label{zeta-Td-surj}
Let $(X,d)$ be a compact metric space, and $A\subset X$ be   non-empty. Then, for all
$\d>0$ and $\t>\d$,
the following statements hold:
\begin{enumerate}
 \item The set $(A')\pde$ is $\t$-connected for all $A'\in \ca C_\t(A)$.
 \item The CRM $\z:\ca C_\t(A) \to \ca C_\t(A\pde)$ is surjective.
 \item If $A$ is closed,  $|\ca C(A)|<\infty$, and $\t\leq \t^*_A/3$, then
the CRM $\z:\ca C_\t(A) \to \ca C_\t(A\pde)$ is   bijective and satisfies
  \begin{equation}\label{zeta-Td-surj-hxx}
  \z(A') = (A')\pde \, , \qquad \qquad A'\in \ca C_\t(A).
  \end{equation}
\end{enumerate}
\end{lemma}

\begin{proof}[Proof of Lemma \ref{zeta-Td-surj}]
\ada i
Since $\t>\d$, there exist an $\e>0$ with $\d+\e<\t$. For $x\in (A')\pde$, there thus
exists an $x'\in A'$ with $d(x,x')\leq \d+\e<\t$, i.e.~$x$ and $x'$ are $\t$-connected.
Since $A'$ is $\t$-connected, it is then easy to show that every pair $x,x''\in  (A')\pde$ is $\t$-connected.

\ada {ii}
Let us fix an $A'\in \ca C_\t(A\pde)$ and an $x\in A'$.
For $n\geq 1$ there then exists an $x_n\in A$ with $d(x,x_n)\leq \d+1/n$
and since by Lemma \ref{finite-comp} there only exist finitely many $\t$-connected components of $A$, we
may assume without loss of generality that there exists an $A''\in \ca C_\t(A)$ with $x_n\in A''$ for all $n\geq 1$.
This yields $d(x,A'')\leq \d+1/n$ for all $n\geq 1$, and hence $d(x,A'')\leq \d$. Consequently, we obtain
$x\in (A'')\pde$, i.e.~we have $(A'')\pde \cap A'\neq \emptyset$.
Since $(A'')\pde \subset A\pde$,
 we then conclude that $(A'')\pde\subset A'$
by Corollary \ref{connect-char-cor} and part {\em i)}.
Furthermore, we clearly have $A''\subset (A'')\pde$,  and hence $\z(A'') = A'$.

\ada {iii}
Let us first consider the case $|\ca C(A)|=1$. In this case, part
\emph{i)} of Proposition \ref{top-connect-new2-prop} shows $|\ca C_\t(A)|=1$,
and thus $|\ca C_\t(A\pde)|=1$ by the already established part \emph{ii)}.
This makes the assertion obvious.

In the case $|\ca C(A)|>1$
we write
$A_1,\dots,A_m$ for the $\t$-connected components of $A$.
By part \emph{iv)} of Lemma \ref{Td-lemma-new} we then obtain
\begin{equation}\label{zeta-Td-surj-h1}
 A\pde = \bigcup_{i=1}^m A_i\pde\, .
\end{equation}
Since $|\ca C(A)|>1$, we further have $\t_A^*<\infty$, and hence
 part
\emph{ii)} of Proposition \ref{top-connect-new2-prop} yields $\ca C(A) = \ca C_\t(A)$.
The definition of $\t_A^*$ thus gives
$d(A_i,A_j)\geq \t_A^* \geq 3\t$ for all $i\neq j$.
Our first goal is to show that
 \begin{equation}\label{zeta-Td-surj-h2}
  d(A_i\pde, A_j\pde) \geq \t\, , \qquad \qquad i\neq j\, .
 \end{equation}
To this end, we fix $i\neq j$ and both an $x_i\in A_i\pde$ and an $x_j\in A_j\pde$. Now,
the compactness of $X$ yields the compactness of $A_i$ and $A_j$ by Lemma \ref{finite-comp},
and hence  part \emph{i)} of Lemma \ref{Td-lemma-new} shows that
there exist $x_i'\in A_i$ and $x_j'\in A_j$ with $d(x_i,x_i')\leq \d$ and $d(x_j,x_j')\leq \d$. This yields
\begin{displaymath}
 3\t \leq d(x_i',x_j') \leq d(x_i',x_i)+ d(x_i,x_j) +d(x_j,x_j') \leq 2\d + d(x_i,x_j)\, ,
\end{displaymath}
and the latter together with $\d< \t$ implies (\ref{zeta-Td-surj-h2}).

Now   {\em i)} showed that each $A_i\pde$, $i=1,\dots,m$, is $\t$-connected.
Combining this with  (\ref{zeta-Td-surj-h1}),  (\ref{zeta-Td-surj-h2}), and Lemma \ref{comp-char},
we   see that $A_1\pde, \dots, A_m\pde$ are the
$\t$-connected components of $A\pde$. The bijectivity of $\z$ then follows from the surjectivity and a
  cardinality argument, and \eqref{zeta-Td-surj-hxx} is obvious.
\end{proof}

The following theorem is an extended version of the statements   of Theorem 2.7
that deal with $\ca C_\t(M_\r\pde)$.

\begin{theorem}\label{reg-cluster-thm}
Let $(X,d)$ be a compact metric space,  $\mu$ be a finite
Borel measure on $X$ and $P$ be a $\mu$-absolutely continuous distribution on $X$
that can be clustered
between  $\r^*$ and $\r^{**}$.
Then the function $\ts$ defined by   (2.6)
is monotonically increasing.
Moreover,
for all $\e^*\in (0,\rss-\rs]$, $\d>0$, $\t\in (\d, \ts(\e^*)]$, and
all $\r\in [0, \r^{**}]$, the following statements hold:
\begin{enumerate}
 \item We have $1\leq|\ca C_\t(M_\r\pde)|\leq2$.
\item If $\r \geq  \r^* + \e^*$, then $|\ca C_\t(M_\r\pde)|=2$ and $\ca C(M_\r) \persist\ca C_\t(M_\r\pde)$.
\item If $|\ca C_\t(M_\r\pde)|=2$, then $\r\geq \r^*$ and
$\ca C_\t(M_{\r^{**}}\pde)\persist \ca C_\t( M_\r\pde)$.
\item If $\ca C_\t(M_{\r^{**}}\mde) \persist \ca C_\t( M_{\r^{**}}\pde)$ and
  $|\ca C_\t(M_\r\mde)|=1$, then   $\r<  \r^*+\e^*$.
\end{enumerate}
\end{theorem}

\begin{proof}[Proof of Theorem \ref{reg-cluster-thm}]
Let us first show the assertions related to the function $\ts$.
To this end, we first observe that for $\e\in (0,\rss-\rs]$ we have
$|\ca C(M_{\rs+\e})| = |\ca C(M_\rss)| =2$ by Definition 2.5. This shows
$\ts(\e) < \infty$.

Let us now fix $\e_1,\e_2\in (0,\rss-\rs]$ with $\e_1\leq \e_2$.
Then Definition
2.5 guarantees that both $M_{\rs+\e_1}$ and
$M_{\rs+\e_2}$ have two topologically connected components and that the
CRM  $\z:\ca C(M_{\rs+\e_2})\to \ca C (M_{\rs+\e_1})$ is bijective.
From Lemma \ref{tau-star-comp} we thus obtain
\begin{displaymath}
 \ts(\e_2) = \frac 1 3\t^*_{M_{\rs+\e_2}}\geq \frac 1 3\t^*_{M_{\rs+\e_1}} = \ts(\e_1)\, .
\end{displaymath}

 \ada i Since $\emptyset \neq M_\r\subset M_\r\pde$, we find $|\ca C_\t(M_\r\pde)|\geq 1$.
 On the other hand, since
 $\t>\d$,  part {\em ii)\/} of Lemma   \ref{zeta-Td-surj} and part {\em i)\/} of Proposition \ref{top-connect-new2-prop}
 yield
 \begin{equation}\label{reg-cluster-lem-h1}
  |\ca C_\t(M_\r\pde)|\leq |\ca C_\t(M_{\r})| \leq |\ca C(M_{\r})|\leq 2\, .
 \end{equation}

\ada {ii} Let us  fix a $\r\in [ \rs+\e^*, \rss]$. For $\e:= \r-\rs$, the monotonicity of $\ts$ then
gives $\ts(\e^*) \leq \ts(\e)$, and hence we obtain
\begin{displaymath}
 \t\leq \frac13\t^*_{M_{\rs+\e^*}}\leq \frac13\t^*_{M_\r}<\infty\, .
\end{displaymath}
 Part {\em ii)\/} of Proposition \ref{top-connect-new2-prop} thus  shows that  the
 CRM $\z_\r:\ca C(M_\r)\to \ca C_\t(M_\r)$ is   bijective.
Furthermore,
 $\d<\t\leq \t^*_{M_\r}/3$ together with
  part {\em iii)\/} of Lemma \ref{zeta-Td-surj} shows that the
CRM $\z_\d: \ca C_\t(M_\r) \to \ca C_\t(M_\r\pde)$ is bijective.
Consequently, the CRM $\z=\z_\d\circ \z_\r: \ca C(M_\r)\to \ca C_\t(M_\r\pde)$ is bijective, and from the
latter we conclude that
$|\ca C_\t(M_\r\pde)|=|\ca C(M_\r)|=2$.

\ada {iii} Since $|\ca C_\t(M_\r\pde)|=2$, the already established
\eqref{reg-cluster-lem-h1} yields $|\ca C(M_{\r})|=2$,
and hence   Definition 2.5 implies both $\r\geq \rs$
and the bijectivity of the CRM
 $\z^{**}:\ca C (M_{\rss})\to \ca C(M_\r)$.
Moreover, for $\rss$, the already established
part {\em ii)\/} shows  that the CRM $\z_M: \ca C_\t(M_{\rss}) \to \ca C_\t(M_\rss\pde)$
is bijective, and the proof of {\em ii)\/} further showed $\ca C(M_{\rss}) = \ca C_\t(M_{\rss})$.
Consequently,  $\z_M$ equals the CRM  $\ca C(M_{\rss}) \to \ca C_\t(M_\rss\pde)$.
 In addition, $\d<\t$ together with
part {\em ii)\/} of Lemma   \ref{zeta-Td-surj} and part {\em i)\/} of Proposition \ref{top-connect-new2-prop}
shows that the CRM $\z_\r: \ca C(M_\r) \to \ca C_\t(M_\r\pde)$ is surjective.
Now, by Lemma 2.4 these maps commute in the sense of the following diagram

\quadiass  {\ca C(M_{\rss})} {\ca C(M_\r)} {\ca C_\t(M_\rss\pde)}{\ca C_\t(M_\r\pde)} {\z^{**}} {\z_M} {\z_\r}{\z} \\
and consequently, the CRM $\z$ is surjective. Since $|\ca C_\t(M_\rss\pde)|= |\ca C(M_{\rss})| = 2$  and $|\ca C_\t(M_\r\pde)|=2$,
we then conclude that $\z$ is bijective.

\ada {iv} We proceed by contraposition.
To this end,  we fix an $\r\in [\rs+\e^*,\rss]$. By the already established part {\em ii)\/}
we then find $|\ca C_\t(M_\r\pde)|=2$,
and part {\em iii)\/} thus shows that the
CRM $\z_M: \ca C_\t(M_\rss\pde)\to \ca C_\t( M_\r\pde)$ is bijective.
Moreover, Lemma 2.4 yields the following diagram

\quadiass  {\ca C_\t(M_\rss\mde)} {\ca C_\t(M_\rss\pde)}{\ca C_\t(M_\r\mde)} {\ca C_\t(M_\r\pde)} {\z} {\z_V} {\z_M}{\z_{V,M}} \\
where $\z$, $\z_V$, and $\z_{V,M}$ are the corresponding CRMs. Now our assumption guarantees that $\z$ is bijective, and hence
the diagram shows that $\z_{V,M}\circ \z_V$ is bijective. Consequently, $\z_V$ is injective, and
from the latter we obtain
$2 = |\ca C_\t(M_\r\pde)|= |\ca C_\t(M_\rss\mde)| \leq |\ca C_\t(M_\r\mde)|$.
\end{proof}

The next lemma investigates situations in which $\ca C_\t(A\mde)$ is persistent in
$\ca C(A)$. In particular, it shows that
if $\t$ is sufficiently large compared to $\d$ and
$|\ca C_\t(A\mde)| = |\ca C(A)|$, then
we obtain persistence.  Informally speaking this means  that gluing $\d$-cuts by $\t$-connectivity
may preserve the component structure.

\begin{lemma}\label{thick-impl-reg}
 Let $(X,d)$ be a compact metric  space, and $A\subset X$ be non-empty and  closed
   with $|\ca C(A)|<\infty$.
 We define   $\psis_A:(0,\infty) \to [0,\infty]$ by
 \begin{displaymath}
    \psis_A(\d) := \sup_{x\in A} d(x,A\mde)\, , \qquad \qquad \d>0.
 \end{displaymath}
Then, for all $\d>0$ and all $\t>2\psis_A(\d)$,
the following statements hold:
\begin{enumerate}
 \item For all $B'\in \ca C(A)$, there is at most one $A'\in \ca C_\t(A\mde)$
with  $A'\cap B'\neq \emptyset$.
  \item We have $|\ca C_\t(A\mde)| \leq |\ca C(A)|$.
 \item If $|\ca C_\t(A\mde)| = |\ca C(A)|$, then $\ca C_\t(A\mde)$ is persistent in
$\ca C(A)$. Moreover, for all $B', B''\in \ca C(A)$ with $B'\neq B''$ we have
\begin{equation}\label{cluster-distance}
 d(B', B'') \geq  \t-2\psis_A(\d)\, .
\end{equation}
\end{enumerate}
\end{lemma}

\begin{proof}[Proof of Lemma \ref{thick-impl-reg}]
\ada i Let us fix a    $\psi >  2\psis_A(\d)$ with $\psi<\t$ and a
$\t'\in (0,\t^*_{A})$ such that $\psi + \t'< \t$, where $\t^*_{A}$ is the constant
defined in Proposition \ref{top-connect-new2-prop}.
Moreover, we fix a $B'\in \ca C(A)$.
By  Proposition \ref{top-connect-new2-prop} we then see that $\ca C(A) = \ca C_{\t'}(A)$, and hence $B'$ is $\t'$-connected.
Now let $A_1,\dots,A_m$ be the $\t$-connected components of $A\mde$.
Clearly,  Lemma \ref{finite-comp} yields $d(A_i,A_j)\geq \t$ for all $i\neq j$.
Assume that {\em i)\/}  is not true, that is,
there exist indices $i_0, j_0$ with $i_0\neq j_0$ such that
$A_{i_0}\cap B'\neq \emptyset$ and $A_{j_0}\cap B'\neq \emptyset$.
Thus, there exist $x'\in A_{i_0}\cap B'$ and $x''\in A_{j_0}\cap B'$, and since $B'$ is $\t'$-connected,
there further exist $x_0,\dots,x_{n+1}\in B'\subset A$
with $x_0=x'$, $x_{n+1} = x''$ and $d(x_i,x_{i+1}) < \t'$ for all $i=0,\dots,n$.
Moreover,  our assumptions guarantee
$d(x_i , A\mde) < \psi/2$ for all $i=0,\dots,n+1$.
For all $i=0,\dots,n+1$, there thus exists an index $\ell_i$ with
\begin{displaymath}
 d(x_i, A_{\ell_i}) < \psi/2\, .
\end{displaymath}
In addition, we have $x_0\in A_{i_0}$ and $x_{n+1} \in A_{j_0}$ by construction,
and hence we may actually choose
$\ell_0 = i_0$ and $\ell_{n+1} = j_0$. Since we assumed $\ell_0\neq \ell_{n+1}$, there then exists an $i\in \{0,\dots,n\}$
with $\ell_i\neq \ell_{i+1}$. For this index, our construction now yields
\begin{displaymath}
 d(A_{\ell_{i}},A_{\ell_{i+1}}) \leq d(x_i,A_{\ell_{i}}) + d(x_i,x_{i+1}) + d(x_{i+1},A_{\ell_{i+1}}) < \psi+ \t'< \t\, ,
\end{displaymath}
which contradicts the earlier established $d(A_{\ell_{i}},A_{\ell_{i+1}}) \geq \t$.

\ada {ii} Since $A\mde\subset A$, there exists, for
every $A'\in \ca C_\t(A\mde)$, a $B'\in \ca C(A)$ with $A'\cap B'\neq \emptyset$.
We pick one such $B'$ and define $\z(A') := B'$. Now part {\em i)\/} shows that
$\z:\ca C_\t(A\mde) \to \ca C(A)$ is injective, and hence we find
$|\ca C_\t(A\mde)| \leq |\ca C(A)|$.

\ada {iii} As mentioned in part {\em ii)\/}, we have
an injective  map $\z:\ca C_\t(A\mde) \to \ca C(A)$ that satisfies
\begin{equation}\label{thick-impl-reg-h1}
 A' \cap \z(A') \neq \emptyset\, , \qquad \qquad A'\in \ca C_\t(A\mde)\, .
\end{equation}
Now, $|\ca C_\t(A\mde)| = |\ca C(A)|$
together with the assumed $|\ca C(A)|<\infty$
implies that $\z$
is actually bijective. Let us first
show that $\z$ is the only map that satisfies (\ref{thick-impl-reg-h1}).
To this end, assume the converse, that is, for some $A'\in \ca C_\t(A\mde)$, there exists an $B'\in \ca C(A)$
with $B'\neq \z(A')$ and $A'\cap B'\neq \emptyset$. Since $\z$ is bijective, there then exists an $A''\in \ca C_\t(A\mde)$
with $\z(A'') = B'$, and hence we have $A''\cap B'\neq \emptyset$
by \eqref{thick-impl-reg-h1}. By part {\em i)\/}, we conclude that
$A'=A''$, which in turn yields $\z(A') = \z(A'') = B'$. In other words, we have found a contradiction, and hence $\z$
is indeed the only map that satisfies \eqref{thick-impl-reg-h1}.

Let us now show that
$\ca C_\t(A\mde)$ is persistent in
$\ca C(A)$.
Since we assumed $|\ca C_\t(A\mde)| = |\ca C(A)|$, it suffices to prove that
 the injective map $\z:\ca C_\t(A\mde) \to \ca C(A)$ defined by \eqref{thick-impl-reg-h1} is a CRM, i.e.~it satisfies
\begin{equation}\label{thick-impl-reg-h0xxx}
  A'\subset \z(A')\,, \qquad \qquad A'\in \ca C_\t(A\mde)\, .
\end{equation}
To show \eqref{thick-impl-reg-h0xxx}, we pick an $A'\in \ca C_\t(A\mde)$ and write
$B_1,\dots,B_m$ for the topologically connected components of $A$. Since $A\mde\subset A$, we then
have $A'\subset B_1\cup \dots\cup B_m$, where the latter union is disjoint.
Now, we have just seen that $\z(A')\in \{B_1,\dots,B_m\}$ is the only component satisfying $A'\cap \z(A')\neq \emptyset$, and therefore
we can conclude $A'\subset \z(A')$.

Finally, let us show \eqref{cluster-distance}. To this end, we first prove that,
 for all
 $A'\in \ca C_\t(A\mde)$ and $x\in \z(A')$ we have
 \begin{equation}\label{cluster-thick}
    d(x, A')  \leq \psis_A(\d)\, ,
 \end{equation}
where $\z:\ca C_\t(A\mde) \to \ca C(A)$ is the bijective CRM considered above.
Let us  assume that \eqref{cluster-thick} is not true,
that is, there exist an $A'\in C_\t(A\mde)$ and an  $x\in \z(A')$ such that
 $d(x, A') > \psis_A(\d)$. Since   $d(x, A\mde) \leq \psis_A(\d)$,
 there further exists an $A''\in C_\t(A\mde)$ with  $d(x, A'') \leq \psis_A(\d)$.
 Obviously, this yields $A'\neq A''$. Let us fix a $\t'\in (0,\t^*_{A})$ with $2 \psis_A(\d) + \t'< \t$,
 and an $x'\in A'$. For $B':= \z(A')$, we then have $x'\in B'$ by \eqref{thick-impl-reg-h0xxx}, and our construction
 guarantees $x\in B'$. Now, the rest of the proof is similar to that of \emph{i)}. Namely, since $B'$ is
 $\t'$-connected, there exist $x_0,\dots,x_{n+1}\in B'$
with $x_0=x$, $x_{n+1} = x'$ and $d(x_i,x_{i+1}) < \t'$ for all $i=0,\dots,n$.
Let $A_1,\dots,A_m$ be the $\t$-connected components of $A\mde$. Then,
for all $i=0,\dots,n+1$, there   exists an index $\ell_i$ with
\begin{displaymath}
 d(x_i, A_{\ell_i}) \leq \psis_A(\d)\, ,
\end{displaymath}
where  we may choose
$A_{\ell_0} = A''$ and $A_{\ell_{n+1}} = A'$.
Since  $\ell_0\neq \ell_{n+1}$, there then exists an $i\in \{0,\dots,n\}$
with $\ell_i\neq \ell_{i+1}$, and  our construction  yields
\begin{displaymath}
 \t\leq d(A_{\ell_{i}},A_{\ell_{i+1}})
 \leq d(x_i,A_{\ell_{i}}) + d(x_i,x_{i+1}) + d(x_{i+1},A_{\ell_{i+1}})
 < 2\psis_A(\d)+ \t'< \t\, .
\end{displaymath}

To prove \eqref{cluster-distance}, we again assume the converse, that
is, that there exist $B',B''\in \ca C(A)$ with $B'\neq B''$ and  $d(B', B'')< \t-2\psis_A(\d)$.
Then
there exist $x'\in B'$ and
$x''\in B''$ such that $d(x',x'') < \t-2\psis_A(\d)$.
Now, since $\z$ is bijective, there exists $A',A''\in  C_\t(A\mde)$ with
$A'\neq A''$, $B'=\z(A')$, and $B''=\z(A'')$.
Using \eqref{cluster-thick},
we then obtain
\begin{displaymath}
  \t\leq d(A', A'') \leq d(x', A') + d(x',x'') + d(x'', A'') < 2\psis_A(\d) + \t-2\psis_A(\d) = \t\, ,
\end{displaymath}
i.e.~we again have found a contradiction.
\end{proof}

The following theorem provides
 an extended version of the statements   of Theorem 2.7
that deal with $\ca C_\t(M_\r\mde)$.

\begin{theorem}\label{main-thick}
Let Assumption C be satisfied and
$\e^*\in (0,\rss-\rs]$, $\d\in (0, \dthick]$, $\t\in (\psi(\d), \ts(\e^*)]$,
and $\r\in [0, \r^{**}]$. Then, we have:
\begin{enumerate}
 \item   We have $1\leq|\ca C_\t(M_\r\mde)|\leq2$.
\item We have  $\ca C_\t(M_{\r^{**}}\mde)\persist \ca C_\t( M_{\r^{**}}\pde)$.
\item If
  $|\ca C_\t(M_\r\mde)|=2$, then   $\r\geq \r^*$ and
$\ca C_\t(M_{\r^{**}}\mde)\persist \ca C_\t( M_\r\mde) \persist \ca C( M_\r)$.
\end{enumerate}
\end{theorem}

\begin{proof}[Proof of Theorem \ref{main-thick}]
\ada i 
We first observe that $\d\leq \dthick$ implies
\begin{displaymath}
 \sup_{x\in M_\r}d(x,M_\r\mde) = \psis_{M_\r}(\d) \leq \cthick \d^\g<\infty\, ,
\end{displaymath}
and thus $M_\r\mde\neq\emptyset$, i.e.~$|\ca C_\t(M_\r\mde)| \geq 1$.
Conversely, we have
$|\ca C_\t(M_\r\mde)| \leq |\ca C(M_\r)|\leq 2$, where the first inequality was established in
part {\em ii)\/} of Lemma \ref{thick-impl-reg} and the second
is ensured by Definition 2.5.

\ada {ii} The monotonicity of $\ts$ established in Theorem \ref{reg-cluster-thm} yields
$\d<\psi(\d) <\t\leq \ts(\e^*)\leq \t_{M_{\rss}}^*/3$.
By part {\em iii)\/} of Lemma \ref{zeta-Td-surj} we then conclude that
the CRM $\ca C_\t(M_{\rss})\to \ca C_\t(M_\rss\pde)$
is bijective, and part \emph{ii)} of Theorem \ref{reg-cluster-thm} shows
$|\ca C_\t(M_{\rss})|= |\ca C_\t(M_{\rss}\pde)| = 2$.
By Lemma 2.4 it thus
 suffices to show that the CRM $\z: \ca C_\t(M_\rss\mde)\to \ca C_\t(M_{\rss})$
is bijective. Furthermore, if $|\ca C_\t(M_\rss\mde)|=1$, this map is automatically injective, and if
$|\ca C_\t(M_\rss\mde)|=2$, the injectivity follows from the surjectivity and the above proven
$|\ca C_\t(M_{\rss})|=2$.
Consequently, it actually suffices to show that
$\z$ is surjective. To this end, we fix a $B'\in \ca C_\t(M_{\rss})$ and an
 $x\in B'$. Then our assumption
ensures $d(x, M_\rss\mde) < \psi(\d)$, and hence there exists an $A'\in \ca C_\t(M_\rss\mde)$
with $d(x,A') < \psi(\d)$. Therefore, $\psi(\d) <\t$ implies that $x$ and $A'$ are $\t$-connected, which
yields $x\in A'$. In other words, we have shown $A'\cap B'\neq \emptyset$. By Lemma \ref{connect-char-cor} and the definition
of $\z$, we conclude that $\z(A') = B'$.

\ada {iii} We have $2=|\ca C_\t(M_\r\mde)|\leq |\ca C(M_\r)|\leq 2$, where the first inequality
was shown in part {\em ii)\/} of Lemma \ref{thick-impl-reg} and the second is guaranteed by Definition
2.5. We  conclude that
$|\ca C(M_\r)|  = 2$, and hence Definition
2.5 ensures
both $\r\geq \rs$ and the bijectivity of the CRM $\z_{\mathrm{top}}:\ca C(M_{\rss})\to \ca C(M_\r)$.
Furthermore, $|\ca C_\t(M_\r\mde)|= |\ca C(M_\r)|$, which has been shown above, together with
part {\em iii)\/} of Lemma \ref{thick-impl-reg} yields a bijective CRM
 $\z_\r : \ca C_\t(M_\r\mde) \to \ca C(M_\r)$, i.e.~the second persistence $C_\t( M_\r\mde) \persist \ca C( M_\r)$ is shown.
Moreover, part {\em ii)\/} of Theorem \ref{reg-cluster-thm} shows $|\ca C_\t(M_\rss\pde)|=2$,
and hence the already established   bijectivity of
$\z^{**}:\ca C_\t(M_{\r^{**}}\mde)\to \ca C_\t( M_{\r^{**}}\pde)$
gives $|\ca C_\t(M_\rss\mde)| = |\ca C_\t(M_\rss\pde)|=2 = |\ca C(M_\rss)|$.
Consequently, part {\em iii)\/} of Lemma \ref{thick-impl-reg} yields a bijective CRM
 $\z_{\rss} : \ca C_\t(M_\rss\mde) \to \ca C(M_{\rss})$.
Then the CRM
$\z: \ca C_\t(M_\rss\mde)\to \ca C_\t( M_\r\mde)$ enjoys the following diagram

\quadiass {\ca C_\t(M_\rss\mde)} {\ca C(M_{\rss})} {\ca C_\t(M_\r\mde)} {\ca C(M_{\r})} {\z_{\rss}} {\z} {\z_{\mathrm{top}}} {\z_\r}\\
whose commutativity follows from  Lemma 2.4.
Then the bijectivity of $ {\z_{\rss}}$, ${\z_{\mathrm{top}}}$, and  ${\z_\r}$
yields the bijectivity of $\z$, which completes the proof.
\end{proof}




\section{Additional Material Related to Thickness}\label{sec:proof-thick}




In this section we discuss some aspects related to the thickness assumption introduced
in Definition 2.6.

To this end, let $(X,d)$ be an arbitrary metric spaces and $A\subset X$.
We then define the function $\psis_A:(0,\infty)\to [0,\infty]$ by
\begin{displaymath}
\psis_A(\d) :=  \sup_{x\in M_\r} d(x, A\mde) \, , \qquad \qquad \d>0.
\end{displaymath}
Obviously, $\psis_{M_\r}$ coincides with the left-hand side of (2.5).
%
%

Our first observation is that the definition of $\psis_A$
immediately yields
$A\subset (A\mde)^{+\psis_A(\d)}$ for all $\d>0$ with $\psis_A(\d)<\infty$,
and it is also straightforward to see that $\psis_A(\d)$ is
the smallest $\psi>0$, for which   this inclusion holds, that is
\begin{displaymath}
 \psis_A(\d) = \min\bigl\{ \psi\geq 0: A\subset (A\mde)^{+\psi}  \bigr\}
\end{displaymath}
for all $\d>0$. In other words,  $\psis_A(\d)$ gives the size of the smallest tube needed to recover a superset of
$A$ from $A\mde$. In particular,
if $\d$ is too large, that is $A\mde = \emptyset$,  we obviously have $\psis_A(\d) = \infty$ and no recovery is
possible.

Intuitively it is not surprising that $\psis_A$ grows at least linearly, that is
\begin{equation}\label{at-least-linear}
 \psis_A(\d) \geq \d
\end{equation}
for all $\d>0$ provided that $d(A, X\setminus A) = 0$. Indeed, $\psis_A(\d) <\d$ for some $\d>0$ gives us an $\eps>0$ such that
$d(x,A\mde) < \d-\eps$ for all $x\in A$. Since $d(A, X\setminus A) = 0$ there then exists an $x\in A$ with
$d(x,X\setminus A)< \eps$, and for this $x$ there exists an $x'\in A\mde$ with $d(x,x')<\d-\eps$.
Now the definition of $A\mde$ gives $d(x', X\setminus A) > \d$, and hence we find a contradiction by
\begin{displaymath}
 \d < d(x', X\setminus A) \leq d(x', x) + d(x, X\setminus A) < \d\, .
\end{displaymath}

For generic sets $A$, the function $\psis_A$ is usually hard to bound, but for some
classes of sets, $\psis_A$ can be computed precisely. For example, for an interval
 $I=[a,b]$, we have  $\psis_I(\d) = \d$ for all $\d\in (0,(b-a)/2]$, and $\psis_I(\d) = \infty$, otherwise. Clearly, this
example can be extended to finite unions of such intervals and for intervals that are not closed, the only difference
occurs at $\d = (b-a)/2$. In higher dimensions, an interesting class of sets $A$ with linear behavior of $\psis_A$
is described by Serra's model, see \citeappendix[p.~144]{Serra82}, that consist of all compact sets $A\subset \Rd$ for
which there is a $\d_0>0$ with
\begin{displaymath}
 A = (A\omdeo)\opdeo = (A\opdeo)\omdeo\, .
\end{displaymath}
If, in addition, $A$ is path-connected, then \citeappendix[Theorem 1]{Walther99a} shows
that this relation also holds for all $\d\in (0,\d_0]$.
In this case,
 we then obtain
\begin{displaymath}
 A = (A\omdee)\opdee \subset  (A\omdee)\pdee \subset (A\mde)\pdee
\end{displaymath}
 for all $\d\in (0,\d_0)$ and   $0<\eps \leq \d_0-\d$.
In other words, we have $\psis_A(\d) \leq \d+\eps$, and
letting $\eps\to 0$, we thus conclude $\psis_A(\d) = \d$ for all $\d\in (0,\d_0)$.
With the help of Lemma \ref{Td-lemma-new}, it is not hard to see that this result generalizes to finite
unions of compact, path-connected sets, which has already been observed in \citeappendix{Walther99a}.
Finally, note that \citeappendix[Theorem 1]{Walther99a} also provides some useful characterizations of (path-connected) compact sets
belonging to Serra's model. In a nutshell, these  are the sets whose boundary is a $(d-1)$-dimensional sub-manifold of $\Rd$
with outward pointing unit normal vectors satisfying a Lipschitz condition.

Fortunately,
our analysis does not require the exact form of   $\psis_A$, but only its
asymptotic behavior for $\d\to 0$. Therefore, it is interesting to note that
 $\psis_A$ is also asymptotically invariant against bi-Lipschitz transformations.
To be more precise, let $(X,d)$ and $(Y,e)$ be two metric spaces and $I:X\to Y$ be a bijective  map
for which there exists a constant $C>0$ such that
\begin{displaymath}
 C^{-1}e(I(x), I(x')) \leq d(x,x') \leq C e(I(x), I(x'))
\end{displaymath}
for all $x,x'\in X$. For $A\subset X$ and $\d>0$, we then
 have $I(A^{+\d/C}) \subset (I(A))\pde \subset I(A^{+C\d})$, which in turn implies
\begin{displaymath}
 C^{-1} \psis_A(\d/C) \leq \psis_{I(A)}(\d) \leq C\psis_A(C\d)
\end{displaymath}
for all $\d>0$. In particular, we have $\psis_A(\d) \preceq \d^\g$ for some $\g\in (0,1]$
if and only if $\psis_{I(A)}(\d) \preceq \d^\g$.

Last but not least we like to mention that based on the sets $A\subset \R^2$ considered in
Example \ref{Aomde}, Example \ref{psis-in-r2} estimates $\psis_A$. In particular,
this example provides various sets $A$ with $\psis_A(\d) \sim \d$
that do not belong to Serra's model, and this class of sets can be further expanded by
using bi-Lipschitz transformations as discussed above.

Now consider  Definition 2.6, which
excludes  thin cusps and bridges, where the thinness and length of both is
controlled by $\g$.
Such assumptions have been widely used in the literature
on level set estimation and density-based clustering.
For example,
a basically identical assumption has been made in \citeappendix{SiScNo09a} for the exponent $\g=1$,
which can be taken, if, e.g.,
 the level sets belong to Serra's model. Moreover, level sets belonging to Serra's model
have
been investigated in \citeappendix{Walther97a}. In particular, \citeappendix[Theorem 2]{Walther97a}
shows that most level sets of a $C^1$-density with Lipschitz continuous gradient
belong to Serra's model.
Unfortunately, however, levels at which the density has a saddle point are excluded in this theorem,
and some other elementary sets such as cubes in $\Rd$ do not belong
to Serra's model, either. For this reason, we allow constants $\cthick>1$ in Definition 2.6.
Moreover, the  exponent $\g<1$ is allowed to provide more flexibility in situations,
in which very thin bridges are expected. However, based on the
discussion on $\psis_A$  as well as the examples provided in Section \ref{suppB:cons-dens}, we
strongly believe, that
in most cases assuming $\g=1$ is reasonable.
With the help of the discussion on $\psis_A$ it is also
easy to see that we have $M_\r \subset (M\mde)^{+\psi(\d)/2}$ for all $\d \in (0,\dthick]$
and all $\r\in (0,\rss]$. In addition, it becomes clear that exponents $\g>1$ are
impossible as soon as $d(M_\r, X\setminus M_\r)=0$ for some $\r\in (0,\rss]$.
Finally, recall that a   less geometric assumption excluding thin features  has been
used by various authors, see  e.g.~\citeappendix{CuFr97a,CuFeFr00a,Rigollet07a} and the references therein,
and an overview of these and similar assumptions can be found in \citeappendix{Cuevas09a}.

Understanding (2.5) in  the one-dimensional case is very simple. Indeed,
if $X\subset \R$ is an interval and $P$ can be topologically clustered
between   $\r^*$ and $\r^{**}$, then, for all $\r\in [0,\rss]$,
the level set $M_\r$ consists of either one or two closed intervals.
Using this, the discussion on $\psis_A$ shows
 that $P$ actually has thick levels of order $\g=1$ up to the level $\rss$. Moreover,
a possible thickness function is
$\psi(\d)= 3\d$ for all $\d\in (0,\dthick]$, where $\dthick$ equals the smaller radius of the two intervals at level $\rss$.

Finally, using the discussion on $\psis_A$ it is not
hard to construct distributions with discontinuous densities that have
thick levels of order, e.g.~$\g=1$. For continuous densities, however, this task is significantly
harder due to the above mentioned saddle point effects at the critical level $\rs$.
Therefore, we have added Example \ref{cont-h-in-r2}, which
provides a large class of such densities in the case $X\subset \R^2$.




\section{Proofs and Results Related to   Algorithm 2.1}\label{app-sec:generic-cluster-algo}




The main goals of this section is to prove Theorem  2.8 and
to provide background material from \citeappendix{Steinwart11a} for the proof of Theorem
 2.9.

\begin{lemma}\label{connect-main-new}
 Let $(X,d)$ be a compact metric space and $\mu$ be a finite Borel measure on $X$
with $\supp \mu = X$.
Moreover, let $P$ be a
 $\mu$-absolutely continuous  distribution on $X$, and
 $(L_\r)_{\r\geq 0}$ be a decreasing family of sets $L_\r\subset X$ such that
 \begin{displaymath}
     M_{\r+\e}\mde\subset L_\r \subset  M_{\r-\e}\pde
 \end{displaymath}
for  some fixed $\d>0$,  $\e \geq 0$, and all  $\r\geq 0$.
For some fixed  $\r\geq 0$ and $\t>0$,  let
$\z:\ca C_\t(M_{\r+\e}\mde)\to \ca C_\t(L_\r)$ be the CRM.
Then we have:
\begin{enumerate}
\item For all $A'\in \ca C_\t(M_{\r+\e}\mde)$ with $A' \cap M_{\r+3\e}\mde\neq \emptyset$ we have
          $\z(A') \cap   L_{\r+2\e} \neq \emptyset$.
\item For all $B'\in \ca C_\t(L_\r)$ with $B'\not \in \z(\ca C_\t(M_{\r+\e}\mde))$, we have
        \begin{align}\label{connect-main-new-h1}
        B' & \subset (X\setminus M_{\r+\e})\pde \cap M_{\r-\e}\pde\\ \label{connect-main-new-h2}
        B' \cap  L_{\r+2\e} & \subset  (X\setminus M_{\r+\e})\pde \cap M_{\r+\e}\pde\, .
        \end{align}
\end{enumerate}
\end{lemma}

\begin{proof}[Proof of Lemma \ref{connect-main-new}]
\ada {i} Using the  CRM property  $A'\subset \z(A')$
and the inclusion $M_{\r+3\e }\mde\subset  L_{\r+2\e}$,
we obtain
\begin{displaymath}
   \emptyset \neq  A'\cap M_{\r+3\e}\mde \subset \xi(A') \cap L_{\r+2\e}\, .
\end{displaymath}

\ada {ii} We fix a
$B'\in \ca C_\t(L_\r)\setminus \z(\ca C_\t(M_{\r+\e}\mde))$.
For $x\in B'$ we then have
\begin{displaymath}
   x\not \in \bigcup_{A'\in \ca C_\t(M_{\r+\e}\mde)}\z(A')\, ,
\end{displaymath}
and hence the
CRM property yields
\begin{displaymath}
 x\not \in \bigcup_{A'\in \ca C_\t(M_{\r+\e}\mde)}A' = M_{\r+\e}\mde\, .
\end{displaymath}
This shows
$x\in (X\setminus M_{\r+\e})\pde$, i.e.~we have proved $B'\subset (X\setminus M_{\r+\e})\pde$.
Now, \eqref{connect-main-new-h1} follows from $B'\subset L_\r\subset M_{\r-\e}\pde$, and
\eqref{connect-main-new-h2} follows from $B' \cap L_{\r+2\e}\subset L_{\r+2\e}\subset M_{\r+\e}\pde$.
\end{proof}

\begin{proof}[Proof of Theorem 2.8]
We first establish the following {\em disjoint\/} union:
\begin{align} \nonumber
 \ca C_\t(L_\r) = \z(\ca C_\t(M_{\r+\e}\mde))
 &\cup \bigl\{B'\in \ca C_\t(L_\r)\setminus \z(\ca C_\t(M_{\r+\e}\mde)):  B'\cap  L_{\r+2\e} \neq \emptyset \bigr\} \\ \label{ghost-motiv-2}
& \cup \bigl\{B'\in \ca C_\t(L_\r):  B'\cap  L_{\r+2\e} = \emptyset \bigr\}\, .
\end{align}
We begin by   showing the auxiliary result
\begin{equation}\label{algo-anal-h1}
 A' \cap M_{\r+3\e}\mde \neq \emptyset\, , \qquad \qquad A'\in \ca C_\t(M_{\r+\e}\mde).
\end{equation}
To this end, we observe that  {\em i)\/} and {\em ii)\/} of
Theorem \ref{reg-cluster-thm} yield $|\ca C_\t(M_\rss\pde)|=2$,
and hence part {\em ii)\/} of Theorem \ref{main-thick} implies $|\ca C_\t(M_\rss\mde)|=2$.
Let $W'$ and $W''$ be the two  $\t$-connected components of $M_\rss\mde$.
 We first assume that $M_{\r+\e}\mde$ has exactly one $\t$-connected component $A'$,
i.e.~$A'= M_{\r+\e}\mde$. Then $\r + 3\e\leq \rss$ and $\r+\e\leq \r+3\e$
imply
\begin{displaymath}
 \emptyset \neq M_\rss\mde \subset M_{\r+3\e}\mde = M_{\r+\e}\mde \cap M_{\r+3\e}\mde =  A'\cap M_{\r+3\e}\mde\, ,
\end{displaymath}
i.e.~we have shown (\ref{algo-anal-h1}).
Let us now assume that  $M_{\r+\e}\mde$ has more than one $\t$-component. Then it has exactly two
such components $A'$ and $A''$ by $\r+\e<\rss$ and part {\em i)\/} of Theorem \ref{main-thick}.
By
part  {\em iii)\/} of Theorem \ref{main-thick}  we may then assume without loss of generality that we have $W'\subset A'$
and $W''\subset A''$. Since $\r + 3\e\leq \rss$ implies
$M_\rss\mde \subset M_{\r+3\e}\mde$,
these inclusions yield $\emptyset \neq W' = W' \cap  M_\rss\mde\subset A'\cap  M_{\r+3\e}\mde$
and $\emptyset \neq W'' = W'' \cap  M_\rss\mde\subset A''\cap  M_{\r+3\e}\mde$. Consequently,
we have proved (\ref{algo-anal-h1}) in this case, too.

Now, from (\ref{algo-anal-h1}) we conclude by part {\em i)\/} of Lemma \ref{connect-main-new}
that $B'\cap  L_{\r+2\e} \neq \emptyset$ for all $B'\in \z(\ca C_\t(M_{\r+\e}\mde))$.
This yields
\begin{align*}
&\bigl\{B'\in \ca C_\t(L_\r)\setminus \z(\ca C_\t(M_{\r+\e}\mde)):  B'\cap  L_{\r+2\e} = \emptyset\bigr\} \\
&= \bigl\{B'\in \ca C_\t(L_\r) :  B'\cap  L_{\r+2\e} = \emptyset\bigr\}\, ,
\end{align*}
which in turn   implies \eqref{ghost-motiv-2}.

Let us now show (2.8).
Clearly, by \eqref{ghost-motiv-2}
it remains to show
\begin{displaymath}
 B' \cap  L_{\r+2\e}  = \emptyset\, ,
\end{displaymath}
for all $B'\in \ca C_\t(L_\r)\setminus \z(\ca C_\t(M_{\r+\e}\mde))$.
Let us  assume  the converse, that is, there exists a
$B'\in \ca C_\t(L_\r)\setminus \z(\ca C_\t(M_{\r+\e}\mde))$
with
$B'\cap  L_{\r+2\e}  \neq \emptyset$.
Since  $L_{\r+2\e}  \subset M_{\r+\e}\pde$, there then exists
an $x\in B' \cap M_{\r+\e}\pde$.
By part \emph{i)} of  Lemma \ref{Td-lemma-new}
this gives an $x'\in M_{\r+\e}$ with $d(x,x')\leq \d$, and hence
we obtain
\begin{displaymath}
 d(x', M_{\r+\e}\mde ) \leq \psis_{M_{\r+\e}}(\d) \leq \cthick \d^\g
< 2\cthick \d^\g\, .
\end{displaymath}
From this inequality we conclude that there exists an $x'' \in M_{\r+\e}\mde$ satisfying
$d(x',x'') < 2\cthick \d^\g$. Let $A''\in \ca C_\t(M_{\r+\e}\mde)$ be the unique $\t$-connected component satisfying
$x''\in A''$. The CRM property then yields $x'' \in A''\subset \z(A'') =:B''$, and thus,
using $c\geq 1$,
we find
\begin{displaymath}
 d(B', B'') \leq d(x, x'') \leq d(x,x') + d(x',x'') < \d + 2\cthick \d^\g \leq 3\cthick \d^\g < \t\, .
\end{displaymath}
However, since $B'\not \in \z(\ca C_\t(M_{\r+\e}\mde))$ and $B''\in \z(\ca C_\t(M_{\r+\e}\mde))$ we obtain
$B'\neq B''$, and hence Lemma \ref{finite-comp}
yields $d(B', B'') \geq \t$.
\end{proof}

\begin{theorem}\label{analysis-main-new-suppl}
Let Assumption  C  be satisfied. Furthermore, let
  $\e^*\leq (\rss - \rs)/9$ ,
 $\d\in (0,  \dthick]$,
$\t\in (\psi(\d),\ts(\e^*)]$, and  $\e\in (0 , \e^*]$.
In addition, let $D$ be a data set and $(L_{D,\r})_{\r\geq 0}$ be a
decreasing family satisfying
\begin{displaymath}
    M_{\r+\e}\mde\subset L_{D,\r} \subset  M_{\r-\e}\pde
\end{displaymath}
for all $\r\geq 0$.
Furthermore, assume that Algorithm 2.1 receives the
parameters $\t$, $\e$, and  $(L_{D,\r})_{\r\geq 0}$.
Then,
the following statements are true:
\begin{enumerate}
 \item The returned level $\rds$ satisfies $\rds \in [\rs+2\e , \rs+\e^*+5\e]$.
 \item We have $|\ca C_\t(M_\rdse\mde)| =2$ and the CRM
 $\z:\ca C_\t (M_\rdse\mde)\to \ca C_\t(L_{D,\rds})$ is injective.
 \item Algorithm 2.1 returns the two  $\t$-connected components
 of $\z(C_\t (M_\rdse\mde))$.
\item There exist
CRMs
$\z_\rss : \ca C_\t(M_\rss\mde)\to \ca C(M_\rss)$ and
$\z_\rdse : \ca C_\t(M_\rdse\mde)\to \ca C(M_\rdse)$ such that
we have a commutative diagram of bijective CRMs:
\quadiass {\ca C_\t(M_\rss\mde)}{\ca C(M_\rss)}{\ca C_\t(M_\rdse\mde)} {\ca C(M_\rdse)}  {\z_\rss} {\z_{\rss, \rds+\e}} {\tilde \z} {\z_\rdse}\\
%
\end{enumerate}
\end{theorem}

\begin{proof}[Proof of Theorem \ref{analysis-main-new-suppl}]
We begin with some general observations. To this end, let
$\r\in [0,\rss-4\e]$ be the   level that is currently considered in Line 3 of Algorithm 2.1.
Then, Theorem 2.8 shows that   Algorithm 2.1
identifies exactly the $\t$-connected components of $L_{D,\r}$  that belong to the set
$\z(\ca C_\t(M_{\r+\e}\mde))$, where $\z:\ca C_\t (M_{\r+\e}\mde)\to \ca C_\t(L_{D,\r})$ is  the CRM.
In the following, we thus consider the set $\z(\ca C_\t(M_{\r+\e}\mde))$. 
Moreover, we note that the returned level $\rds$ always satisfies $\rds\geq \r + 3\e$ by Line 4 and Line 6,
and equality holds if and only if $|\z(\ca C_\t(M_{\r+\e}\mde))|\neq 1$.

\ada i
Let us first consider the case {$\r \in [0,\rs-\e)$.}
Then $\r+\e  <\rs$ together with
 part {\em i)\/} and {\em iii)\/} of Theorem \ref{main-thick} shows $|\ca C_\t(M_{\r+\e}\mde)|=1$, and hence
 $|\z(\ca C_\t(M_{\r+\e}\mde))|=1$. Our initial consideration then show, that
 Algorithm 2.1 does not leave its loop, and thus $\rds\geq  \rs + 2\e$.

Let us now consider the case {$\r \in [\rs+\e^*+\e , \rs+\e^*+2\e]$.}
Here we first note that Algorithm 2.1 actually inspects such an $\r$, since
it iteratively inspects all $\r = i\e$, $i=0,1,\dots$, and the  width of the interval above  is $\e$.
Moreover, our assumptions on $\e^*$ and $\e$ guarantee $\rs+\e^*+2\e\leq \rss-4\e$, and hence we have
$\r \in [\rs+\e^*+\e , \rss-4\e]$, i.e., we are in the situation described at the beginning of the proof.
We  write  $\z_{V}: \ca C_\t(M_\rss\mde) \to \ca C_\t(M_{\r+\e}\mde)$,
$\z_{M}: \ca C_\t(M_\rss\pde) \to \ca C_\t(M_{\r-\e}\pde)$,
and $\z_{V,M}: \ca C_\t(M_{\r+\e}\mde) \to \ca C_\t(M_{\r-\e}\pde)$ for the CRMs between the involved sets.
We then obtain the commutative diagram

\quadiann {\ca C_\t(M_{\r+\e}\mde)} {\ca C_\t(M_{\r-\e}\pde)} {\ca C_\t(M_\rss\mde)} {\ca C_\t(M_\rss\pde)} {\z_{V,M}} {\z_V} {\z_M} {\z^{**}}\\
where the CRM $\z^{**}$ is bijective by part {\em ii)\/}  of Theorem \ref{main-thick}.
Moreover, $\r-\e \geq \rs + \e^*$ together with  part {\em ii)\/}  of Theorem \ref{reg-cluster-thm}
shows $|\ca C_\t(M_{\r-\e}\pde)|=2$, and by {\em iii)\/}  of Theorem \ref{reg-cluster-thm}
we conclude that $\z_M$ is bijective. Similarly, $\r+\e \geq \rs+\e^*$ and the bijectivity of $\z^{**}$
show by {\em iv)\/}  of Theorem \ref{reg-cluster-thm} that $|\ca C_\t(M_{\r+\e}\mde)|=2$, and thus
$\z_V$ is bijective by  part {\em iii)\/}  of Theorem \ref{main-thick}.
Consequently, $\z_{V,M}$ is bijective. Let us further consider the CRM $\z':\ca C_\t(L_{D,\r}) \to \ca C_\t(M_{\r-\e}\pde)$.
Then Lemma 2.4 yields another diagram:

\tridia {\ca C_\t(M_{\r+\e}\mde)} {\ca C_\t(M_{\r-\e}\pde)} {\ca C_\t(L_{D,\r})}  {\z_{V,M}} \z{\z'}\\
Since $\z_{V,M}$ is bijective, we then find  that $\z$ is injective, and since
we have already seen that $M_{\r+\e}\mde$ has two $\t$-connected components, we conclude
that $\z(\ca C_\t(M_{\r+\e}\mde))$ contains two elements. Consequently, the stopping criterion of Algorithm 2.1
 is satisfied, that is, $\rds = \r+3\e  \leq \rs+\e^*+5\e$.

\ada {ii}
Theorem 2.8 shows that in its last run through the loop Algorithm 2.1
identifies exactly the $\t$-connected components of $L_{D,\r}$ that belong to the set
$\z_{-3\e}(\ca C_\t(M_{\r +\e}\mde)$, where
$\r := \rds-3\e$ and
$\z_{-3\e}:\ca C_\t (M_{\r+\e}\mde)\to \ca C_\t(L_{D,\r})$ is  the CRM.
Moreover, since Algorithm 2.1 stops at $\rds-3\e$, we have
$|\z_{-3\e}(\ca C_\t(M_{\r+\e}\mde))|\neq 1$ by our remarks at the beginning of the proof,
and thus
$|\ca C_\t(M_{\r+\e}\mde)|\neq 1$.
From the already proven part \emph{i)} we further know that
$\r +\e = \rds-2\e \leq \rs + \e^* + 3\e\leq \rs + 4\e^*\leq  \rss$, and part {\em i)\/} of Theorem \ref{main-thick}
hence gives
$|\ca C_\t(M_{\r+\e}\mde)|= 2$. For later purposes, note that the latter
together with $|\z_{-3\e}(\ca C_\t(M_{\r+\e}\mde))|\neq 1$
implies the injectivity of $\z_{-3\e}$.
Now, part {\em iii)\/} of Theorem \ref{main-thick}
shows that the CRM $\z_{\rss, \r +\e}: {\ca C_\t(M_\rss\mde)}\to  {\ca C_\t(M_{\r+\e}\mde)}$ is bijective.
Let us   consider the following commutative diagram:

\tridia {\ca C_\t(M_\rss\mde)} {\ca C_\t(M_{\r+\e}\mde)} {\ca C_\t(M_\rdse\mde)}   {\z_{\rss, \r +\e}} {\z_{\rss, \rdse}}{\tilde \z}\\
where the remaining two maps are the corresponding CRMs,
whose existence is guaranteed by $\rds+\e\leq \rds + 7\e^*\leq \rss$ and $\r+\e\leq \rds + \e$, respectively.
Now the bijectivity of $\z_{\rss, \r +\e}$ shows
that $\z_{\rss, \rdse}$ is injective. Moreover,  $\rdse \leq \rss$ implies
$|\ca C_\t(M_\rdse\mde)|\leq 2$ by part {\em i)\/} of Theorem \ref{main-thick},
while $\rss \geq \rs+\e^*$ implies
 $|\ca C_\t(M_\rss\mde)|=2$ by part \emph{iv)} of Theorem \ref{reg-cluster-thm} and part \emph{ii)} of
Theorem \ref{main-thick}.
Therefore,
$\z_{\rss, \rdse}$ is actually bijective. This yields both $|\ca C_\t(M_\rdse\mde)| =  2$, which is the first assertion,
and the bijectivity of $\tilde \z$.
Let us  consider yet another
commutative diagram

\quadiass {\ca C_\t(M_\rdse\mde)} {\ca C_\t(M_{\r+\e}\mde)} {\ca C_\t((L_{D,\rds})} {\ca C_\t(L_{D,\r})} {\tilde \z} {\z} {\z_{-3\e}} {\z'}\\
where again, all occurring maps are the CRMs between the respective sets.
Now we have already shown that $\z_{-3\e}$ is injective and that  $\tilde \z$ is bijective. Consequently, $\z$ is injective.

\ada {iii} This assertions
follows from Theorem 2.8 and the inequality $\rds  \leq \rss-3\e$, which follows from
part \emph{i)}.

\ada {iv} We have already seen in the proof of part \emph{ii)} that $|\ca C_\t(M_\rss\mde)|=2$, and consequently
part {\em iii)\/} of Lemma \ref{thick-impl-reg} shows that there exists a bijective
CRM $\z_{\rss} : \ca C_\t(M_\rss\mde)\to \ca C(M_{\rss})$.
Moreover,  part \emph{ii)} shows $|\ca C_\t(M_{\rds+\e}\mde)|=2$, thus
part {\em iii)\/} of Lemma \ref{thick-impl-reg} yields another bijective CRM
$\z_{\rdse} : \ca C_\t(M_\rdse\mde)\to \ca C(M_{\rdse})$.
Furthermore, in the proof of part {\em ii)\/} we have already seen that
 CRM $\z_{\rss, \rdse}$
is bijective. This gives the  diagram.
\end{proof}




\section{Additional Material Related to Assumption A}




In this section we discuss Assumption A, which describes the partitions needed for our
histogram approach, in more detail.

We begin with an example  of partitions satisfying Assumption A.

\begin{example}\label{ex:unif-part}
Let $X:= [0,1]^\dpart$ be  equipped with the metric defined by the  supremum norm $\snorm \cdot_{\ell_\infty^\dpart}$,
and $\lb^\dpart$ be the   Lebesgue measure.
For $\d\in (0,1]$, there then exists a unique
 $\ell\in \N$ with $\frac 1 {\ell+1}<\d\leq \frac 1 \ell$.
We define  $h:= \frac 1 {1+\ell}$
and write $\ca A_\d$ for  the usual partition of $[0,1]^\dpart$ into hypercubes of side-length $h$.
Then, for each $A_i\in \ca A_\d$, we   have $\diam A_i = h \leq \d$ and $\lb^\dpart(A_i) = h^{\dpart} \geq 2^{-d} \d^\dpart$.
Moreover, we obviously have $|\ca A_\d| = h^{-\dpart} \leq 2^\dpart\d^{-\dpart}$, and hence
$(\ca A_\d)_{\d\in (0,1]}$ satisfies Assumption A with
 $\cpart:= 2^\dpart$.
\end{example}

The next lemma describes a general situation in which there exist  partitions satisfying Assumption A.
For its formulation, recall that the covering numbers of a compact metric space $(X,d)$ are defined by
\begin{displaymath}
   \ca N(X,d,\d) := \min\biggl\{n\geq 1: \exists x_1,\dots,x_n \in X \mbox{ with } X\subset \bigcup_{i=1}^n B(x_i, \d) \biggr\} ,\, \d>0,
\end{displaymath}
where again $B(x,\d)$ denotes the closed ball with center $x$ and radius $\d$.

\begin{lemma}\label{lemma-assA}
   Let $(X,d)$ be a compact metric space for which there exist constants $c>0$ and $\dpart>0$ such that
   \begin{displaymath}
      \ca N(X,d,\d) \leq c \d^{-\dpart}\, , \qquad \qquad \d\in (0,1/4].
   \end{displaymath}
Moreover, assume that there exists a finite measure $\mu$ on $X$ such that
\begin{displaymath}
   \mu(B(x,\d)) \geq c^{-1} \d^\dpart
\end{displaymath}
for all
$x\in X$ and $\d\in (0,1/4]$. Then Assumption A is satisfied for $\dpart$ and  $\cpart = 4^\dpart c$.
\end{lemma}

Note that the unit spheres $\mathbb S^{d}\subset \R^{d+1}$
together with their surface measures
satisfy the assumptions for $\dpart = d-1$, see also Corollary \ref{cor-assA}.

\begin{proof}[Proof of Lemma \ref{lemma-assA}]
   Let us recall that a $\d$-packing in $X$ is a family $y_1,\dots,y_m\in X$ with $d(y_i,y_j) > 2\d$ for all $i\neq j$.
   Let us write
   \begin{displaymath}
      \ca M(X,d,\d) := \max\Bigl\{ m\geq 1: \exists \mbox{$\d$-packing } y_1,\dots,y_m \mbox{ in } X \Bigr\}
   \end{displaymath}
  for the size of the largest possible $\d$-packing in $X$.
  Then it is well-known that we have the following inequalities between these packing numbers and the covering numbers:
%
  \begin{equation}\label{lemm-assA-h1}
     \ca M(X,d,\d) \leq \ca N(X,d,\d) \leq \ca M(X,d,\d/2)\, , \qquad \qquad \d>0.
  \end{equation}
  Let us now fix a $\d\in (0,1]$ and a maximal $\d/4$-packing $y_1,\dots,y_m$ in $X$. By  \eqref{lemm-assA-h1} we conclude that
  \begin{displaymath}
     m = \ca M(X,d,\d/4) \leq \ca N(X,d,\d/4) \leq 4^\dpart c \d^{-\dpart}\, .
  \end{displaymath}
  To construct the partition $\ca A_\d$, we consider
  a Voronoi partition $A_1,\dots,A_m$ that corresponds to the points $y_1,\dots,y_m$, where the behavior of the cells
  on their boundary may be arbitrary, i.e.~ties may be arbitrarily resolved. Our next goal is to show
  \begin{equation}\label{lemm-assA-h2}
     B(y_i, \d/4) \subset A_i \subset B(y_i, \d/2) \, , \qquad \qquad i=1,\dots,m.
  \end{equation}
  To prove the left inclusion, we fix an $x\in B(y_i,\d/4)$. For $j\neq i$, we then find
  \begin{displaymath}
     \d/2 < d(y_i,y_j) \leq d(y_i, x) + d(x,y_j) \leq \d/4 + d(x,y_j)\, ,
  \end{displaymath}
  and hence $d(x,y_j) > \d/4 \geq d(x,y_i)$. From the latter  we   conclude that $x\in A_i$.

  For the proof of the right inclusion, we assume that it does not hold for some
  index $i\in \{1,\dots,m\}$. Then there exists an $x\in A_i$ such that
  $d(x,y_i)> \d/2$. On the hand, since $y_1,\dots,y_m$   is a  maximal $\d/4$-packing in $X$, there exists a
  $j\in \{1,\dots,m\}$ with $d(x,y_j) \leq 2 \d/4 = \d/2$, and hence
  we have $d(x,y_j) \leq \d/2< d(x,x_i)$.
  This implies $x\not \in A_i$, i.e.~we have found a contradiction.

  Now, using \eqref{lemm-assA-h2}, we obtain both $\mu(A_i) \geq \mu( B(y_i, \d/4)) \geq 4^{-\dpart} c^{-1} \d^\dpart$
  and $\diam A_i \leq \diam B(y_i, \d/2) \leq \d$.
\end{proof}

The next corollary
in particular shows that one of the assumptions made in Lemma \ref{lemma-assA} can be omitted
if the measure behaves regularly on balls.

\begin{corollary}\label{cor-assA}
  Let $(X,d)$ be a compact metric space and $\mu$ be a finite measure on $X$ for which there exists a constant $K\geq 1$ such hat
   \begin{displaymath}
    K^{-1} \leq  \frac{\mu( B(y, \d))}{\mu(B(x,\d))} \leq K\, , \qquad \qquad x,y\in X, \, \d\in (0,1/4].
  \end{displaymath}
%
   If  there exist constants $c>0$ and $\dpart>0$ such that
   \begin{displaymath}
      \ca N(X,d,\d) \leq c \d^{-\dpart}\, , \qquad \qquad \d\in (0,1/4],
   \end{displaymath}
then Assumption A is satisfied for $\dpart$ and  $\cpart = 4^\dpart cK$. Similarly, if
\begin{displaymath}
    \mu(B(x,\d)) \geq c^{-1} \d^\dpart\, , \qquad \qquad \d\in (0,1/8],
\end{displaymath}
holds true, then Assumption A is satisfied for $\dpart$ and  $\cpart = 8^\dpart cK$.
\end{corollary}

If
$(X,d,\cdot)$ is a compact group with invariant metric $d$ and $\mu$ is its Haar measure, then
we have $K=1$.
Moreover, if $X\subset \R^d$ is a sufficiently smooth manifold
and $\mu$ is its surface measure, then the corollary is also applicable.

\begin{proof}[Proof of Corollary \ref{cor-assA}]
%
  To show the first assertion,
  we fix a $\d\in (0,1/4]$ and a minimal $\d$-net  $x_1,\dots, x_n$ of $X$.
  For an   $x\in X$ we then obtain
  \begin{displaymath}
     1 = \mu(X) \leq \sum_{i=1}^n  \mu(B(x_i,\d)) \leq  n K \mu(B(x,\d)) \leq  cK \d^{-\dpart} \mu(B(x,\d))\, .
  \end{displaymath}
  By Lemma \ref{lemma-assA} we thus obtain the first assertion.

  To prove the second assertion we fix a $\d\in (0,1/4]$ and a maximal $\d/2$-packing $y_1,\dots,y_m$ of $X$.
  Then $B(y_i,\d/2)\cap B(y_j, \d/2)= \emptyset$  for $i\neq j$ implies
    \begin{displaymath}
     1 = \mu(X) \geq \sum_{i=1}^m  \mu(B(y_i,\d/2)) \geq m K^{-1}\mu(B(x,\d/2)) \geq  m 2^{-\dpart} c^{-1} K^{-1}\d^\dpart\, ,
  \end{displaymath}
  and hence $\ca N(X,d,\d) \leq \ca M(X,d,\d/2) = m \leq 2^\dpart cK \d^{-\dpart}$ by \eqref{lemm-assA-h1}.
   Lemma \ref{lemma-assA} then yields the second assertion.
\end{proof}




\section{Material Related to Basic Properties of Histograms}




The goal of this section is to establish
the key inclusion  (2.7) for our histogram-based approach.
The material of this section is taken from
\citeappendix{Steinwart11a}.

%

Our first result shows that $ h_{D, \d}$
\emph{uniformly} approximates its infinite-sample counterpart
\begin{displaymath}
 h_{P, \d}(x) :=  \sum_{j=1}^m \frac {P(A_j)} {\mu(A_j)} \cdot  \eins_{A_j}(x) \qspace x\in X,
\end{displaymath}
with high probability, where $\ca A_\d= (A_1,\dots,A_m)$ for a fixed $\d>0$.

\begin{theorem}\label{uniform-approx}
Let Assumption A be satisfied and
 $P$ be a distribution on $X$.
Then, for all $n\geq 1$, $\e>0$, and $\d>0$, we have
\begin{displaymath}
 P^n\Bigl(\bigl\{ D\in X^n: \inorm {h_{D, \d}- h_{P, \d}} \geq \e \bigr\}\Bigr)
 \leq 2 \cpart \exp\Bigl(- \dpart \ln \d-  \frac {2 n\e^2\d^{2\dpart}}{\cpart^{2}}\Bigr)\, .
\end{displaymath}
In addition, if $P$ is $\mu$-absolutely continuous and there
exists a bounded $\mu$-density $h$ of $P$, then, for all $n\geq 1$, $\e>0$, and $\d>0$, we have
\begin{displaymath}
 P^n\Bigl( D\in X^n: \inorm {h_{D, \d}- h_{P, \d}} \geq \e \Bigr)
 \!\leq\!  2 \cpart \exp\biggl(\!  \ln \d^{-\dpart} -  \frac {3n\e^2 \d^{\dpart} }{\cpart (6  \inorm h + 2 \e )   }\!\biggr)   .
\end{displaymath}
\end{theorem}

\begin{proof}[Proof of Theorem \ref{uniform-approx}]
 We fix an $A\in \ca A_\d$ and write $f:= \mu(A)^{-1}\eins_A$.
Then $f$ is non-negative and our assumptions ensure $\inorm f \leq \cpart \d^{-\dpart}$.
Consequently, Hoeffding's inequality, see e.g.~\citeappendix[Theorem 8.1]{DeGyLu96}, yields
\begin{displaymath}
 P^n\biggl(\Bigl\{ D\in X^n: \Bigl|\frac 1 n \sum_{i=1}^n f(x_i) - \E_Pf  \Bigr| < \e \Bigr\}\biggr)
\geq 1 -
2 \exp\Bigl( -  \frac {2 n\e^2\d^{2\dpart}}{\cpart^{2}}\Bigr)
\end{displaymath}
 for all $n\geq 1$ and $\e>0$, where we assumed $D=(x_1,\dots,x_n)$.
 Furthermore, we have $\frac 1 n \sum_{i=1}^n f(x_i) = \mu(A)^{-1} D(A)$ and
$\E_P f = \mu(A)^{-1}P(A)$. By a union bound argument and $|\ca A_\d| \leq \cpart \d^{-\dpart}$, we thus obtain
\begin{displaymath}
 P^n\biggl(\!\Bigl\{ D\in X^n: \sup_{A\in \ca A_\d}\Bigl|\frac {D(A)} {\mu(A)} - \frac {P(A)}{\mu(A)} \Bigr| < \e \Bigr\}\!\biggr)
\geq 1 - 2 \cpart \d^{-\dpart}
\exp\Bigl( -  \frac {2 n\e^2\d^{2\dpart}}{\cpart^{2}}\Bigr)
\, .
\end{displaymath}
Since, for $x\in X$ and $A\in \ca A_\d$ with $x\in A$, we have $h_{D, \d}(x) =  \mu(A)^{-1} D(A)$
 and $h_{P, \d}(x) =  \mu(A)^{-1} P(A)$, we then find the first assertion.

To show the second inequality, we
 write $f:= \mu(A)^{-1}(\eins_A- P(A))$ for a fixed $A\in \ca A_\d$. This yields $\E_Pf=0$,
$\inorm f \leq \cpart \d^{-\dpart}$, and
\begin{displaymath}
 \E_P f^2 \leq \mu(A)^{-2} P(A) \leq \mu(A)^{-1} \inorm h \leq \cpart \d^{-\dpart} \inorm h\, .
\end{displaymath}
Consequently, Bernstein's inequality, see e.g.~\citeappendix[Theorem 8.2]{DeGyLu96},
yields
\begin{displaymath}
 P^n\biggl(\Bigl\{ D\in X^n: \Bigl|\frac 1 n \sum_{i=1}^n f(x_i) \Bigr| < \e \Bigr\}\biggr)
\geq 1 -
2 \exp\biggl( -  \frac {3n\e^2 \d^{\dpart} }{\cpart (6  \inorm h + 2 \e )   }\biggr)
\, .
\end{displaymath}
Using $\frac 1 n \sum_{i=1}^n f(x_i) = ( D(A) - P(A))\mu(A)^{-1}$,
the rest of the proof follows the lines of the proof of the first inequality.
\end{proof}

The next result  specifies the vertical and horizontal uncertainty of
a plug-in level set estimate $\{\hat h\geq \r\}$, if $\hat h$ is a uniform approximation
of $h_{P, \d}$.

\begin{lemma}\label{histo-include}
Let Assumption A be satisfied,
$P$ be a $\mu$-absolutely continuous  distribution on $X$, and
$\hat h:X\to \R$ be a function with $\inorm {\hat h-h_{P, \d}}\leq \e$ for some $\e\geq 0$.
Then, for all     $\r\geq 0$, the following statements hold:
\begin{enumerate}
 \item If $P$ is upper normal at the level $\r+\e$, then we have  $M_{\r+\e}\mde\subset \{ \hat h \geq \r\}$.
 \item If $P$ is upper normal at the level $\r-\e$, then we have $ \{\hat h \geq \r\}   \subset M_{\r-\e}\pde$.
\end{enumerate}
\end{lemma}

%

\begin{proof}[Proof of Lemma \ref{histo-include}]
\ada {i} We will show the equivalent inclusion $\{ \hat h < \r\}   \subset (X\setminus M_{\r+\e})\pde$. To this end,
we fix an $x\in X$ with $\hat h(x) <\r$.
If $x\in X\setminus M_{\r+\e}$, we immediately obtain
 $x\in  (X\setminus M_{\r+\e})\pde$,
and hence we may restrict our considerations to the case $x\in M_{\r+\e}$.
Then,   $\hat h(x) < \r$ together with  $\inorm {\hat h-h_{P,  \d}}\leq \e$ implies $h_{P, \d}(x) \leq \hat h(x) + \e < \r+\e$.
Now let $A$ be the unique cell of the partition $\ca A_\d$ satisfying $x\in A$.
The definition of $h_{P, \d}$ together with the assumed $0<\mu(A)<\infty$ then yields
\begin{equation}\label{include-h1}
  \int_{A} h\, d\mu = P(A) = \hpd(x) \mu(A)< (\r+\e) \mu(A)\, ,
\end{equation}
where $h:X\to [0,\infty)$ is an arbitrary $\mu$-density of $P$.
Our next goal is to show that there exists an $x'\in (X\setminus M_{\r+\e})\cap A$. Suppose the converse, that
is $A \subset M_{\r+\e}$.
Then the upper normality of $P$ at the level $\r+\e$ yields
$\mu(A\setminus \{h\geq \r+\e\}) \leq \mu(M_{\r+\e}\setminus \{h\geq \r+\e\})= 0$,
and hence we conclude that
$\mu(A\cap \{h\geq \r+\e\}) = \mu(A)$.
This leads to
\begin{displaymath}
 \int\limits_{A} h\, d\mu
= \int\limits_{A\cap \{h\geq \r+\e\} } h\, d\mu + \int\limits_{A\setminus \{h\geq \r+\e\} } h\, d\mu
= \int\limits_{A\cap \{h\geq \r+\e\} } h\, d\mu
\geq (\r+\e) \mu(A)\, .
\end{displaymath}
However, this inequality
contradicts (\ref{include-h1}), and hence there does exist an
$x'\in (X\setminus M_{\r+\e})\cap A$.
This implies
 $d(x,X\setminus M_{\r+\e}) \leq d(x,x') \leq \diam A \leq \d$, 
i.e.~we have shown $x\in  (X\setminus M_{\r+\e})\pde$.

\ada {ii}
Let us  fix an
$x\in X$ with $\hat h(x)\geq \r$.
If $x\in M_{\r-\e}$, we immediately obtain
 $x\in M_{\r-\e}\pde$,
and hence it remains to consider  the case $x\in X\setminus M_{\r-\e}$.
Clearly, if $\r-\e\leq 0$, this case is impossible, and hence we may additionally assume $\r-\e>0$.
Then,   $\hat h(x) \geq  \r$ together with  $\inorm {\hat h-h_{P, \d}}\leq \e$ yields
$h_{P, \d}(x) \geq \hat h(x) - \e \geq \r-\e$.
Now let $A$ be the unique cell of the partition $\ca A_\d$ satisfying $x\in A$.
By the definition of $h_{P, \d}$  and $\mu(A)>0$
we then obtain
\begin{equation}\label{include-h2}
  \int_{A} h\, d\mu = P(A)= \hpd(x) \mu(A) \geq  (\r-\e) \mu(A)\, ,
\end{equation}
where $h:X\to [0,\infty)$ is an arbitrary $\mu$-density of $P$.
Next we show that there exists an $x'\in M_{\r-\e}\cap A$. Suppose the converse holds, that
is $A \subset X\setminus M_{\r-\e}$.
Then the assumed upper normality of $P$ at the level $\r-\e$ yields
\begin{displaymath}
   \mu(M_{\r-\e} \symdif \{h\geq \r-\e\}) = 0\, ,
\end{displaymath}
and thus we find $ \mu((X\setminus M_{\r-\e}) \symdif \{h< \r-\e\}) = 0$ by $A\symdif B = (X\setminus A) \symdif (X\setminus B)$.
Combining this with the assumed $A \subset X\setminus M_{\r-\e}$, we obtain
\begin{displaymath}
 \mu\bigl(A\setminus \{h < \r-\e\}\bigr) \leq \mu\bigl((X\setminus M_{\r-\e})\setminus \{h< \r-\e\}\bigr)= 0\, ,
\end{displaymath}
and this implies
\begin{displaymath}
 \int\limits_{A} h\, d\mu
= \int\limits_{A\cap \{h< \r-\e\} } h\, d\mu + \int\limits_{A\setminus \{h<\r-\e\} } h\, d\mu
= \int\limits_{A\cap \{h< \r-\e\} } h\, d\mu
< (\r-\e) \mu(A)\, .
\end{displaymath}
This  contradicts (\ref{include-h2}), and hence there does exist an
$x'\in M_{\r-\e}\cap A$.
This yields
$ d(x,M_{\r-\e}) \leq d(x,x') \leq \diam A \leq \d$,
i.e.~we have shown $x\in  M_{\r-\e}\pde$.
\end{proof}




\section{Proofs and Additional Material  Related to the Consistency}\label{sec:proof-cons}




In this section we prove Theorem 4.1.
Furthermore, it contains additional material related to the assumptions made in that theorem.

\begin{lemma}\label{approx-lemma}
 Let $(X,d)$ be a   metric  space, $\mu$ be a finite Borel
measure on $X$, and $(A_\r)_{\r\in \R}$ be a decreasing  family of
closed subsets of $X$.
For $\rs\in \R$, we write
\begin{displaymath}
\dot A_\rs:= \bigcup_{\r>\rs}\mathring A_\r \qquad \qquad \mbox{ and } \qquad \qquad \hat A_\rs:= \bigcup_{\r>\rs}A_\r  \, .
\end{displaymath}
Then we have
\begin{align*}
  \dot A_\rs  = \bigcup_{\r>\rs}\bigcup_{\e>0}\bigcup_{\d>0} A_{\r+\e}\mde\, . 
\end{align*}
Moreover, the following statements are equivalent:
\begin{enumerate}
 \item $\mu(\hat A_\rs\setminus \dot A_\rs) = 0$.
 \item For all $\e>0$, there exists a $\r_\e>\rs$ such that, for all $\r\in (\rs,\r_\eps]$, we have $\mu(A_\r\setminus \mathring A_\r)\leq \e$.
\end{enumerate}
\end{lemma}

\begin{proof}[Proof of Lemma \ref{approx-lemma}]
 To show the first equality, we observe that (\ref{inf-tube}) implies
\begin{displaymath}
 \bigcap_{\r>\rs}\bigcap_{\e>0} \bigcap_{\d>0} (X\setminus A_{\r+\e} )\pde
= \bigcap_{\e>0} \bigcap_{\r>\rs} \overline{X\setminus A_{\r+\e} }
=  \bigcap_{\r>\rs} \overline{X\setminus A_{\r} }\, .
\end{displaymath}
Moreover, every set $A\subset X$ satisfies $\overline{X\setminus A} = X\setminus \mathring A$, and hence
we obtain
\begin{displaymath}
\bigcap_{\r>\rs}\bigcap_{\e>0} \bigcap_{\d>0} (X\setminus A_{\r+\e} )\pde
=  \bigcap_{\r>\rs} \overline{X\setminus A_{\r}}
= \bigcap_{\r>\rs} (X\setminus \mathring A_\r)
=  X\setminus \bigcup_{\r>\rs} \mathring A_\r\, .
\end{displaymath}
Therefore, by taking the complement we find
\begin{align*}
\bigcup_{\r>\rs} \mathring A_\r
= X\!\setminus\! \biggl( \bigcap_{\r>\rs}\bigcap_{\e>0} \bigcap_{\d>0} (X\setminus A_{\r+\e} )\pde\biggr)
 &= \bigcup_{\r>\rs}\bigcup_{\e>0}\bigcup_{\d>0}\bigl( X\setminus (X\setminus A_{\r+\e} )\pde\bigr)\\
 &= \bigcup_{\r>\rs}\bigcup_{\e>0}\bigcup_{\d>0} A_{\r+\e}\mde\, .
\end{align*}

\atob i {ii}
Let us fix an $\e>0$. Since
$\mathring A_\r = \bigcup_{\r'\geq \r}\mathring A_{\r'} \nearrow \dot A_\rs$ for $\r\searrow \rs$,
the $\s$-continuity of finite measures yields a $\r_\e>\rs$ such that $\mu(\hat A_\rs\setminus \mathring A_\r)\leq \e$
for all $\r\in (\rs,\r_\e]$. Using $A_\r\subset \hat A_\rs$ for $\r>\rs$, we then obtain the assertion
$\mu(A_\r\setminus \mathring A_\r) \leq \mu(\hat A_\rs\setminus \mathring A_\r) \leq \e$.

\atob {ii} i Let us fix an $\e>0$. For $\r\in (\rs,\r_\e]$, we then have $\mathring A_\r\subset \dot A_\rs$, and hence
our assumption yields $\mu(A_\r\setminus \dot A_\rs)\leq \e$.
In other words, we have $\lim_{\r\searrow \rs} \mu(A_\r\setminus \dot A_\rs) = 0$.
Moreover, we have $A_\r\nearrow \hat A_\rs$ for $\r\searrow \rs$, and hence
the $\s$-continuity of $\mu$ yields
$\lim_{\r\searrow \rs} \mu(A_\r\setminus \dot A_\rs) = \mu(\hat A_\rs \setminus \dot A_\rs)$.
\end{proof}

\begin{lemma}\label{gen-inv}
 Let $f:(0,1]\to (0,\infty)$ be a monotonously increasing function and $g:(0,f(1)]\to [0,1]$
 be its generalized inverse, that is
 \begin{displaymath}
  g(y):= \inf\bigl\{ x\in (0,1]: f(x)\geq y  \bigr\}\, , \qquad \qquad y\in (0,1].
 \end{displaymath}
Then we have $lim_{y\to 0^+}g(y) = 0$.
\end{lemma}

\begin{proof}[Proof of Lemma \ref{gen-inv}]
 Let $(y_n)\subset (0,f(1)]$ be a sequence with $y_n\to 0$.
 For $n\geq 1$, we write $E_n := \{ x\in (0,1]: f(x)\geq y_n  \}$.
 Let us fix an $\e\in (0,1]$. Since $f$ is strictly positive, we then find $f(\e)>0$, and hence there
 exists an $n_0\geq 1$ such that $f(\e)\geq y_n$ for all $n\geq n_0$.  Thus, we have $\e\in E_n$
 for all $n\geq n_0$, and from the latter we obtain
 $g(y_n) = \inf E_n \leq \e$ for such $n$.
\end{proof}

Before we prove Theorem 4.1, let us briefly illustrate the additional assumption
$\mu(\overline{A_i^* \cup A_2^*}\setminus (A_1^* \cup A_2^*)) = 0$.
To this end, we fix a $\mu$-density $h$ of $P$. Then Lemma \ref{Mr-include-both-new} tells us that
\begin{displaymath}
   A_i^* \cup A_2^*
   = \bigcup_{\r>\rs} M_\r
   \subset  \bigcup_{\r>\rs} \overline{\{h\geq \r  \}}
   \subset \overline{  \bigcup_{\r>\rs} {\{h\geq \r  \}} }
   = \overline{\{ h>\rs\}}\, .
\end{displaymath}
Using the normality in Assumption C, which implies \eqref{reg3},
we then obtain
\begin{align*}
   \mu\bigl(\overline{A_i^* \cup A_2^*}\setminus (A_1^* \cup A_2^*)\bigr)
   \leq \mu\bigl( \overline{\{ h>\rs\}}\setminus \{ h>\rs\} \bigr)
   &\leq \mu(\partial\{h>\rs\}) \\
   &=  \mu(\partial\{h\leq \rs\})\, .
\end{align*}
Consequently, the additional assumption is satisfied, if there exists a $\mu$-density $h$ of $P$
such that  $ \mu(\partial\{h\leq \rs\})=0$. In this respect  recall, that Lemma \ref{regular-new} showed that
$P$ is normal, if, for all $\r\in \R$, we have a $\mu$-density $h$ of $P$ with
$ \mu(\partial\{h\geq \r\})=0$.

\begin{proof}[Proof of Theorem 4.1]
We fix an $\eps>0$.
For   $n\geq 1$,
$\t:= \t_n$, and $\e:= \e_n$, we define $\e^*_n$ by the right hand-side of (3.4).
Then, Lemma \ref{gen-inv} shows $0<\e^*_n \leq \eps \wedge (\rss-\rs)/9$
for   sufficiently large $n$.
In addition,
$\d_n$ and $\e_n$ satisfy  (3.2) for
sufficiently large $n$ by (4.1), and
we also have 
$\d_n\leq  \dthick$ for
sufficiently large $n$.
Thus,  there is an $n_0\geq 1$ such that,
for all $n\geq n_0$, the values $\e_n$, $\d_n$, $\t_n$ and $\e^*_n$
satisfy the assumptions of Theorem 3.1 and
$\e^*_n \leq \eps$.
%

Let us now consider an $n\geq n_0$ and a   data set $D\in X^n$
 satisfying both
the assertions \emph{i)} - \emph{v)} of Theorem \ref{analysis-main-new-suppl} and (2.10).
By  Theorem 3.1 and our previous considerations we then know that
the probability $P^n$ of $D$ is not less than $1-e^{-\vs}$.
Now, part \emph{i)}   of Theorem \ref{analysis-main-new-suppl} yields $\rds - \rs \geq 2\e_n > 0$ and
\begin{displaymath}
 \rds-\rs\leq \e^*_n+5\e_n \leq 6\e^*_n \leq 6\eps\, ,
\end{displaymath}
i.e.~we have shown the first convergence.

To
prove the second convergence,
we  write $A_i$, $i=1,2$, for the two topologically connected
components of $M_{\rss}$. For $\r\in(\rs,\rss]$, we further define
$A^i_\r:= \z_\r(A_i)$, where $\z_\r:\ca C (M_{\rss})\to \ca C(M_\r)$
is  the CRM. In addition, we write $A^i_\r:= \emptyset$ for $\r>\rss$
and $A^i_\r:= X$ for $\r\leq \rs$. Let us first show
\begin{equation}\label{Approx-main-h3}
 \mu(\hat A^i_\rs\setminus \dot A^i_\rs) = 0
\end{equation}
for $i=1,2$, where we used the notation of Lemma \ref{approx-lemma}.
To this end, we fix an $\eps>0$.
Since $P$ is lower and upper normal at every level $\r\in [\rs,\rss]$
we find, for an arbitrary $\mu$-density $h$ of $P$,
\begin{displaymath}
   \mu(\hat M_\rs\setminus \dot M_\rs) = \mu\bigl(\{h>\rs\}  \setminus \dot M_\rs\bigr) = 0\, ,
\end{displaymath}
where we used \eqref{reg3}, \eqref{reg4}, and the notation of Lemma \ref{approx-lemma}.
 Lemma \ref{approx-lemma} then shows that there exists a $\r_\eps>\rs$ such that
\begin{equation}\label{Approx-main-h3xx}
   \mu(M_\r\setminus \mathring M_\r)\leq \eps
\end{equation}
 for all $\r\in (\rs,\r_\eps]$,
 where we may assume without loss of generality that $\r_\eps\leq \rss$.
Let us now fix a $\r\in (\rs, \r_\eps]$. Then we obviously have
$\mathring A^1_\r\cup \mathring A^2_\r \subset \mathring M_\r$.
To prove that the converse inclusion also holds,
we pick an $x\in \mathring M_\r$. Without loss of generality we may assume that $x\in A_\r^1$.
Since $A_\r^2$ is closed and thus compact, we then have $\e:= d(x,A_\r^2)>0$.
Moreover, since $\mathring M_\r$ is open, there
exists a $\d\in (0,\e)$ such that $B(x,\d)\subset \mathring M_\r$. This yields
$B(x,\d) \subset A_\r^1 \cup A_\r^2$, and by   $d(x,A_\r^2)> \d$, we conclude that $B(x,\d) \subset A_\r^1$.
This shows $x\in \mathring A_\r^1$, and hence we indeed have $\mathring M_\r = \mathring A^1_\r\cup \mathring A^2_\r$.
Now we use this equality to obtain
\begin{displaymath}
 M_\r\setminus \mathring M_\r
 =  \bigl(A^1_\r\setminus (\mathring A^1_\r\cup \mathring A^2_\r)\bigr) \cup
  \bigl(A^2_\r\setminus (\mathring A^1_\r\cup \mathring A^2_\r)\bigr)
  = (A^1_\r\setminus \mathring A^1_\r) \cup (A^2_\r\setminus \mathring A^2_\r)\, .
\end{displaymath}
By \eqref{Approx-main-h3xx},
this implies $\mu(A^i_\r\setminus \mathring A^i_\r)\leq \eps$, and thus Lemma \ref{approx-lemma}
shows \eqref{Approx-main-h3}.

Let us now fix an $\eps>0$ and a $\vs \geq 1$.
By the  equality of Lemma \ref{approx-lemma}
and the $\s$-continuity of finite measures
there then exist
$\d_\eps>0$, $\e_\eps>0$, and  $\r_\eps\in (\rs,\rss]$ such that, for all
$\e\in (0,\e_\eps]$, $\d\in (0,\d_\eps]$,   $\r\in (\rs,\r_\eps]$, and $i=1,2$, we have
$\mu ( \dot A^i_\rs \setminus (A_{\r+\e}^i)\mde   ) \leq \eps$.
Combining this with  $A^*_i = \hat A^i_\rs$, which holds  by the definition of the clusters
$A^*_i$, and
Equation \eqref{Approx-main-h3} we then obtain
\begin{equation}\label{Approx-main-h4}
 \mu \bigl(   A^*_i  \setminus (A_{\r+\e}^i)\mde    \bigr)
 =
 \mu \bigl(    \hat A^i_\rs  \setminus (A_{\r+\e}^i)\mde   \bigr)
=
\mu \bigl( \dot A^i_\rs \setminus (A_{\r+\e}^i)\mde   \bigr) \leq \eps\, .
\end{equation}
Moreover, our assumption $\mu(\overline{A_i^* \cup A_2^*}\setminus (A_1^* \cup A_2^*)) = 0$
means $\mu(\overline{\hat M_\rs}\setminus \hat M_\rs) = 0$, and since by part \emph{iii)}
of Lemma \ref{Td-lemma-new} we know that
\begin{displaymath}
 \bigcap_{\d>0} \biggl( \bigcup_{\r>\rs} M_\r\biggr)\pde    = \overline{\bigcup_{\r>\rs} M_\r} = \overline{\hat M_\rs}
\end{displaymath}
we find
\begin{displaymath}
   \mu\Bigl( \Bigl( \bigcup_{\r>\rs} M_\r\Bigr)\pde  \setminus \hat M_\rs   \Bigr)\leq \eps
\end{displaymath}
for all sufficiently small $\d>0$. From this it is easy to conclude that
\begin{equation}\label{Approx-main-h1}
 \mu( M_{\r-\e}\pde \setminus \hat M_\rs  ) \leq \eps
\end{equation}
 for  all sufficiently small $\e>0$, $\d>0$ and all $\r>\rs+\e$. Without loss of generality, we may thus assume that
\eqref{Approx-main-h1} also holds for all $\e\in (0,\e_\eps]$, $\d\in (0,\d_\eps]$ and all $\r>\rs+\e$.

For given $\t:= \t_n$ and $\e:= \e_n$ we now define $\e^*_n$ by the right hand-side of (3.4). Then, Lemma \ref{gen-inv} shows $\e^*_n\to 0$,
and hence we obtain $\e^*_n \leq \min\{\frac{\r_\eps-\rs}{9}, \eps,  \e_\eps\}$
for all sufficiently large $n$. In addition,
$\d_n$ and $\e_n$ satisfy  (3.2) for
sufficiently large $n$ by (4.1), and
we also have $\e_n \leq \eps\wedge \e_\eps$ and $\d_n\leq \d_\eps \wedge \dthick$ for
sufficiently large $n$.
Consequently,  there exists an $n_0\geq 1$ such that,
for all $n\geq n_0$, the values $\e_n$, $\d_n$, $\t_n$ and $\e^*_n$
satisfy the assumptions of Theorem 3.1 as well as $\e_n\leq \eps\wedge\e_\eps$ and
 $\d_n\leq \d_\eps$.

Let us now consider an $n\geq n_0$ and a   data set $D\in X^n$
 satisfying both
the assertions \emph{i)} - \emph{v)} of Theorem \ref{analysis-main-new-suppl} and
(2.10).
By  Theorem 3.1 and our previous considerations we then know that
the probability $P^n$ of $D$ is not less than $1-e^{-\vs}$.
Now, part \emph{i)}   of Theorem \ref{analysis-main-new-suppl}
gives both
$\rds \geq \rs +2\e_n>\rs+\e_n$ and $\rds\leq \rs+\e^*_n+5\e_n \leq \rs+6\e^*_n \leq \r_\eps$,
and
 hence
\eqref{Approx-main-h4} and \eqref{Approx-main-h1} hold for $\e:= \e_n$, $\d:= \d_n$,
and $\r:= \rds$.
Consequently, (2.10) shows
\begin{align*}
   \mu(B_1(D) \!\symdif\! A_1^*)  + \mu(B_2(D) \!\symdif\! A_2^*)
&\leq
  2\mu \bigl(A_1^* \setminus  (A^1_{\r+\e})\mde\bigr)
 +2\mu \bigl(A_2^* \setminus  (A^2_{\r+\e})\mde\bigr) \\ \label{cluster-chunk-generic-bound}
  &\qquad +
  \mu\bigl ( M_{\r-\e}\pde \setminus \{ h>\rs \}\bigr) \\
  &\leq 4\eps +  \mu\bigl ( M_{\r-\e}\pde \setminus \hat M_\rs\bigr) \\
  &\leq 5\eps\, ,
\end{align*}
where in the second to last step we also used \eqref{reg4}.
\end{proof}




\section{Additional Material  Related to Rates}\label{sec:proof-rates}




In this section, the assumption made in Section 4
are discussed in some more detail.

Let us begin with the following lemma, which gives a sufficient condition for a non-trivial
separation exponent.

\begin{lemma}\label{hoelder-sep}
 Let $X\subset \Rd$ be  compact and convex, $\snorm\cdot$ be some norm on $\Rd$,
and $P$ be a Lebesgue absolutely continuous distribution on $X$
 that can be clustered
between the   levels $\rs$ and $\rss$. Assume that  $P$ has
a continuous density $h$ and that there exist constants $c>0$ and $\theta\in (0,\infty)$ such that
\begin{equation}\label{hoelder-sep-ineq}
   \bigl| h(x) - h(x')| \leq c \, \snorm{x-x'}^\theta
\end{equation}
for all $x\in \{h\leq \rs\}$, $\r\in (\rs, \rss]$,  and $x'\in \partial_X M_\r$, where
$\partial_X M_\r$ denotes the boundary of $M_\r$ in $X$.
Then
%
the clusters of $P$ have  separation exponent $\theta$.
\end{lemma}

\begin{proof}[Proof of Lemma \ref{hoelder-sep}]
 Let  $\e\in (0,\rss-\rs]$ and $A_1$ and $A_2$ be the
 connected components of $M_{\rs+\e}$.
Since $A_1$ and $A_2$ are closed, they are
compact, and hence there exist $x_1\in A_1$ and $x_2\in A_2$ with
\begin{equation}\label{hoelder-sep-h1}
 a:=\snorm{x_1 -x_2} = d(A_1,A_2)\, ,
\end{equation}
where we note that $A_1\cap A_2 = \emptyset$ implies $a>0$.
For $t\in [0,1]$, we now consider
\begin{displaymath}
 x(t) := tx_1 + (1-t)x_2\, .
\end{displaymath}
Since $X$ is convex, we note that $x(t) \in X$ for all $t\in [0,1]$. Our first goal is to show that
$x_i\in \partial_X M_{\rs+\e}$ for $i=1,2$. To this end, we assume the converse,
e.g.~$x_2\in \mathring M_{\rs+\e}$. 
Then there   exists an
$\eps\in (0,a)$ with $B_X(x_2,\eps)\subset \mathring A_2$, where $B_X(x_2,\eps):= \{x\in X: \snorm{x-x_2}\leq \eps\}$.
Now $\snorm{x(\eps/a)-x_2} = \eps$ implies
$x(\eps/a)\in A_2$, while $\snorm{x(\eps/a)-x_1} = a-\eps$ shows $\snorm{x(\eps/a)-x_1} < d(A_1,A_2)$.
Together this contradicts \eqref{hoelder-sep-h1}.

For what follows, let us now observe that
$t\mapsto x(t)$ is a continuous map on $[0,1]$, and
since $h$ is  continuous,
there  exists a $t^*\in [0,1]$
with $h(x(t^*)) = \min_{t\in [0,1]}h(x(t))$.
Our next goal is to show that
\begin{equation}\label{hoelder-sep-h2}
 h(x(t^*)) \leq \rs\, .
\end{equation}
To this end, we assume the converse, that is
$h(x(t^*)) > \rs$. Then there exists a $\d\in (0,\e]$
such that $h(x(t))> \rs+\d$ for all $t\in [0,1]$, and therefore
an  application of Lemma \ref{Mr-include-both-new} using the
continuity of $h$ yields $x(t) \in M_{\rs+\d}$ for all $t\in [0,1]$.
In other words, $x_1$ and $x_2$ are path-connected in $M_{\rs+\d}$,
and since the connecting path is a straight line,
it is easy to see that $x_1$ and $x_2$ are $\t$-connected for all $\t>0$.
Let us pick a $\t\leq 3\t^*(\d) = \t^*_{M_{\rs+\d}}$.
Since $|\ca C(M_{\rs+\d})|=2$, part \emph{ii)} of
Proposition \ref{top-connect-new2-prop} then shows
$\ca C(M_{\rs+\d}) = \ca C_\t(M_{\rs+\d})$. Let $\tilde A_1$ and $\tilde A_2$ be
the two topologically connected components of $M_{\rs+\d}$.
Our previous considerations then showed that $\tilde A_1$ and $\tilde A_2$
are also the two $\t$-connected components of $M_{\rs+\d}$.
Now, $\d\leq \e$ gives a CRM $\z:\ca C(M_{\rs+\e}) \to \ca C(M_{\rs+\d})$,
which is bijective, since $P$ can be clustered between $\rs$ and $\rss$.
Without loss of generality we may thus assume that $\z(A_i) = \tilde A_i$ for $i=1,2$.
This yields $x_i\in A_i \subset \tilde A_i$, i.e.~$x_1$ and $x_2$ do not
belong to the same $\t$-connected component of $M_{\rs+\d}$.
Clearly, this contradicts our observation that
$x_1$ and $x_2$ are $\t$-connected, and hence \eqref{hoelder-sep-h2} is proven.

Now assume without loss of generality that $t^*\in [1/2,1)$. Since we have already seen that
$x_1\in \partial_X M_{\rs+\e}$, our assumption
\eqref{hoelder-sep-ineq} and \eqref{hoelder-sep-h2}
yield
%
%
\begin{displaymath}
 \bigl|h(x(t^*))-h(x_1)\bigr| \leq c \,\snorm{x(t^*) - x_1}^\theta\, .
\end{displaymath}
 In addition, Lemma \ref{Mr-include-both-new}   shows
$x_1\in M_{\rs+\e} \subset \{h\geq \rs + \e \}$.
Combining these estimates with \eqref{hoelder-sep-h1} and
$d(A_1,A_2) = \t^*_{M_{\rs+\e}} = 3\t^*(\e)$,
we find
\begin{align*}
  \rs+\e
\leq h(x_1)
\leq h(x(t^*)) +  c \,\snorm{x(t^*) - x_1}^\theta
&\leq \rs + c \,\snorm{x(t^*) - x_1}^\theta\\
&\leq  \rs + c\, 2^{-\theta} d^\theta(A_1,A_2)\\
& =   \rs + c\, (3/2)^{-\theta} \t^*(\e)^\theta\, ,
\end{align*}
and from the latter the assertion easily follows.
\end{proof}

Note that \eqref{hoelder-sep-ineq} holds,
if the density $h$ in
Lemma \ref{hoelder-sep} is actually  $\theta$-H\"older-continuous, and it is easy to see that the
converse is, in general, not true.
Moreover, using the inclusion $\partial_X M_\r \subset \{h=\r  \}$ established in Lemma \ref{Mr-include-both-new},
it is easy
to check that \eqref{hoelder-sep-ineq} is \emph{equivalent} to
\begin{equation}\label{hoelder-sep-ineq-new}
   \bigl| h(x) - \r| \leq c \, d(x, \partial_X M_\r )^\theta
\end{equation}
for all $x\in \{h\leq \rs\}$ and  $\r\in (\rs, \rss]$. Note that
a localized but two-sided version of this
 condition has been used in \citeappendix{SiScNo09a} for a level set estimator
that is adaptive with respect to  the Hausdorff metric.

Our next goal is to discuss the assumptions made in Theorem 4.7 in more detail.
To this end, we need a couple of technical lemmata.

\begin{lemma}\label{dist-to-boundary}
   Let $X\subset \Rd$ be compact and convex and $d$ be a metric on $X$ that is defined by a norm on $\Rd$.
   Then, we have
   \begin{displaymath}
      d(x, \partial_X A) \leq d (x, X\setminus A)
   \end{displaymath}
        for all $A\subset X$ and $x\in \overline A$,
   where $\partial_X A$ denotes  the boundary of $A$ in $X$.
\end{lemma}

\begin{proof}[Proof of Lemma \ref{dist-to-boundary}]
    Before we begin with the proof we note that
    $\overline B^X = \overline B^{\Rd}$ for all $B\subset X$ since $X$ is closed,
    i.e., taking the closure with respect to $X$ or $\Rd$ is the same. Like in the statement
    of the lemma, we will thus omit the superscript.
   Let us now write $\d:= d (x, X\setminus A)$. Then there exists a sequence $(x_n)\subset X\setminus A$ such that
   $d(x,x_n)\to \d$. Since $X$ is assumed to be compact, so is $\overline {X\setminus A}$, and  thus there
   exists an $x_\infty \in \overline {X\setminus A}$ such that $d(x, x_\infty)\leq \d$. Obviously, it suffices to
   show $x_\infty\in \partial_X A$. Let us assume the converse.
   Since $\partial_X A = \overline A \cap  \overline {X\setminus A}$,
   we then have $x_\infty \not \in \overline A$, that is $x_\infty \in X\setminus \overline A$.
   Now, the latter set is open in $X$,
   and hence there exists an $\e>0$ such that $B_X(x_\infty, \e)\subset X\setminus \overline A$, where $B_X(x_\infty, \e)$ denotes
   the closed ball in $X$ that has  center $x_\infty$ and radius $\e$.
   This $\e$ must satisfy $\e<\d$, since otherwise we would find a contradiction to $x\in \overline A$ by
   $x\in B_X(x_\infty, \d) \subset B_X(x_\infty, \e) \subset X\setminus \overline A$.
   For $t:= \e/\d \in (0,1)$ we now define $x' := tx+(1-t)x_\infty$.
   The convexity of $X$ implies $x'\in X$, and
   since $d$ is defined by a norm, we   have
   $d(x_\infty, x') = t d(x,x_\infty) \leq \e$. Together, this yields
   $x'\in B_X(x_\infty, \e)\subset X\setminus \overline A \subset X\setminus A$. Consequently,
   $d(x,x') = (1-t) d(x,x_\infty) \leq (1-t) \d < \d$ implies
   $d(x, X\setminus A) < \d$, which contradicts the definition of $\d$.
\end{proof}

\begin{lemma}\label{tube-around-the-boundary}
   Let $X\subset \Rd$ be compact and convex and $d$ be a metric on $X$ that is defined by a norm on $\Rd$.
   Then, for all $A\subset X$ and $\d>0$, we have
   \begin{displaymath}
      A\pde \setminus A\mde \subset (\partial_X A)\pde\, ,
   \end{displaymath}
   where the operations $A\pde$ and $A\mde$ as well as the boundary $\partial_X A$ are with respect to the metric space $(X,d)$.
\end{lemma}

\begin{proof}[Proof of Lemma \ref{tube-around-the-boundary}]
   Let us fix an $x\in  A\pde \setminus A\mde = A\pde \cap (X\setminus A)\pde$.
   If $x\in \overline A$, then Lemma \ref{dist-to-boundary}
   immediately yields $d(x, \partial_X A) \leq d (x, X\setminus A) \leq \d$, that is $x\in (\partial_X A)\pde$.
   It thus suffices to consider the case $x\not \in \overline A$. Then we find
   $x\in X\setminus \overline A \subset X\setminus A \subset \overline {X\setminus A}$, and hence
   another application of Lemma \ref{dist-to-boundary}
   yields $d(x, \partial_X (X\setminus A)) \leq d(x,A)\leq \d$. Now the assertion easily follows from
   $\partial_X (X\setminus A) = \overline{X\setminus A} \cap \overline{X\setminus (X\setminus A)} =
   \overline{X\setminus A} \cap \overline A = \partial_X A$.
\end{proof}

The next lemma shows that assuming an
$\a$-smooth boundary with
$\a>1$ does not make sense. It further shows
that, for each  level set with rectifiable boundary in the sense of \citeappendix[3.2.14]{Federer69},
the bound
(4.9) holds with $\a=1$.

\begin{lemma}\label{mass-of-bound}
  Let $\lb^d$ be the $d$-dimensional Lebesgue measure,
   $\ca H^{d-1}$ be the $(d-1)$-dimensional Hausdorff measure
  on $\Rd$,
  and $\s_{d}$ be the volume of the $d$-dimensional unit Euclidean  ball in $\R^{d}$.
  Then,
   for every non-empty, bounded, and  measurable subset $A\subset \Rd$ the following statements hold:
   \begin{enumerate}
    \item There exists  a $\d_A>0$, such that for $\underline c_A:=  d \s_d^{1/d} \lb^d(\overline A)^{1-1/d}/2$
      and
     all $\d\in (0,\d_A]$, we have
    \begin{displaymath}
       \lb^d(A\pde \setminus A\mde) \geq \underline c_A \cdot \d     \, .
    \end{displaymath}
    \item If $\partial A$ is $(d-1)$-rectifiable
    and $\ca H^{d-1}(\partial A)>0$,
    there exists  a $\d_A>0$, such that,
    for all $\d\in (0,\d_A]$, we have
    \begin{displaymath}
          {\lb^d\bigl( A\pde \setminus A\mde\bigr) }
    \leq 4  \ca H^{d-1}(\partial A) \cdot \d\, .
    \end{displaymath}
    \end{enumerate}
\end{lemma}

\begin{proof}[Proof of Lemma \ref{mass-of-bound}]
  Let
    us first recall that, for an integer $0\leq m\leq d$,
   the upper and lower Minkowski content of a $B\subset \Rd$ is defined by
   \begin{align*}
       \ca M^{*m} (B) &:= \limsup_{\d\to 0^+} \frac {\lb^d(B\pde)}{\s_{d-m} \d^{d-m}} \\
      \ca M^m_* (B) &:= \liminf_{\d\to 0^+} \frac {\lb^d(B\pde)}{\s_{d-m} \d^{d-m}} \, ,
   \end{align*}
  where $\s_{d-m}$ denotes the $\lb^{d-m}$-volume of the unit Euclidean  ball in $\R^{d-m}$.
  It is easy to check that these definitions coincide with those in \citeappendix[3.2.37]{Federer69}.

  \ada i Since in the case $\lb^d(\overline A) = 0$ there is nothing to prove, we restrict our considerations to the case
  $\lb^d(\overline A) > 0$. Now, $A$ is   bounded, and hence we have
    $\lb^d(\overline A) < \infty$.
  The isoperimetric inequality  \citeappendix[3.2.43]{Federer69} thus yields
  \begin{displaymath}
      d \s_d^{1/d} \lb^d(\overline A)^{1-1/d} \leq \ca M^{d-1}_* (\partial A)  \, ,
  \end{displaymath}
  and hence, there exists a $\d_A>0$, such that, for all $\d\in (0,\d_A]$, we have
   \begin{displaymath}
     \frac {d \s_d^{1/d} \lb^d(\overline A)^{1-1/d}}2 \leq \frac {\lb^d\bigl((\partial A)\pde\bigr) }{\s_1\d} \leq
     \frac {\lb^d\bigl( A^{+2\d} \setminus A^{-2\d}\bigr) }{2\d} \, ,
  \end{displaymath}
  where in the last estimate we used part {\em viii)\/} of Lemma \ref{Td-lemma-new} and $\s_1 = 2$.

 \ada {ii} Since $\partial A$ is closed and  $(d-1)$-rectifiable in the sense of \citeappendix[3.2.14]{Federer69}, we find
 \begin{displaymath}
    \ca M^{*(d-1)} (\partial A) = \ca H^{d-1}(\partial A)
 \end{displaymath}
  by \citeappendix[3.2.39]{Federer69}. Moreover, since $\partial A$ is bounded, the boundary is contained in a compact set $X\subset \Rd$
  such that the relative boundary $\partial_XA$ of $A$ in $X$ equals $\partial A$ and the sets $A\pde $ and $A\mde$ considered in $X$
  equal the sets $A\pde $ and $A\mde$ when considered in $\Rd$ for all $\d\in (0,1]$.
  By Lemma \ref{tube-around-the-boundary} there thus exists a $\d_A>0$ such that
   \begin{displaymath}
    \frac {\lb^d\bigl( A\pde \setminus A\mde\bigr) }{2\d} \leq \frac{\lb^d\bigl((\partial A)\pde\bigr) }{\s_1\d}
    \leq 2  \ca H^{d-1}(\partial A)
  \end{displaymath}
  for all $\d\in (0,\d_A]$.
\end{proof}

The next  lemma shows that a  bound  (4.9) together
with a regular behavior of $h$ around the   level of interest  ensures a non-trivial
flatness exponent.

\begin{lemma}\label{reg-implies-flat}
 Let $(X,d)$ be a complete, separable metric space, $\mu$ be a finite Borel measure on $X$
with $\supp \mu = X$,  and $P$ be a $\mu$-absolutely continuous distribution on $X$.
Furthermore, let  $\r\geq 0$ be a level and $h$ be a $\mu$-density of $P$
for which there
exist constants $c>0$, $\a\in (0,1]$, $\d_0>0$, and $\theta \in (0,\infty)$ such that
\begin{equation}\label{boundary-at-crit-level}
 \mu(M_\r\pde \setminus M_\r\mde) \leq c \d^\a
\end{equation}
for all $\d\in (0,\d_0]$ and
\begin{equation}\label{hoelder-reg-ineq}
  d(x, \partial M_\r )^\theta \leq   c\,  \bigl| h(x) - \r|
\end{equation}
for all $x\in \{h> \r\}$. Then $P$ has flatness exponent $\a/\theta$ at level $\r$.
\end{lemma}


\begin{proof}[Proof of Lemma \ref{reg-implies-flat}]
 Let us fix an $s>0$. For $x\in \{0<h-\r<s\}$ we then find
$d(x, \partial M_\r)^\theta \leq c s$ by \eqref{hoelder-reg-ineq}, that is
$x\in (\partial M_\r)\pde$ for $\d:= (c s)^{1/\theta}$.
Using part \emph{viii)} of Lemma \ref{Td-lemma-new}, we conclude that
$x\in M_\r^{+2\d}\setminus M_\r^{-2\d}$. In the case $2\d\leq \d_0$, we thus obtain
\begin{displaymath}
 \mu\bigl(\{ 0< h-\r <s    \}  \bigr) \leq \mu\bigl( M_\r^{+2\d}\setminus M_\r^{-2\d}  \bigr)
\leq 2^\a c\, \d^\a = 2^\a c^{1+\a/\theta}s^{\a/\theta}\, ,
\end{displaymath}
and since $\mu$ is a finite measure, it is then easy to see that we can increase the constant
on the right-hand side  so that it holds for all $s>0$.
\end{proof}

\setcounter{section}{0}
\renewcommand{\thesection}{B.\arabic{section}}

\vspace*{3ex}
\noindent
\textbf{Appendix B. Continuous Densities in two Dimensions.}
In this appendix, we present a couple of two-dimensional examples that show
that the assumptions imposed in the paper are not only met by many discontinuous densities, but also
by many continuous densities.
\section{Single Two-Dimensional Sets}\label{sec:single-twod}
In this section we consider the operations $\oplus \d$ and $\ominus \d$ for a specific class of sets $A\subset \R^2$.

We begin with an example
of a set $A\subset \R^2$, for which we can compute $A\opde$ and  $A\omde$ explicitly. This example will
be the base of all further examples.

\begin{example}\label{Aomde}
   Let $X:= [-1,1]\times [-2,2]$ be equipped with the metric defined by the supremums norm.
   Furthermore, for $x^\pm_- \in (-0.6, -0.4)$ and $x^\pm_+\in (0.4,0.6)$ we fix two continuous functions
   $f^-, f^+:[-1,1] \to [-1, 1]$ such that $f^+$ is increasing on $[-1,x^+_-]\cup [0,x^+_+]$ and decreasing
   on $[x^+_-,0]\cup [x^+_+,1]$, while $f^-$ is decreasing on $[-1,x^-_-]\cup [0,x^-_+]$ and increasing
   on $[x^-_-,0]\cup [x^-_+,1]$.
        In addition, assume that $\{f^- < 0\} = \{f^+>0\}$ and $\{f^- = 0\} = \{f^+ = 0\}$
        as well as $f^-(\pm 0.5)<0$ and $f^+(\pm 0.5)>0$.
%
        Now consider the (non-empty) set $A$ enveloped by $f^\pm$, that is
   \begin{displaymath}
      A:= \bigl\{  (x,y)\in X: x\in [-1, 1] \mbox{ and } f^-(x)\leq y\leq f^+(x)\bigr\}\, .
   \end{displaymath}
        To describe $A\omde$ for $\d\in (0,0.1]$, we define
        $f_{-\d}^\pm:[-1,1]\to [-1, 1]$ by
  \begin{displaymath}
     f_{-\d}^\pm(x) :=
     \begin{cases}
        f^\pm(-1)   & \mbox{ if }\, x\in [-1,-1+\d]\\
         f^\pm(0)  & \mbox{ if }\, x\in [-\d,+\d]\\
          f^\pm(1)   & \mbox{ if }\, x\in [1-\d,1]\\
     \end{cases}
  \end{displaymath}
 and $f_{-\d}^-(x) := f^-(x-\d) \vee f^-(x+\d)$, respectively
  $f_{-\d}^+(x) := f^+(x-\d) \wedge f^+(x+\d)$ for the remaining $x\in [-1,1]$.
  Then we have
  \begin{displaymath}
     A\omde = \bigl\{  (x,y)\in X: x\in [-1, 1] \mbox{ and } f_{-\d}^-(x)+\d\leq y\leq f_{-\d}^+(x)-\d\bigr\}\, .
  \end{displaymath}
  Moreover, to describe $A\opde$, we define
    \begin{align*}
     x_{0,-1} & := \min\bigl\{ x\in [-1,-0.5] : f^+(x) - f^-(x)  \geq 0\bigr\} \\
      x_{0, -0} & := \max\bigl\{ x\in [-0.5, 0]:  f^+(x) - f^-(x) \geq 0\bigr\} \\
       x_{0, +0} & := \min\bigl\{ x\in [0, 0.5 ] : f^+(x) - f^-(x) \geq 0 \bigr\} \\
        x_{0, +1} & := \max\bigl\{ x\in [0.5 ,1] :f^+(x) - f^-(x) \geq0 \bigr\} \, ,
  \end{align*}
  where the minima are   attained by the continuity of $f^\pm$ and the fact that all sets are non-empty.
  Furthermore, we define
        $f_{+\d}^\pm:[-1,1]\to [-1, 1]$ by
  \begin{displaymath}
     f_{+\d}^\pm(x) :=
     \begin{cases}
        f^\pm(x+\d)   & \mbox{ if }\, x\in [-1\vee (x_{0,-1} -\d), x_-^\pm -\d]\\
         f^\pm(x_-^\pm)  & \mbox{ if }\, x\in [x_-^\pm -\d, x_-^\pm +\d]\\
         f^\pm(x_+^\pm)  & \mbox{ if }\, x\in [x_+^\pm -\d, x_+^\pm +\d]\\
         f^\pm(x-\d)   & \mbox{ if }\, x\in [x_+^\pm +\d,  (x_{0,+1} +\d) \wedge 1]\,
     \end{cases}
  \end{displaymath}
  as well as $f_{+\d}^-(x)  := f^-(x-\d) \wedge f^-(x+\d)$ and $f_{+\d}^+(x)  := f^+(x-\d) \vee f^+(x+\d)$
%
%
  for $x\in [x_-^\pm +\d, x_+^\pm -\d]\setminus (x_{0,-0}+\d, x_{0,+0}-\d)$
  and $f_{+\d}^\pm(x):= -2\d$ for the remaining $x\in [-1,1]$.
  Then we have
  \begin{displaymath}
     A\opde = \bigl\{  (x,y)\in X: x\in [-1, 1] \mbox{ and } f_{+\d}^-(x)-\d\leq y\leq f_{+\d}^+(x)+\d\bigr\}\, .
  \end{displaymath}
  Finally, we have $|\ca C(A)| \leq 2$ with $|\ca C(A)| = 2$ if and only if $ x_{0, -0} <x_{0, +0}$, and in the latter
  case we further have $\t^*_A = x_{0, +0} - x_{0, -0}$.
\end{example}

\begin{proof}[Proof of Example \ref{Aomde}]
Let us fix a $\d\in (0,1/10]$. To simplify notations, we further write $g^-:= f_{-\d}^-+\d$ and
$g^+:= f_{-\d}^+-\d$.

  Proof of ``$A\omde\subset\dots$''. By   $A\omde = X\setminus (X\setminus A)\opde$ it suffices  to show that
  \begin{displaymath}
     \bigl\{  (x,y)\in X: x\in [-1, 1] \mbox{ and   } \bigl(y<g^-(x) \mbox{ or } y> g^+(x)\bigr)\bigr\}
     \subset (X\setminus A)\opde\, .
  \end{displaymath}
 By symmetry, it further suffices to consider the case $x\geq 0$ and $y>g^+(x)$. Moreover, to show the
 inclusion above, it finally suffices to  find $x'\in [-1,1]$ and $y'\in [-2,2]$ with $|x-x'|\leq \d$, $|y-y'|\leq \d$ and
 $y'>f^+(x')$. However, this task is straightforward. Indeed,  we can always set $y':= (y+\d)\wedge 2$, and
 if $x\in [0,\d]$ then $x':= 0$
 works, since $y'=(y+\d)\wedge 2  >g^+(x)+\d=f^+(0) = f^+(x')$,
 while for $x\in [1-\d,1]$, the choice $x':=1$  does by an analogous argument. Finally, if $x\in (\d,1-\d)$, we
 set $x':= x-\d$ if $g^+(x) = f^+(x-\d)-\d$ and $x':= x+\d$ if $g^+(x) = f^+(x+\d)-\d$.

 Proof of ``$A\omde\supset\dots$''.  Again, it suffices to consider $x\geq 0$. Let us fix a $y$ with
 $g^-(x)\leq y\leq g^+(x)$. Then, our goal is to show   $(x,y)\not\in (X\setminus A)\opde$, i.e.,
 \begin{equation}\label{psis-in-r2-h1}
    \inorm{(x,y) - (x',y')} > \d
 \end{equation}
  for all $(x',y')\in X\setminus A$. In the following, we thus fix a pair $(x',y')\in X\setminus A$
  for which \eqref{psis-in-r2-h1} is not true and show that this leads to a contradiction.
  We begin by considering the case $x\in [0,\d]$. Since \eqref{psis-in-r2-h1} is not true, we find $|x-x'|\leq \d$, and
  hence $x^\pm_-\leq x'\leq x^\pm_+$. Then, if $y'>f^+(x')$, this leads to
  \begin{displaymath}
     y\leq g^+(x) = f^+(0) - \d \leq f^+(x') - \d < y'-\d\, ,
  \end{displaymath}
  which contradicts the assumed $|y-y'|\leq \d$. The case $y'<f^-(x')$ analogously leads to a contradiction.
  Now consider the case $x\in [1-\d,1]$. Then $|x-x'|\leq \d$ implies $x'\geq x^\pm_+$. Thus,
   $y'>f^+(x')$ leads to another contradiction by
  \begin{displaymath}
     y\leq g^+(x) = f^+(1) - \d \leq  f^+(x') - \d < y'-\d\, ,
  \end{displaymath}
  and the case $y'<f^-(x')$ can be treated analogously. It thus remains to consider the case $x\in [\d,1-\d]$.
  Then $|x-x'|\leq \d$ implies $x-\d\leq x'\leq x+\d$. For $x'\leq x^+_+$ we thus find $f^+(x-\d) \leq f^+(x')$, while for
  $x'\geq x^+_+$ we   find $f^+(x+\d) \leq f^+(x')$. For $y'>f^+(x')$ we hence obtain a contradiction
  by
  \begin{displaymath}
      y\leq g^+(x) = (f^+(x-\d) \wedge f^+(x+\d)) - \d \leq  f^+(x') - \d < y'-\d\, ,
  \end{displaymath}
  and, again, the case $y'<f^-(x')$ can be shown similarly.

  Proof of ``$A\opde\subset\dots$''. Let us fix a pair $(x,y) \in A\opde$. Without loss of generality we restrict our
  considerations to the case $y\geq 0$ and $x\in [-1,0]$. To show that $y\leq f_{+\d}^+(x)+\d$ we assume the converse, that is
  $y > f_{+\d}^+(x)+\d$. Since $(x,y) \in A\opde$ we then find $(x',y')\in A$ with $\inorm{(x,y)-(x',y')}\leq \d$.
  From the latter we infer that both $x-\d\leq x'\leq x+\d$ and
  \begin{equation}\label{low-yp}
     y' \geq y-\d > f_{+\d}^+(x) \, .
  \end{equation}
  If $x\in [-1,  -1\vee (x_{0,-1} -\d))$  we get a contradiction, since $(x',y')\in A$ implies
  $x\geq x'-\d\geq x_{0,-1}-\d$. Moreover, for  $x\in [-1\vee (x_{0,-1} -\d), x_-^+ -\d]$, we obtain
  \begin{displaymath}
   f_{+\d}^+(x) = f^+(x+\d) \geq  f^+(x') \geq y'\, ,
  \end{displaymath}
  which contradicts \eqref{low-yp}. If  $x\in [x_-^+ -\d, x_-^+ +\d]$ we get a
  contradiction from
     $f_{+\d}^+(x) =  f^+(x_-^+) \geq  f^+(x') \geq y'$, and if
   $x\in [x_-^+ +\d, 0 \wedge (x_{0,-0}+\d)]$ we have
    \begin{displaymath}
     f_{+\d}^+(x) = f^+(x-\d) \vee f^+(x+\d) \geq   f^+(x-\d) \geq  f^+(x') \geq y'\,
  \end{displaymath}
  which again contradicts \eqref{low-yp}. Finally, if $x\in ( 0 \wedge x_{0,-0}+\d, 0]$ we obtain a contradiction
  from $x>  x_{0,-0}+\d \geq x'+\d$.

   Proof of ``$A\opde\supset\dots$''. Let us fix a pair $(x,y)\in X$ with $ f_{+\d}^-(x)-\d\leq y\leq f_{+\d}^+(x)+\d$.
   Without loss of generality we again consider the case $y\geq 0$ and $x\in [-1,0]$, only.
   To show $(x,y)\in A\opde$ we need to find a pair $(x',y')\in A$ with $\inorm{(x,y)-(x',y')}\leq \d$.
   Let us assume that we have found an $x'$ with $|x-x'|\leq \d$ and $f(x')\geq y-\d$.
   For $y'$ defined by
   \begin{displaymath}
      y' := f(x') \wedge (y+\d)
   \end{displaymath}
   we then immediately obtain $y'\leq y+\d$. Moreover, if we actually have $y'= y+\d$, then we  obtain $|y-y'|\leq \d$,
   while in the case $y'< y+\d$ we find $y' = f(x') \geq y-\d$, that is again $|y-y'|\leq \d$.
   Thus, it suffices to find an $x'$ with the properties above.
   To this end, we first observe that we can exclude the case $x\in [-1,  -1\vee (x_{0,-1} -\d))$, since for such
   $x$ we have $0\leq y \leq f_{+\d}^+(x) + \d = -\d$. Analogously, we can exclude the case
   $x\in ( 0 \wedge (x_{0,-0}+\d), 0]$. Now consider the case
   $x\in [-1\vee (x_{0,-1} -\d), x_-^+ -\d]$. For $x':= x+\d$ we then have
   \begin{displaymath}
      f(x') = f(x+\d) = f_{+\d}^+(x) \geq y-\d\, ,
   \end{displaymath}
   and hence $x'$ satisfies the desired properties. Moreover, for
   $x\in [x_-^+ -\d, x_-^+ +\d]$ we define $x':= x_-^+$, which   gives $|x-x'|\leq \d$.
   In addition, we again have
      $f(x') = f(x_-^+) = f_{+\d}^+(x) \geq y-\d$.
  Finally, let us consider the case $x\in [x_-^+ +\d, 0 \wedge (x_{0,-0}+\d)]$.
  Let us first assume that  $f(x-\d)\geq f(x+\d)$.
  For $x':= x-\d$ we then obtain
      $f(x') = f(x-\d) = f_{+\d}^+(x) \geq y-\d$.
   Analogously, if $f(x-\d)\leq f(x+\d)$, then $x':= x+\d$ has the desired properties.

   Finally, $|\ca C(A)| \leq 2$ is obvious, and so is the equivalence between
   $|\ca C(A)| = 2$ and  $ x_{0, -0} <x_{0, +0}$. In the latter case, $A_1 := \{(x,y)\in A: x\leq x_{0, -0}\}$
   and $A_2 := \{(x,y)\in A: x\geq x_{0, +0}\}$ are the two components of $A$, and from this it is easy to conclude that
   $\t^*_A = x_{0, +0} - x_{0, -0}$.
\end{proof}

Our next example shows how to estimate the function  $\psis_A$ for
the sets considered in Example \ref{Aomde}

\begin{example}\label{psis-in-r2}
        Let us consider the situation of Example \ref{Aomde}. To simplify the presentation, let us
        additionally assume
        that the monotonicity of $f^+$ and $f^-$ is actually strict and that $A$ has sufficiently thick parts on both sides
        of the $y$-axis in the sense of
        \begin{equation}\label{psis-in-r2-thick}
         [-0.8, -0.2] \cup [0.2, 0.8] \subset \{f^-\leq-0.2\} \cap \{f^+\geq 0.2\}\, .
        \end{equation}
        Note that, for all $\d\in (0,0.1]$, this condition in particular ensures
        that $A\omde$ contains open neighborhoods around the points $ (-0.5, 0)$ and $(0,0.5)$.
        Moreover, for $\d\in [0,0.1]$ we define
  \begin{align*}
     x_{\d,-1} & := \min\bigl\{ x\in [-1,-0.8] : f^+(x) - f^-(x)  \geq 2\d\bigr\} \\
      x_{\d, -0} & := \max\bigl\{ x\in [-0.2, 0]:  f^+(x) - f^-(x) \geq 2\d\bigr\} \\
       x_{\d, +0} & := \min\bigl\{ x\in [0, 0.2 ] : f^+(x) - f^-(x) \geq 2\d \bigr\} \\
        x_{\d, +1} & := \max\bigl\{ x\in [0.8 ,1] :f^+(x) - f^-(x) \geq 2\d \bigr\} \, ,
  \end{align*}
  where we  note that the minima and maxima are   attained by \eqref{psis-in-r2-thick} and
        the continuity of $f^\pm$.  For the same reason we further have $x_{\d,-1}<-0.8$,
        $x_{\d, -0} > -0.2$, $x_{\d, +0} <0.2$, and   $x_{\d, +1} > 0.8$.
        Then,  $f_{-\d}^+$ has exactly two local maxima $x_{\d,-}^+$ and
                $x_{\d,+}^+$, satisfying $x_{\d,-}^+\in [-1, 0]$ and  $x_{\d,+}^+\in [0, 1]$,
        and   $f_{-\d}^-$ has exactly two local minima $x_{\d,-}^-$ and
                $x_{\d,+}^-$, satisfying $x_{\d,-}^-\in [-1, 0]$ and  $x_{\d,+}^-\in [0, 1]$.
        Moreover, for all $\d\in (0,0.1]$ we have
        \begin{align*}
         \psis_A(\d) &\leq \d +
                \Bigl(\max\bigl\{ |x_{\d,i} - x_{0,i}|: i\in \{-1, -0, +0, +1\}    \bigr\}\\
     &\qquad \qquad \vee
                \max\bigl\{ |f^i(x^i_j) - f_{-\d}^i(x^i_{\d,j})|: i,j\in \{ -, +\}    \bigr\} \Bigr)\, .
        \end{align*}
        The right hand-side of this inequality can be further estimated under some
        regularity assumptions. Indeed, if there exist $c>0$ and $\g\in (0,1]$ such that
        \begin{equation}\label{psis-in-r2-xpm}
         |f^\pm(x^\pm_\pm) - f^\pm(x)|  \leq c |x^\pm_\pm -x|^\g\, , \qquad \qquad x\in [x^\pm_\pm-0.1,x^\pm_\pm+0.1]\, ,
        \end{equation}
        then, for all $\d\in (0, 0.1]$, we can bound the second maximum by
        \begin{displaymath}
         \max\bigl\{ |f^i(x^i_j) - f_{-\d}^i(x^i_{\d,j})|: i,j\in \{ -, +\}    \bigr\}  \leq c \d^\g\, .
        \end{displaymath}
        In addition, if, for some $i\in \{-1, -0, +0, +1\}$, we write
        $2\d_0:=f^+(x_{0,i}) - f^-(x_{0,i})$, then $|x_{\d,i} - x_{0,i}| = 0$
        for all $\d\in (0,\d_0]$, i.e.~the corresponding term in the first maximum
        disappears for these $\d$. If $\d_0<0.1$, and
%
        we additionally assume, for example, that
        \begin{equation}\label{psis-in-r2-0pm}
         |f^\pm (x )| \geq   c^{-1/\g} |x_{0, -1} -x|^{1/\g}
        \end{equation}
        for all $x\in [x_{0, -1}, -0.8]$, then we have $|x_{\d,-1} - x_{0,-1}| \leq c \d^\g$ for all $\d\in (\d_0,0.1]$.
        Combining these assumptions
        we  obtain a variety of sets $A$ satisfying $\psis_A(\d) \leq (c+1)\d^\g$ for all $\d\in (0, 0.1]$,
        and these examples of sets can be even further extended by considering bi-Lipschitz transformations of $X$.
\end{example}

Before we can prove the assertions made in the example above, we need to establish the following technical lemma.

\begin{lemma}\label{translate-f}
 Let $x^*\in [2/5, 3/5]$ and $f:[0,1]:\to \R$ be a continuous function that is strictly
increasing on $[0,x^*]$ and strictly decreasing on $[x^*,1]$. For $\d\in (0,1/8]$ we define
$f_{-\d}:[0,1]\to \R$ by
\begin{displaymath}
 f_{-\d}(x)
:=
\begin{cases}
 f(0) & \mbox{ if } x\in [0,\d]\\
f(x-\d) \wedge f(x+\d)  & \mbox{ if } x\in [\d, 1-\d]\\
f(1) & \mbox{ if } x\in [1-\d, 1]\, .
\end{cases}
\end{displaymath}
Then there exists exactly one $x_\d^*\in [0,1]$ such that $f_{-\d}(x_\d^*) \geq f_{-\d}(x)$ for all $x\in [0,1]$.
Moreover, we have $x_\d^* \in (x^*-\d, x^*+\d)$ and
 $x_\d^*$ is the only element $x\in [\d,1-\d]$ that satisfies $f(x-\d) = f(x+\d)$.
Finally, we have
\begin{displaymath}
 f_{-\d}(x)
=
\begin{cases}
f(x-\d)  & \mbox{ if } x\in [\d, x_\d^*]\\
f(x+\d)  & \mbox{ if } x\in [x_\d^*, 1-\d]\, .
\end{cases}
\end{displaymath}
\end{lemma}

\begin{proof}[Proof of Lemma \ref{translate-f}]
We first show
 that there is an $x_0\in (x^*-\d, x^*+\d)$
such that $f(x_0-\d) = f(x_0+\d)$. To this end, we observe
 $g: [x^*-\d, x^*+\d]\to \R$ defined by $g:= f(\mycdot-\d) - f(\mycdot+\d)$ is continuous, and since
$g(x^*-\d) = f(x^*-2\d) - f(x^*) < 0$ and $g(x^*+\d) = f(x^*) - f(x^*+2\d) > 0$, we find an $x_0\in (x^*-\d, x^*+\d)$
such that $g(x_0) = 0$ by the intermediate value theorem.

Let us now show that $f(x-\d) < f(x+\d)$ for all $x\in [\d, x_0]$ and
$f(x-\d) > f(x+\d)$ for all $x\in [x_0, 1-\d]$. Clearly, for
$x\in [\d,x^*-\d]$, the strict monotonicity of $f$
on $[0,x^*]$ yields
$f(x-\d) < f(x+\d)$. Moreover, for $x\in (x^*-\d, x_0)$,
we have $f(x-\d) < f(x_0-\d)=f(x_0+\d) < f(x+\d)$
since
$f(\mycdot-\d):[x^*-\d, x^*+\d]\to \R$ is strictly increasing, while
$f(\mycdot+\d):[x^*-\d, x^*+\d]\to \R$ is  strictly decreasing.
This shows the assertion for $x\in [\d, x_0]$, and the assertion for
$x\in [x_0, 1-\d]$ can be shown analogously.

Combining the two results  above, we find that there exists exactly one
$x_0\in [\d,1-\d]$ satisfying $f(x_0-\d) = f(x_0+\d)$, and for this $x_0$ we further know $x_0\in (x^*-\d, x^*+\d)$.
In addition, these results show
\begin{displaymath}
 f_{-\d}(x)
=
\begin{cases}
f(x-\d)  & \mbox{ if } x\in [\d, x_0]\\
f(x+\d)  & \mbox{ if } x\in [x_0, 1-\d]\, .
\end{cases}
\end{displaymath}

Let us now return to global maximizers of $f_{-\d}$. To this end, we first observe that
the existence of a global maximum of $f_{-\d}$ follows from the continuity of $f_{-\d}$ and the compactness of $[0,1]$.
Let us now fix an $x_\d\in [0,1]$ at which this global maximum is attained by $f_{-\d}$.
We first observe that $x_\d\in (\d, 1-\d)$. Indeed, if, e.g., we had $x_\d\geq 1-\d$,
we would obtain $f(1) = f_{-\d}(x_\d) \geq f_{-\d}(1-2\d) = f(1-3\d) \wedge f(1-\d) = f(1-\d) > f(1)$
using $1-3\d > x^*$, and $x_\d \leq \d$ would similarly lead to a contradiction.
We next show that we actually have $x_\d\in [x^*-\d, x^*+\d]$. To this end,  it suffices to show
\begin{equation}\label{translate-f-h1}
 x_\d \geq x^*-\d \qquad \qquad \Longleftrightarrow \qquad \qquad x_\d \leq x^*+\d\, .
\end{equation}
To show one implication, assume that $x_\d \geq x^*-\d$. Since $f_{-\d}$ attains its maximum at $x_\d$, we then obtain
\begin{align*}
 f(x_\d+\d)
\geq f(x_\d-\d) \wedge f(x_\d+\d)
 = f_{-\d}(x_\d)
\geq f_{-\d}(x^*+\d)
= f(x^*+2\d)\,.
\end{align*}
Now $x_\d+\d\leq x^*+2\d$ follows from the assumed $x_\d +\d\geq x^*$ and the strict monotonicity of $f$ on $[x^*,1]$.
Analogously, $x_\d \leq x^*+\d \Rightarrow x_\d \geq x^*-\d$ can be shown, and hence \eqref{translate-f-h1}
is indeed true.

Finally, we can prove the remaining assertion. To this end, we pick again an $x_\d$ at which $f_{-\d}$ attains its
maximum. Then we have already seen that $x_\d \in [x^*-\d, x^*+\d]$. Now observe that assuming $x_\d <x_0$
leads to $f(x_\d-\d) < f(x_0 - \d) = f(x_0+\d) < f(x_\d+\d)$
using $x_0, x_\d\in [x^*-\d, x^*+\d]$, which in turn
yields the contradiction
\begin{align*}
 f_{-\d}(x_\d) = f(x_\d\!-\!\d) \!\wedge\! f(x_\d \!+ \!\d) =  f(x_\d\!-\!\d)
< f(x_0 \!- \!\d) \!\wedge\! f(x_0\! +\! \d)
 = f_{-\d}(x_0)\, .
\end{align*}
Analogously, we find a contradiction assuming $x_\d >x_0$, and hence we have $x_\d = x_0$.
Consequently, $x_\d$ is unique and solves $f(x-\d) = f(x+\d)$.
\end{proof}

\begin{proof}[Proof of Example \ref{psis-in-r2}]
We first note that the existence and uniqueness of the local extrema is guaranteed by Lemma \ref{translate-f}.
In addition, this lemma actually shows $x_{\d,-}^+\in (x^+_--\d, x^+_-+\d)$,
$x_{\d,-}^-\in (x^-_--\d, x^-_-+\d)$, $x_{\d,+}^+\in (x^+_+-\d, x^+_++\d)$, and
$x_{\d,+}^-\in (x^-_+-\d, x^-_++\d)$.
Moreover,
  we  have
\begin{displaymath}
 \psis_A(\d) = \sup_{z\in A} d(z,A\mde) \leq \sup_{z\in A} d(z,A\omde)
\end{displaymath}
by $A\mde \subset A\omde$.
We will thus estimate $d(z, A\omde)$ for $z:= (x,y)\in A$.

We begin with the case $x\in [-1, x_{\d,-1}]$. For later purposes,   note that
the definition of $A$   yields $x\geq x_{0,-1}$.
By the monotonicity of $f^\pm$ on $[-1, -0.8+\d]$ we further know
$f^\pm_\d(x+\d) = f^\pm(x)$.
We write $x':= x_{\d,-1} + \d$ and
\begin{displaymath}
 y':=
\begin{cases}
 f^-( x_{\d,-1})+\d & \mbox{ if } y\leq f^-( x_{\d,-1} ) + \d \\
y & \mbox { if } y\in [ f^-( x_{\d,-1} ) + \d,  f^+( x_{\d,-1} ) - \d]\\
 f^+( x_{\d,-1})-\d & \mbox{ if } y\geq f^+( x_{\d,-1} ) - \d\, .
\end{cases}
\end{displaymath}
If $y\leq f^-( x_{\d,-1} ) + \d$, we then obtain $y\leq y'$ and $y'= f^-( x_{\d,-1})+\d \leq f^-(x)+\d \leq y+\d$, that is
$|y-y'|\leq \d$, and
it is easy to check that the same is true in the two other cases. Consequently, we have
$\inorm{(x,y) - (x',y')} =x_{\d,-1} + \d - x$, and our construction further ensures
\begin{displaymath}
 y' \in [ f^-( x_{\d,-1} ) + \d,  f^+( x_{\d,-1} ) - \d] =  [ f_{-\d}^-(x') + \d,  f_{-\d}^+(x') - \d]\, .
\end{displaymath}
By Example \ref{Aomde} we conclude
 $(x',y')\in A\omde$, and  from this we easily find
\begin{equation}\label{psis-in-r2-h2}
 d(z, A\omde) \leq  \d + x_{\d,-1} - x \leq \d + x_{\d,-1} - x_{0,-1}\, .
\end{equation}

To show that \eqref{psis-in-r2-h2} is also true in the case $x \in  [x_{\d,-1}, -0.8+\d]$, we first observe that
the monotonicity of $f^\pm$  on $[-1, -0.8+2\d]$   yields
\begin{displaymath}
f^+(x) - f^-(x) \geq  f^+( x_{\d,-1}) - f^-( x_{\d,-1}) \geq 2\d\, ,
\end{displaymath}
and consequently, we can define
\begin{displaymath}
 y':=
\begin{cases}
 f^-(x)+\d & \mbox{ if } y\leq f^-(x) + \d \\
y & \mbox { if } y\in [ f^-(x) + \d,  f^+(x) - \d]\\
 f^+(x)-\d & \mbox{ if } y\geq f^+(x) - \d\, .
\end{cases}
\end{displaymath}
If $y\leq f^-(x) + \d$ we then obtain $y\leq y'$ and $y' = f^-(x) + \d \leq y+\d$, that is $|y-y'|\leq \d$, and
again it is easy to check that the same is true in the two other cases. Writing $x':= x+\d$, we thus have
$\inorm{(x,y) - (x',y')} =\d$. Moreover, the construction together with $f^\pm_\d(x+\d) = f^\pm(x)$ ensures
\begin{displaymath}
 y' \in  [ f^-(x) + \d,  f^+(x) - \d] =  [ f_{-\d}^-(x') + \d,  f_{-\d}^+(x') - \d] \, ,
\end{displaymath}
and hence we find $(x',y')\in A\omde$ by Example \ref{Aomde}. Thus, we
have shown
$d(z, A\omde) \leq \d \leq \d + x_{\d,-1} - x_{0,-1}$,
i.e.~\eqref{psis-in-r2-h2} is true for all $x\in [-1, -0.8+\d]$.

Now consider the case $x\in [-0.8+\d, -0.2-\d]$. Here, we will focus on the sub-case $y\geq 0$, since the subcase
$y\leq 0$ can be treated analogously. For later purposes, note that we have $f^-(x\pm\d) \leq -2 \d$.
 Now, if $x\in [-0.8+\d, x^+_{\d,-}-\d]$,
  we set $x' := x+\d$ and $y':= y \wedge (f^+(x) -\d)$. This gives $y' \leq y$ and $y-\d\leq f^+(x) -\d \leq y'$,
and hence we again have $\inorm{(x,y) - (x',y')} =\d$. Moreover, our constructions together with Lemma \ref{translate-f}
ensures
\begin{displaymath}
 y' \in [-\d, f^+(x) -\d] = [ - \d,  f_{-\d}^+(x') - \d] \subset [ f_{-\d}^-(x') + \d,  f_{-\d}^+(x') - \d]\, ,
\end{displaymath}
that is $(x',y')\in A\omde$, and hence \eqref{psis-in-r2-h2} is true in this case, too.
The next case, we consider, is $x\in [x^+_{\d,-}-\d, x^+_{\d,-}+\d]$. In this case we set $x':= x^+_{\d,-}$
and $y' := y \wedge (f_{-\d}^+(x^+_{\d,-})-\d)$. This implies
\begin{displaymath}
 y' \in [-\d, f_{-\d}^+(x^+_{\d,-})-\d] \subset [ f_{-\d}^-(x')+\d,   f_{-\d}^+(x')-\d]\, ,
\end{displaymath}
and hence $(x',y')\in A\omde$. We further   have $|x-x'| \leq \d$
and, if $y \leq f_{-\d}^+(x^+_{\d,-})-\d$, we also have $|y-y'| = 0$. Conversely, if
$y \geq f_{-\d}^+(x^+_{\d,-})-\d$, we find
\begin{displaymath}
 y  \leq f^+(x)\leq f^+(x^+_-)  = f^+(x^+_-) - (f_{-\d}^+(x^+_{\d,-})-\d) + y' \, ,
\end{displaymath}
that is $|y-y'| \leq \d + f^+(x^+_-) - f_{-\d}^+(x^+_{\d,-})$. Combining the latter two  cases, we therefore obtain
$\inorm{(x,y) - (x',y')} \leq \d + f^+(x^+_-) - f_{-\d}^+(x^+_{\d,-})$, that is
$d(z, A\omde)   \leq \d +  f^+(x^+_-) - f_{-\d}^+(x^+_{\d,-})$.
Since all remaining cases can be treated analogously, the proof of the general estimate of $\psis_A(\d)$ is finished.

Now consider the additional assumptions of $f^\pm$. For example, assume
\begin{displaymath}
 |f^+(x^+_-) - f^+(x)|  \leq c |x^+_- -x|^\g
\end{displaymath}
for all $ x\in [x^+_--0.1,x^+_-+0.1]$.
         Lemma \ref{translate-f} shows  $x^+_{\d,-} \in (x^+_--\d, x^+_-+\d)$.
        Without loss of generality, we  assume  $x^+_{\d,-}\in [x^+_-, x^+_-+\d)$. Using Lemma \ref{translate-f}
        and $x^+_{\d,-}-\d \in [x^+_--\d, x^+_-) \subset [x^+_--0.1,x^+_-+0.1]$, we then
        obtain
\begin{displaymath}
 \bigl|f^+(x^+_-) - f_{-\d}^+(x^+_{\d,-})\bigr|
        = \bigl|f^+(x^+_-) - f^+(x^+_{\d,-} - \d)\bigr|
        \leq c \bigl| x^+_- - x^+_{\d,-} + \d \bigr|^\g
        \leq c \d^\g\, .
\end{displaymath}
Now   assume that, for e.g.~$i:= -1$, we have  $\d_0>0$. For  $\d\in (0,\d_0]$ we then find
$f^+(x_{0,-1}) - f^-(x_{0,-1}) \geq 2\d$, and thus $x_{0, -1} = x_{\d, -1} = -1$.
Conversely, let $\d\in (\d_0, 0.1]$. Then we have $f^+(x_{0,-1}) - f^-(x_{0,-1}) < 2\d$ and a simple
application of the intermediate value theorem thus yields $f^+(x_{\d,-1}) - f^-(x_{\d,-1}) = 2\d$.
%
Using the additional assumption on $f^\pm$
around the point $x_{0, -1}$, we then find
\begin{align*}
 2 c^{-1/\g} |x_{\d,-1} - x_{0,-1}|^{1/\g}
 \leq | f^-(x_{\d,-1})| \!+\! |  f^+(x_{\d,-1})|
&= f^+(x_{\d,-1}) - f^-(x_{\d,-1}) \\
&= 2\d\, ,
\end{align*}
that is $|x_{\d,-1} - x_{0,-1}| \leq c\d^\g$.
\end{proof}

\section{Continuous Densities}\label{suppB:cons-dens}

In this section we present  a class of continuous densities on $\R^2$ that meet
the assumptions made in the paper.
The first example, which represents the main result of this supplement, shows that
many continuous distributions satisfy our thickness assumption.

\begin{example}\label{cont-h-in-r2}
 Let $X:= [-1,1]\times [-2,2]$ be equipped with the metric defined by the supremums norm.
Moreover, let $P$ be a Lebesgue absolutely continuous distribution that has a continuous density $h$.
Furthermore, assume that there exists a   $\rss>0$, such that, for all  $\r\in (0,\rss]$, the
level set $M_\r$   is
of the form considered in Example \ref{psis-in-r2}.
 In addition, we assume that there is a constant $K\in (0,1)$ such that
\begin{equation}\label{cont-h-in-r2-h0}
 \bigl|h(x,y) - \rs - x^2 + y^2 \bigr| \leq K (x^2+y^2   )
\end{equation}
for some $\rs\in [0,\rss)$ and  all $(x,y) \in \{h>0\} \cap  \bigl( [-0.2, 0.2] \times (-1.1, 1.1) \bigr)$.
Moreover,  assume that $h$ is continuously differentiable
on the sets
\begin{align*}
 A_1 := \{h>0\} &\cap   \Bigl(\bigl((-0.7, -0.3) \cup (0.3, 0.7) \bigr)\times \bigl((-1.1, -0.2) \cup (0.2, 1.1)\bigr)\Bigl) \\
 A_2 := \{h>0\} &\cap   \Bigl(\bigl((-1, -0.8) \cup (0.8, 1) \bigr)\times \bigl((-1.1, 0) \cup (0.2, 1.1)\bigr)\Bigl) \\
A_3 := \{h>0\} &\cap \biggl\{\!(x,y)\in X: x\in (-0.2, 0) \cup (0, 0.2) \mbox{ and } |y| \!<\! \sqrt{\frac{1\!+\!K}{1\!-\!K}} |x|\biggr\}
\end{align*}
with $h_y:= \frac {\partial h}{\partial y} \neq  0$ on $A_1$ and $h_x:= \frac {\partial h}{\partial x} \neq  0$ on $A_2\cup A_3$.
Finally, assume that there is a constant $C>0$ such that
$|{h_x}| \leq C|{h_y}|$ on $A_1$ and $|{h_y}| \leq C|{h_x}|$ on $A_2\cup A_3$.
Then  $P$
has thick levels of order  $\g=1$ with $\dthick = 0.1$ and
\begin{displaymath}
 \cthick = 1 + \max\biggl\{C, \sqrt{\frac{1+K}{1-K}}\biggr\}\, .
\end{displaymath}
Moreover, $P$
can be  clustered between
$\rs$ and $\rss$ and we have
\begin{equation}\label{cont-h-in-r2-ts}
   \frac 2 {\sqrt{1-K}} \sqrt \e \Leq \t^*_{M_{\rs+\e}} \Leq  \frac 2 {\sqrt{1+K}} \sqrt \e\, , \qquad \qquad \e\in (0,\rss-\r].
\end{equation}
\end{example}

\begin{proof}[Proof of Example \ref{cont-h-in-r2}]
 Since we consider the Lebesgue measure on $X$, we have $M_0 = X$. Moreover, we have $X\mde = X$
since we consider the operation in $X$, and from this, we immediately see
$\psis_X(\d) = 0$ for all $\d>0$. Consequently, there is nothing to prove for $\r=0$.

 Let us now fix some $\r\in (0,\rss]$. Moreover, let $f^\pm:[-1,1]\to [-1,1]$ be the two functions
satisfying the assumptions of Example \ref{psis-in-r2} and
\begin{displaymath}
 M_\r = \bigl\{  (x,y)\in X: x\in [-1, 1] \mbox{ and } f^-(x)\leq y\leq f^+(x)\bigr\}\, .
\end{displaymath}
We pick an $(x,y)\in M_\r$ with $y= f^+(x)$ or  $y= f^-(x)$. Then we find $(x,y)\in \partial M_\r$, and thus
we have $h(x,y) = \r$ by Lemma \ref{Mr-include-both-new}, that is $h(x, f^\pm(x)) = \r$.

Our first goal is to verify \eqref{psis-in-r2-xpm}. To this end, we solely focus without
loss of generality to the case $x^+_+$ and $f^+$, since the other cases can be treated analogously.
 Let us fix an
$x\in [x^+_+-0.1,x^+_++0.1]$. Then we have $x\in (0.3, 0.7)$ and
thus $f^+(x) \in (0.2, 1.1)$ by \eqref{psis-in-r2-thick}. Consequently, $h$ is continuously differentiable
in $(x, f^+(x))$. By the implicit function theorem and the previously shown $h(x', f^+(x')) = \r$
for all $x'\in (0.3, 0.7)$
we then conclude that $f^+$ is continuously differentiable at $x$ and
\begin{equation}\label{cont-h-in-r2-h1}
 (f^+(x))'
= -\biggl( \frac{\partial h}{\partial y} \bigl(x, f^+(x) \bigr)   \biggr)^{-1} \cdot  \frac{\partial h}{\partial x} \bigl(x, f^+(x) \bigr)
= \frac {h_x(x, f^+(x))}{h_y(x, f^+(x))}\, .
\end{equation}
Using $|{h_x}| \leq C|{h_y}|$ on $A_1$, we thus find $|(f^+(x))'|\leq C$, and hence $f^+$ is Lipschitz
continuous on $(0.3, 0.7)$ with Lipschitz constant smaller than or equal to $C$. This implies
 \eqref{psis-in-r2-xpm} with constant $C$ and exponent $\g=1$.

Now consider the endpoints $x_{0, \pm 1}$, where again it suffices to consider
 one case, say $x_{0,-1}$, due to symmetry.
Let us write $2\d_0:=f^+(x_{0,-1}) - f^-(x_{0,-1})$. Then, if $\d_0\geq 0.1$, we have
$|x_{\d,-1} - x_{0,-1}| = 0$ for all $\d\in (0,0.1]$ by  Example \ref{psis-in-r2},
and hence it suffice to show \eqref{psis-in-r2-0pm} in the case $\d_0< 0.1$.
Observing that it actually suffices to show \eqref{psis-in-r2-0pm} for all $x\in (x_{0,-1}, -0.8)$
by continuity, we begin by fixing such an $x$. By monotonicity we then have
$0 < f^+(x) < f^+(0.8)< 1.1$, and
hence $h$ is continuously differentiable
at $(x, f^+(x))$. The implicit function theorem and the previously shown $h(x', f^+(x')) = \r$
for all $x'\in (x_{0,-1}, -0.8)$,
then shows that $f^+$ is continuously differentiable at $x$ and \eqref{cont-h-in-r2-h1} holds.
Using $|{h_y}| \leq C|{h_x}|$ on $A_2$, we then find $|(f^+(x))'|\geq 1/C$, and the fundamental theorem
of calculus thus yields
\begin{displaymath}
 \bigl|f^+(x') -  f^+(x) \bigr|= \biggl| \int_x^{x'}  (f^+(t))' dt \biggr| \geq C^{-1}  |x'-x|
\end{displaymath}
for all $x,x'\in (x_{0,-1}, -0.8)$. Now, letting $x'\to x_{0,-1}$, we obtain
\begin{displaymath}
 |f^+(x)| \geq f^+(x) -  f^+(x_{0,-1}) = \bigl|f^+(x) -  f^+(x_{0,-1}) \bigr| \geq  C^{-1}  |x_{0,-1}-x|
\end{displaymath}
for all $x\in (x_{0,-1}, -0.8)$, i.e.~\eqref{psis-in-r2-0pm} holds with constant  $C$ and   $\g=1$.

Finally, let us consider the points $x_{0, \pm 0}$, where yet another time, we only focus on one case,
say $x_{0,+0}$.
For $x\in [x_{0,+0},0.2]$, we then have 
\begin{equation}\label{cont-h-in-r2-h3}
 \r = h(x, f^+(x)) \leq \rs + (1+K)x^2 + (K-1) (f^+(x)) ^2\, ,
\end{equation}
that is $(f^+(x))^2 \leq \frac{\rs-\r}{1-K} + \frac{1+K}{1-K} x^2$. Analogously, we can find a lower bound on
$(f^+(x))^2$, so that we end up having
\begin{equation}\label{cont-h-in-r2-h2}
 (f^+(x))^2 \in \biggl[\frac{\rs-\r}{1+K} + \frac{1-K}{1+K} x^2, \frac{\rs-\r}{1-K} + \frac{1+K}{1-K} x^2 \biggr],
\end{equation}
and an analogue result holds for $(f^-(x))^2$.
Again, our goal is to show an analogue of \eqref{psis-in-r2-0pm}. To this end, we first consider the
case $\r\in (0, \rs]$. By \eqref{cont-h-in-r2-h0}, we then know that $h(0,0) = \rs\geq \r$, and hence $f^+(0)\geq 0$.
Analogously, we find  $f^-(0)\leq 0$, which together implies $x_{0,+0}=0$.
Furthermore, for $x\in [x_{0,+0},0.2]$, \eqref{cont-h-in-r2-h2} gives
\begin{displaymath}
 f^+(x) \geq \sqrt{ \frac{\rs-\r}{1+K} + \frac{1-K}{1+K}x^2}
\geq \sqrt{\frac{1-K}{1+K}}|x|
=\sqrt{ \frac{1-K}{1+K}}\, |x_{0,+0}-x| \, ,
\end{displaymath}
%
%
that is \eqref{psis-in-r2-0pm} holds with constant $\sqrt{ \frac{1+K}{1-K}}$ and exponent $\g=1$.
Let us now consider the case $\r\in (\rs, \rss]$. For $x\in (x_{0,+0},0.2)$, \eqref{cont-h-in-r2-h2} then yields
\begin{displaymath}
 f^+(x) \leq \sqrt{ \frac{\rs-\r}{1-K} + \frac{1+K}{1-K}x^2}
< \sqrt{\frac{1+K}{1-K}}|x| \, ,
\end{displaymath}
and thus we find $(x, f^+(x)) \in A_3$. Consequently, $h$ is continuously differentiable at
$(x, f^+(x))$, and \eqref{cont-h-in-r2-h1} holds. As for $x_{0, -1}$, we can then show that
\eqref{psis-in-r2-0pm} holds with constant $C$ and exponent $\g=1$.

In order to show that $P$
can be  clustered between the  levels
$\rs$ and $\rss$, we first note that the assumed continuity of $h$
guarantees that $P$ is normal by Lemma \ref{regular-new}.
%
%
Let us now fix a $\r\in (\rs, \rss]$. Since from
\eqref{cont-h-in-r2-h0} we infer that $h(0,0) = \rs$, we then obtain
$(0,0)\not \in M_\r$. The latter implies $x_{0,-0}<0<x_{0,+0}$, where $x_{0,-0}$ and $x_{0,+0}$
are the points defined in Example \ref{psis-in-r2} for the   set $M_\r$. By Example \ref{Aomde}
we then see that $\ca C(M_\r)|= 2$. Analogously, for $\r\in [0,\rs]$, the equality
$h(0,0) = \rs$ implies $x_{0,-0}=0=x_{0,+0}$, which shows $\ca C(M_\r)|= 1$.
Finally, the bijectivity of  $\z:\ca C (M_{\r^{**}})\to \ca C(M_\r)$ follows from the
form of the connected components described in Example \ref{Aomde}.

Let us finally prove \eqref{cont-h-in-r2-ts}. To this end, we fix an $\e\in (0,\rss-\r]$ and define $\r:= \rs+\e$.
Then we have already observed that  $x_{0,-0}<0<x_{0,+0}$, and hence $f^\pm(x_{0,\pm 0}) = 0$.
For $x:= x_{0,+0}$ we then obtain
\begin{displaymath}
   \r = h(x, f^+(x)) \leq \rs + (1+K) x^2
\end{displaymath}
by \eqref{cont-h-in-r2-h3}, and applying some simple transformations we thus find
$x_{0,+0} = x \geq \sqrt{\frac{\r-\rs}{1+K}} = \sqrt{\frac{\e}{1+K}}$. For $x:= x_{0,+0}$ we further have
\begin{displaymath}
   \r = h(x, f^+(x)) \geq  \rs + (1-K) x^2 \, ,
\end{displaymath}
and thus $x_{0,+0} \leq \sqrt{\frac{\e}{1-K}}$. Since analogous estimates can be derived for $x_{0,-0}$, the
formula $\t^*_{M_{\rs+\e}} = x_{0, +0} - x_{0, -0}$ found in Example \ref{Aomde}   gives \eqref{cont-h-in-r2-ts}.
\end{proof}

The last example of this appendix shows that the distributions from the previous example have
a smooth boundary.

\begin{example}\label{cont-h-in-r2-smooth}
   Let $X$ and $P$ be as in Example \ref{cont-h-in-r2}. Then
   the clusters have
    an $\a$-smooth boundary for  $\a=1$ and
    \begin{displaymath}
       \cbound = 8 \Biggr(10+C+ \sqrt{\frac{1+K}{1-K}}\Biggr)\, .
    \end{displaymath}
\end{example}

\begin{proof}[Proof of Example \ref{cont-h-in-r2-smooth}]
   Let us first consider the case $0<\d\leq 0.1$.
   To this end, we fix a $\r\in (\rs, \rss]$. Without loss of generality, we only consider the connected component $A$
   with $x<0$ for all $(x,y)\in A$.
  We know that $A\pdeh\setminus A\mdeh \subset A\opde\setminus A\omde$ and the latter
  two sets have been calculated in Example \ref{Aomde}. In the following, we will only estimate
  $\lb^2(\{(x,y): y\geq 0  \}\cap A\opde\setminus A\omde)$, the case $y\leq 0$ can be treated analogously.
  Our first intermediate result towards the desired estimate is
  \begin{align*}
     \lb^2\bigl([-1 \vee (x_{0,-1}-\d), x_{\d,-1}]\times [0,2]   \cap A\opde\setminus A\omde\bigr)
     &\leq 2 |(x_{0,-1}-\d) - x_{\d,-1}| \\
     & \leq 2\d + 2 |x_{0,-1} - x_{\d,-1}| \\
     & \leq 2(1+C)\d\, ,
  \end{align*}
  where in the last step we used that the proof of Example \ref{cont-h-in-r2} showed \eqref{psis-in-r2-0pm}
  for $c=C$ and $\g=1$.
  Moreover, we have
    \begin{align*}
     \lb^2\bigl([x_{\d,-1}, x_-^+-\d]\times [0,2]   \cap A\opde\setminus A\omde\bigr)
     & = \int\limits_{x_{\d,-1}}^{x_-^+-\d}\!\! f^+(x\!+\!\d) - f^+(x\!-\!\d) +2\d \, dx\\
     & \leq 2\d + \int_{x_-^+-\d}^{x_-^++\d} f(x) \, dx \\
     & \leq 4 \d
  \end{align*}
  and analogously we obtain $\lb^2\bigl([x_-^++\d, x_{\d,-0}]\times [0,2]   \cap A\opde\setminus A\omde\bigr) \leq 4\d$.
  In addition, we easily find $\lb^2\bigl([x_-^+-\d, x_-^++\d]\times [0,2]   \cap A\opde\setminus A\omde\bigr) \leq 4\d$ and
  finally, we have
    \begin{align*}
    \lb^2\bigl([x_{\d,-0}, 0 \wedge (x_{0,-0}+\d)]\times [0,2]   \cap A\opde\setminus A\omde\bigr)
     &\leq 2 \bigl|x_{\d,-0} - x_{0,-0}-\d \bigr| \\
     &\leq 2\d + 2 \sqrt{\frac{1+K}{1-K}} \d\, ,
  \end{align*}
  where we used that the proof of Example \ref{cont-h-in-r2} showed \eqref{psis-in-r2-0pm}
  for $c=\sqrt{\frac{1+K}{1-K}}$ and $\g=1$. Combining all these estimates we obtain
  \begin{align*}
      \lb^2\bigl([-1, 0]\times [0,2]   \cap A\opde\setminus A\omde\bigr) \leq 4 \Biggr(6+C+ \sqrt{\frac{1+K}{1-K}}\Biggr) \d
  \end{align*}
  for all $\d\in (0, 0.05]$. Moreover, for $\d\in [0.05, 1]$ we easily obtain
  \begin{displaymath}
      \lb^2\bigl([-1, 0]\times [0,2]   \cap A\opde\setminus A\omde\bigr) \leq 2 \leq 40 \d\, .
  \end{displaymath}
  Combining both estimates and adding the case $y\leq 0$, we then obtain the assertion.
\end{proof}



%
